\documentclass{amsart}[11pt]
\usepackage{amsmath}
\usepackage{amssymb}
\usepackage{mathrsfs}
\usepackage{amsfonts}
\usepackage{graphicx}
\usepackage{texdraw}
\usepackage{graphpap}
\usepackage{enumitem}
\usepackage{chngcntr}
\usepackage{wasysym}
\usepackage{color}
\usepackage{todonotes}

\usepackage{comment}

\usepackage{tikz}
\usetikzlibrary{arrows}
\usetikzlibrary{decorations.pathmorphing}
\usetikzlibrary{decorations.markings}
\usetikzlibrary{patterns}
\usetikzlibrary{automata}
\usetikzlibrary{positioning}
\usepackage{tikz-cd}
\tikzset{->-/.style={decoration={
				markings,
				mark=at position #1 with {\arrow{latex}}},postaction={decorate}}}
	
	\tikzset{-<-/.style={decoration={
				markings,
				mark=at position #1 with {\arrowreversed{latex}}},postaction={decorate}}}

\usetikzlibrary{shapes.misc}\tikzset{cross/.style={cross out, draw,
         minimum size=2*(#1-\pgflinewidth),
         inner sep=0pt, outer sep=0pt}}

\usepackage{pgfplots}
	
\allowdisplaybreaks

\usepackage[colorlinks=true]{hyperref}
\hypersetup{urlcolor=blue, citecolor=red, linkcolor=blue}

\usetikzlibrary{decorations.markings}
\usetikzlibrary{spy}

\newcommand{\boundellipse}[3]
{(#1) ellipse (#2 and #3)}

\theoremstyle{plain}
\newtheorem{thm}{Theorem}[section]

\newtheorem{cor}[thm]{Corollary}
\newtheorem{lem}[thm]{Lemma}

\newtheorem{prop}[thm]{Proposition}
\newtheorem{defn}[thm]{Definition}

\theoremstyle{remark}

\newtheorem*{rem}{Remark}

\newtheorem*{ex}{Example}

\numberwithin{equation}{section}
\newcommand{\nlap}{\partial_{\textrm{n}} \Delta}
\newcommand{\SH}{\mathrm{SH}}
\newcommand{\lfun}{\lambda}

\newcommand{\normal}{{\tt{n}}}

\newcommand{\dNpara}{u}

\newcommand{\dn}{\partial_{\textrm{n}}}

\newcommand{\calA}{{\mathcal A}}

\newcommand{\calH}{{\mathscr H}}

\newcommand{\calN}{{\mathcal N}}
\newcommand{\calS}{{\mathsf S}}

\newcommand{\calF}{{\mathcal F}}
\newcommand{\fluct}{\operatorname{fluct}}

\newcommand{\calW}{{\mathscr W}}

\newcommand{\pol}{{p}}

\newcommand{\Int}{\operatorname{Int}}

\newcommand{\trace}{\operatorname{trace}}
\newcommand{\erfc}{\operatorname{erfc}}

\newcommand{\R}{{\mathbb R}}

\newcommand{\C}{{\mathbb C}}

\newcommand{\eps}{{\varepsilon}}

\newcommand{\re}{\operatorname{Re}}
\newcommand{\im}{\operatorname{Im}}

\def\ROne{\calA(r_1,\delta_n)}
\def\RTwo{\calA(r_2,\delta_n)}

\renewcommand{\d}{{\partial}}
\newcommand{\dbar}{\bar{\partial}}
\newcommand{\1}{\mathbf{1}}

\newcommand{\dist}{\operatorname{dist}}
\newcommand{\supp}{\operatorname{supp}}

\newcommand{\radlap}{\partial_r \Delta}

\newcommand{\bigO}{\mathcal{O}}

\begin{document}

\title[Fluctuations through a spectral gap]{The two-dimensional Coulomb gas: fluctuations through a spectral gap}

\author{Yacin Ameur}
\address{Yacin Ameur\\
Department of Mathematics\\
Lund University\\
22100 Lund, Sweden}
\email{ Yacin.Ameur@math.lu.se}
\author{Christophe Charlier}
\address{Christophe Charlier\\
Department of Mathematics\\
Lund University\\
22100 Lund, Sweden}
\email{Christophe.Charlier@math.lu.se}

\author{Joakim Cronvall}
\address{Joakim Cronvall\\
Department of Mathematics\\
Lund University\\
22100 Lund, Sweden}
\email{Joakim.Cronvall@math.lu.se}



\begin{abstract} We study a class of radially symmetric Coulomb gas ensembles at inverse temperature $\beta=2$, for which the droplet consists of a number of concentric annuli, having at least one bounded ``gap'' $G$, i.e., a connected component of the complement of the droplet, which disconnects the droplet. Let $n$ be the total number of particles.
Among other things, we deduce fine asymptotics as $n \to \infty$ for the edge density and the correlation kernel near the gap, as well as for the cumulant generating function of fluctuations of smooth linear statistics. We typically find an oscillatory behaviour in the distribution of particles which fall near the edge of the gap.
These oscillations are given explicitly in terms of a discrete Gaussian distribution, weighted Szeg\H{o} kernels, and the Jacobi theta function, which depend on the parameter $n$.
\end{abstract}

\maketitle

\section{Introduction}
\subsection{Coulomb droplets with spectral gaps} In recent years, much work has been done relating to statistical properties of two-dimensional Coulomb gas ensembles near the edge of a connected droplet. Typically these works have focused on properties near the ``outer boundary'', i.e., the boundary of the unbounded component $U$ of the complement of the droplet. See for example \cite{ADM,AC,AM,AKM,BBNY2,BF,BF2022,BS2021,BY2022,ES,F,Fo,FJ,HW,LR,LSe,RV}.

In the present work, we study a class of radially symmetric ensembles (at inverse temperature $\beta=2$), for which the droplet $S$ consists of a finite number of concentric annuli, having at least one bounded ``spectral gap'' $G$, i.e., a
component of the complement $\C\setminus S$ which disconnects $S$.
A schematic picture, of a typical droplet under study, is given in
Figure \ref{pic1}.

Similar to what goes on near the outer boundary, we shall find that the particles that fall near the edge $\d G$ of a spectral gap tend to form a strongly correlated ``field'', but with an additional uncertainty built into it, since there are two disjoint boundary components near which each individual particle could fall. Among other things, we shall quantify the additional uncertainty in terms of a discrete Gaussian distribution, which varies (or ``oscillates'') with the total number of particles.
In a sense, we thus obtain new two-dimensional counterparts to results in the multi-cut regime found in e.g.~\cite{DKMVZ1999,BG2,CFWW}.

It is worth remarking that we here exclusively study droplets with ordinary ``soft edges''. This means that the particle density varies continuously in a neighbourhood of the boundary of $S$, with a quick but smooth (error-function type) drop-off in the direction of the complement. A very different but yet somewhat parallel setting, with ``hard edges'' where the density vanishes in a highly discontinuous manner, is studied in \cite{ACCL,C}.

\begin{figure}[h!]
\begin{center}
\begin{tikzpicture}
\node at (0,0) {\includegraphics[width=5cm]{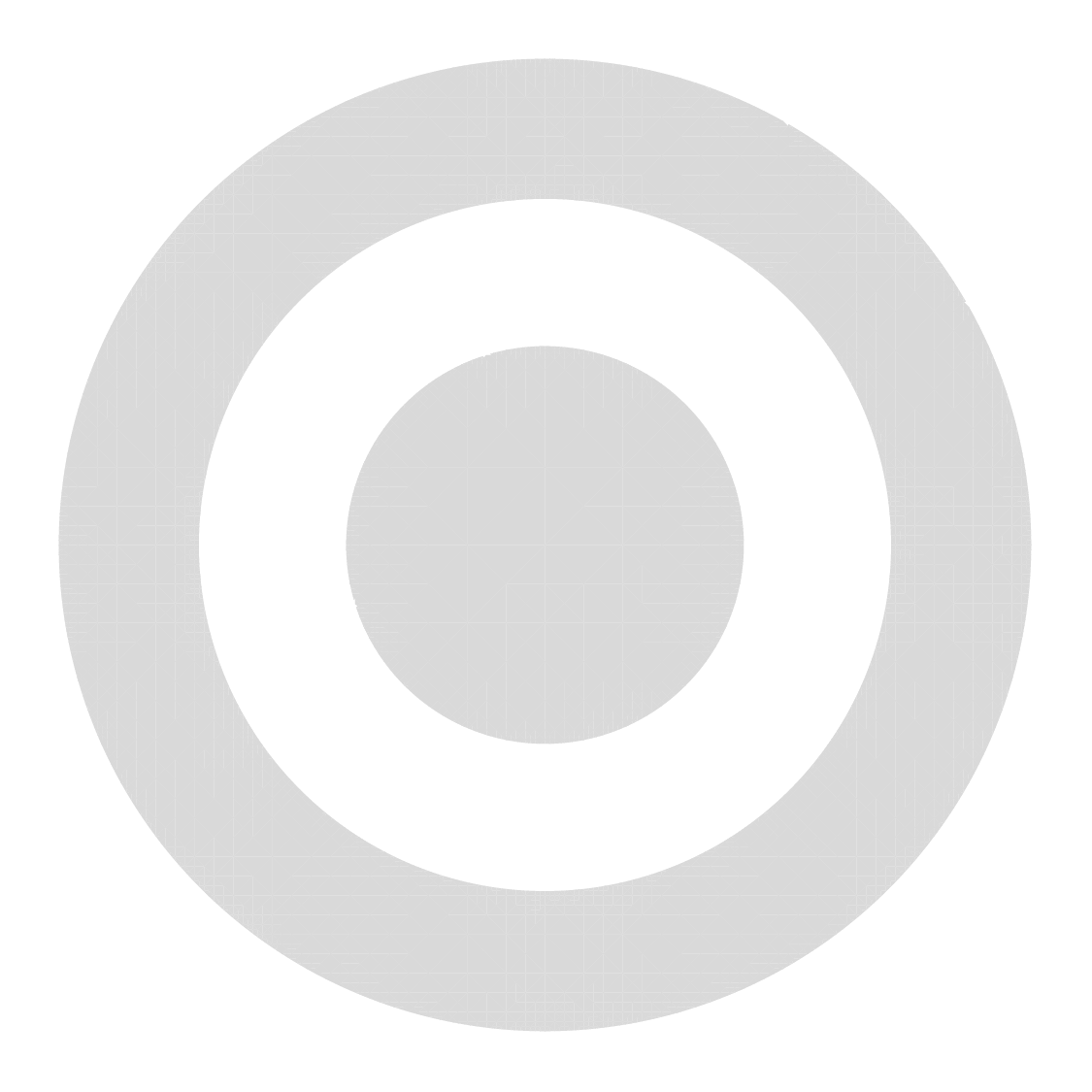}};
\node at (0,0) {\footnotesize $S$};
\node at (45:1.22) {\footnotesize $G$};
\node at (45:1.85) {\footnotesize $S$};
\node at (45:2.5) {\footnotesize $U$};
\draw[fill] (2.23,0) circle (0.03cm);
\node at (2.28,0.2) {\footnotesize $b_{N}$};
\draw[fill] (1.58,0) circle (0.03cm);
\node at (1.58,0.2) {\footnotesize $r_{2}$};
\draw[fill] (0.91,0) circle (0.03cm);
\node at (0.91,0.2) {\footnotesize $r_{1}$};
\end{tikzpicture}
\end{center}
\caption{The gap $G=\{r_1<|z|<r_2\}$ disconnects the droplet $S$. The domain $U=\{|z|>b_N\}\cup\{\infty\}$ is the component of $\hat{\C}\setminus S$ containing $\infty$.}
\label{pic1}
\end{figure}
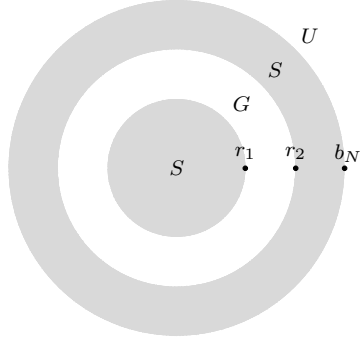

\subsubsection{Some potential theoretic preliminaries} \label{backgr} We begin by recalling some general principles of weighted potential theory, with respect to an arbitrary admissible (not necessarily rotationally symmetric) external potential, i.e.,
a function $$Q:\C\to\R\cup\{+\infty\}$$ whose properties are specified below.

Given a compactly supported unit (positive) Borel measure on $\C$ (i.e.
$\mu(\C)=1$) we define its weighted logarithmic energy by
$$I_Q[\mu]=\int_{\C^2}\log \frac 1 {|z-w|}\, d\mu(z)\, d\mu(w)+\mu(Q),$$
where we write $\mu(Q)=\int_\C Q\, d\mu$. If we think of $\mu$ as a blob of charge, the first term represents the self-interaction energy, and the second one gives the energy from interaction with the external potential.

Here and in what follows,
the potential $Q$ is assumed to be lower semicontinuous, finite on some set of positive capacity, and ``large'' near infinity in the sense that $Q(z)-2\log|z|\to\infty$ as $|z|\to\infty$. By standard results (cf.~ e.g.~\cite{ST}) there then exists a unique equilibrium measure $\sigma$ on $\C$ minimizing $I_Q$ over all compactly supported unit Borel measures on $\C$.

The support of $\sigma$ is termed the \textit{droplet} and denoted $S=S[Q]$. Assuming (as we will) that $Q$ is $C^2$-smooth in a neighbourhood of $S$, we have by Frostman's theorem (see \cite[Theorem II.1.3]{ST})
that
\begin{align}\label{def of eq measure}
d\sigma=\Delta Q\cdot\1_S\, dA,
\end{align}
where we use the conventions
$$\Delta=\d\dbar=\frac 1 4(\d_{xx}+\d_{yy}),\qquad dA=\frac 1\pi\, dxdy.$$
(Here and in what follows, $\d=\frac 1 2(\frac \d {\d x}-i\frac \d {\d y})$ and $\dbar=\frac 1 2(\frac \d {\d x}+i\frac \d {\d y})$ are the usual complex derivatives with respect to $z=x+iy$.)

Note that $Q$ is subharmonic on $S$ (since $\sigma$ is a measure).

Given a positive measure $\mu$, we write $U^\mu$ for the usual logarithmic potential
\begin{align}\label{Umu}
U^\mu(z)=\int_\C\log\frac 1 {|z-w|}\, d\mu(w).
\end{align}
It is known (see \cite[Theorem I.1.3]{ST}) that there exists a constant $\gamma=\gamma_Q$ (``Robin's constant'') such that the equilibrium measure satisfies
\begin{equation}\label{ob1}Q+2U^\sigma=\gamma\qquad \text{on } S\end{equation}
and
\begin{equation}\label{ob2}Q+2U^\sigma\ge \gamma\qquad \text{on } \C\setminus S.\end{equation}

We denote by
$$\check{Q}(z)=\gamma-2U^\sigma(z)$$
the so-called obstacle function, which is a subharmonic function \cite[Theorem 0.5.6]{ST} satisfying
$\check{Q}=Q$ on $S$, $\check{Q}\le Q$ on $\C$ (see also Figure \ref{fig1}) and
$$\check{Q}(z)=2\log|z|+O(1),\qquad \text{as}\qquad z\to\infty.$$

Moreover, $\check{Q}$ is harmonic on $\C\setminus S$ and globally $C^{1,1}$-smooth,
i.e.,
its gradient is Lipschitz continuous (this directly follows from \eqref{Umu}). In the sense of distributions,
$\Delta \check{Q}=\Delta Q\cdot \1_S.$

Many of the above facts are easy to understand by a Perron-family argument: let $\SH_1$ be the family of all subharmonic functions $s:\C\to \R$ such that $s\le Q$ everywhere on $\C$ and $s(z)\le 2\log|z|+O(1)$ as $z\to\infty$. By 
\cite[Theorem I.4.1]{ST} $\check{Q}(z)$ is the envelope
\begin{equation}\label{envo}
\check{Q}(z)=\sup\{s(z)\,;\, s\in\SH_1\}.
\end{equation}

We shall make a few further mild assumptions on our potential $Q$. First of all, we shall assume that the Laplacian $\Delta Q$ is \textit{strictly} positive on the boundary $\d S$. (The case where $\Delta Q$ vanishes on some subset of $\d S$ requires a different analysis and will not be considered in this work.) It is also convenient to assume that we have strict inequality $Q>\check{Q}$ in the complement $\C\setminus S$,
i.e., $S=S^*$ where $S^*=\{Q=\check{Q}\}$ is the contact set. We remark that for smooth potentials, it might happen that the complement $\C\setminus S$ has infinitely many components, but if $Q$ is real-analytic near $\d S$ there are only finitely many such components due to Sakai's theory in \cite{Sa}.

With these preliminaries out of the way, we specialize to a class of radial potentials such that the droplet has one or several ``gaps''.

\begin{figure}[ht]
\includegraphics[width=8cm]{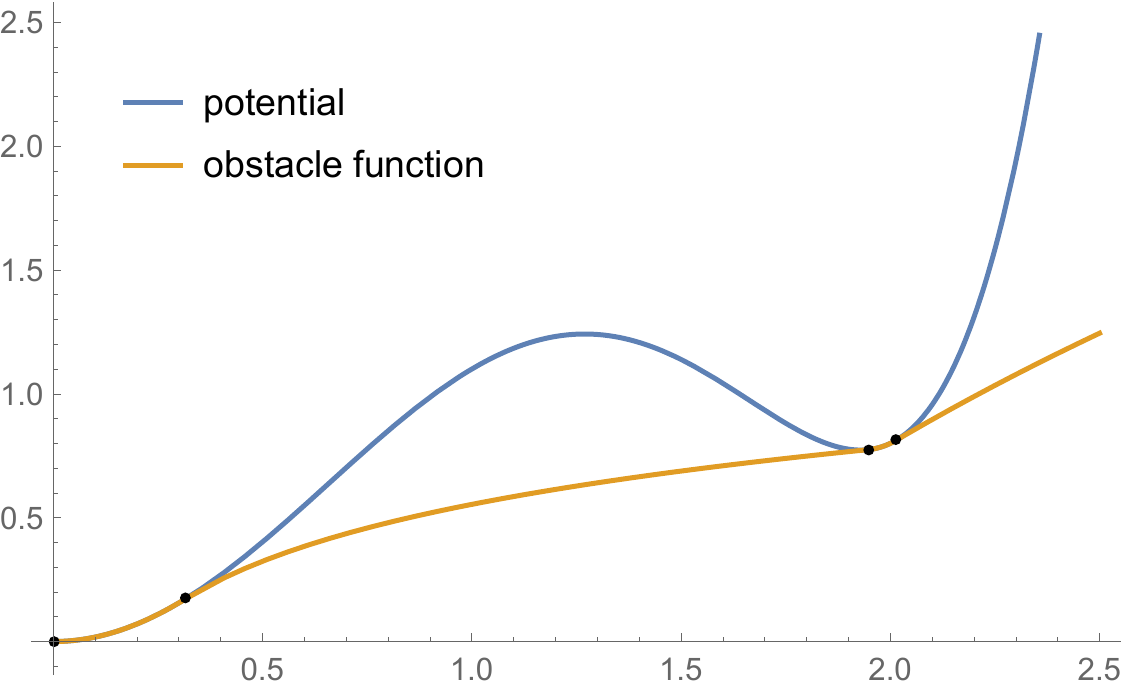}
\caption{The blue and orange curves are radial cross-sections of the graphs of $Q$ and $\check{Q}$, respectively, with $Q=a|z|^6-b|z|^4+c|z|^2$, $a=0.1$, $b=0.8$ and $c=1.8$. Here $N=1$, and the black dots are $a_{0}=0$, $b_{0}=r_{1}$, $a_{1}=r_{2}$, and $b_{1}$ (see \eqref{union}). (For $z\in G$,
$\check{Q}(z)=A+B\log |z|$ with $A$ and $B$ specified by \eqref{obs1}, \eqref{obs3}.)}
\label{fig1}
\end{figure}

\subsubsection{Class of potentials} \label{kalops} Let $Q$ be a potential that obeys the assumptions above and which is
radially symmetric,
\begin{align}\label{Q is rad symm}
Q(z)=q(|z|),\qquad q:[0,+\infty) \to \R.
\end{align}

We remark that already the case when $q(r)$ is a polynomial in $r^2$ is quite rich, and a concretely minded reader may think of this class
in what follows.

It is easily seen that the connected components of the droplet are then  closed concentric annuli (some of which might be circles). To avoid ``degenerate'' cases, we will assume that the droplet $S$ is a finite union of annuli:
\begin{equation}\label{union}S=\cup_{j=0}^N\{a_j\le |z|\le b_j\}\end{equation}
where $0\le a_0<b_0<a_1<b_1<\cdots$.

To focus on the main case of interest, we shall assume that $N\ge 1$ and thus that there is a ``gap'', i.e., a component of $\C\setminus S$ of the form
$$G=\{r_1<|z|<r_2\}$$
where we write $r_1=b_j$ and $r_2=a_{j+1}$ for some $j$ between $0$ and $N-1$.

In addition to our above assumptions, we will generally (unless the opposite is made explicit) assume that $Q$ is $C^6$-smooth in some neighbourhood of $\d G$.

We shall study one- and two-point correlations in a neighbourhood of the closure of a gap, especially near the boundary circles $|z|=r_1$ and $|z|=r_2$.
We remark that we can also regard the unbounded component of the complement, $U=\{|z|>b_N\}\cup\{\infty\}$
as a spectral gap.

Asymptotics near the outer boundary $|z|=b_N$ have been well studied, for example in \cite{ADM,AC,AKM,HW,LR,FJ} and references therein, but we shall nevertheless find some new contributions also for this case.

\subsubsection{Determinantal point processes} Given an $n$-point configuration $\{z_j\}_1^n$ we define the Hamiltonian
\begin{equation}\label{hamn}H_n=\sum_{j\ne k}\log\frac 1 {|z_j-z_k|}+n\sum_{j=1}^n Q(z_j).\end{equation}

With $dA_n(z_1,\ldots,z_n)=dA(z_1)\cdots dA(z_n)$ the normalized Lebesgue measure in $\C^n$, we then consider the Gibbs probability measure
\begin{equation}\label{bogi}d\mathbb{P}_n=\frac 1 {Z_n}e^{-H_n}\, dA_n,\end{equation}
where $Z_n=\int_{\C^n}e^{-H_n}\, dA_n$ is the partition function.

The Coulomb gas in external potential $Q$ (at inverse temperature $\beta=2$) is a sample $\{z_j\}_1^n$, picked randomly with respect to $\mathbb{P}_n$.

For $k\le n$ we define the $k$-point correlation function as the unique continuous function
on $\C^k$
satisfying (with $\{z_j\}_1^n$ a random sample and $\mathbb{E}_n$ expectation with respect to $\mathbb{P}_n$)
$$\mathbb{E}_n[f(z_1,\ldots,z_k)]=\frac {(n-k)!}{n!}\int_{\C^k}f\, R_{n,k}\, dA_k$$
for all compactly supported continuous functions $f$ on $\C^k$.

For each fixed $k\ge 1$,  we have Johansson's convergence theorem \cite{Deift,J,HM}
$$\frac 1 {n^k} R_{n,k}\, dA_k\to d\sigma(z_{1})\ldots d\sigma(z_{k}),\qquad (n\to\infty),$$
in the weak sense of measures on $\C^k$.

As is well-known (see e.g. \cite{M,ST}) the process $\{z_j\}_1^n$ is determinantal, i.e., there exists
a correlation kernel $K_n(z,w)$ such that $R_{n,k}(z_1,\ldots,z_k)=\det(K_n(z_i,z_j))_{i,j=1}^k$.

The kernel $K_n(z,w)$ is only determined up to multiplication by cocycles $c_n(z,w)=h_n(z)/h_n(w)$ where $h_n$ is a non-vanishing measurable function.
We fix a canonical choice in the following way.

Let $\calW_n\subset  L^2=L^2(\mathbb{C},dA)$ be the subspace of all weighted (holomorphic) polynomials on $\C$ of the form
$$\calW_n=\{p=P\cdot e^{-\frac n 2 Q}\,:\, P \mbox{ is a holomorphic polynomial of degree }
\le n-1\},$$
where the norm in $L^2$ is defined by $\|f\|^2=\int_\C|f|^2\, dA$.

The canonical correlation kernel $K_n$ is just the reproducing kernel for the space $\calW_n$, i.e.,
$$K_n(z,w)=\sum_{j=0}^{n-1}\frac {p_j(z)\overline{p_j(w)}}{\|p_j\|^2},$$
where $\{p_j\}_0^{n-1}$ is the orthogonal basis of $\calW_n$ consisting of the weighted monomials
$$\pol_j(z)=z^je^{-\frac n 2 q(r)},\qquad (r=|z|).$$

Following Mehta \cite{M}, we denote the 1-point function by
\begin{equation}\label{1pt}R_n(z):=K_n(z,z).\end{equation}

\subsubsection{Error-function asymptotics}
For ease of reference, we note the following fact.

\begin{prop} \label{crude} Suppose that $Q$ is radially symmetric, is $C^2$-smooth in a neighbourhood of $S$, and strictly subharmonic on $\d G$.
Let $p$ be a boundary point of $S$ and let $\normal_1(p)$ be the unit normal to $\d S$ pointing out of $S$.
Then, as $n\to\infty$, we have, uniformly for $t$ in  compact subsets of $\C$,
\begin{equation}\label{unie}R_n(p+\frac {t}{\sqrt{2n\Delta Q(p)}}\normal_1(p))=n\Delta Q(p)\frac 1 2 \erfc t +o(n),\end{equation}
where $\erfc$ is the usual complementary error function
\begin{equation}\label{erfcdef}\erfc t=\frac 2 {\sqrt{\pi}}\int_t^{+\infty}e^{-s^2}\, ds.\end{equation}
\end{prop}

\begin{proof}[Remark on the proof.] As far as we are aware, error-function asymptotics of the above type was first noted for the case of the Ginibre ensemble in \cite{FH}.
Universality of the asymptotic in \eqref{unie} for a large class of potentials was settled in \cite{HW}, which however
only discusses asymptotic at an outer boundary $\d U$. On the other hand, Proposition \ref{crude} is immediate from \cite[Theorem 1.8]{AKM}, which applies also for other boundary components.
\end{proof}

In what follows, we shall find and exploit a subleading term in \eqref{unie}, which typically turns out to be of order $\sqrt{n}$. Subleading terms are commonly referred to as finite-size corrections \cite{GFF} in random matrix theory and appeared also for the elliptic Ginibre ensemble in the work \cite{LR}; more on that below. (Finite-size corrections find applications in e.g.~\cite{FM1,FM2}.)

We note also that \cite{MR} gives a recent application of asymptotics \`{a} la \eqref{unie} to freezing problems.

\subsubsection{Twin peaks}

Since the obstacle function $\check{Q}(z)$ is harmonic and radially symmetric in the gap $G=\{r_1<|z|<r_2\}$ there are constants $A$ and $B$ such that
$$\check{Q}(z)=A+B\log r,\qquad r_1\le r=|z|\le r_2.$$
(This follows by a well-known theorem on harmonic functions on annuli, see e.g. \cite{Ax,ST}.)

Since $\check{Q}=Q$ on $\d G$ we have
\begin{equation}\label{obs1}A=q(r_1)-B\log r_1=q(r_2)-B\log r_2.\end{equation}

By $C^1$-smoothness of $\check{Q}$ and $Q$ in the radial direction at $r=r_1$, we find that
$$q'(r_1)=\frac d {dr}(A+B\log r)\bigg|_{r=r_1}=\frac B {r_1}.$$
Since a similar relation holds at $r=r_2$ we find
\begin{equation}\label{obs2}B=r_1q'(r_1)=r_2q'(r_2).\end{equation}

Note that \eqref{obs1} implies
\begin{equation}\label{obs3}
B=\frac {q(r_2)-q(r_1)} {\log(r_2/r_1)}.
\end{equation}
The following proposition (essentially a case of Gauss' Theorem \cite[Theorem II.1.1]{ST}) gives an intrinsic meaning to the parameter $B$, in terms of the equilibrium measure $\sigma$.
\begin{prop}\label{paralem}
	The parameter $B$ equals to
$B=2\cdot \sigma(\{z\,;\, |z|\le r_1\}).$
\end{prop}

\begin{proof} Recall that the droplet is given as the union \eqref{union}. For $k=0,\ldots,N-1$, write
 $B_k:=b_k q'(b_k)$. By \eqref{obs2} 
 we know that $B_k=a_{k+1}q'(a_{k+1})$. Using also \eqref{Q is rad symm}, we can thus write
\begin{align*}
\frac{B_k}{2} = \frac{1}{2\pi i} \int\limits_{|z|=b_k} \partial Q(z)\, dz = \frac{1}{2\pi i} \int\limits_{|z|=a_{k+1}} \partial Q(z)\, dz,
\end{align*}
where the curves are oriented in the counterclockwise direction.

By Stokes' theorem
$$
\frac{1}{2\pi i} \int\limits_{|z|=b_k} \partial Q(z)\, dz - \frac{1}{2\pi i} \int\limits_{|z|=a_k} \partial Q(z)\, dz = \int\limits_{ a_k\leq|z|\leq b_k} \Delta Q(z)\, dA(z),
$$
and hence, by \eqref{def of eq measure}, $\frac{B_k}{2}=\frac{B_{k-1}}{2} + \sigma(\{z: a_k\leq |z|\leq b_k\})$ for $k=0,\ldots,N$, where $B_{-1}:=a_{0}q'(a_{0})=0$ (since $\check{Q}$ is constant on the disc $|z|\le a_0$). By iteration we get
$$
\frac{B}{2}=	\frac{B_j}{2} = \frac{B_{-1}}{2}+ \sum\limits_{k=0}^{j} \sigma(\{z: a_k\leq |z|\leq b_k \}) = \frac{B_{-1}}{2}+\int\limits_{|z|\leq b_j} d\sigma,
$$	
for all $j$, $0\le j\le N$. Since $B_{-1}=0$, the claim follows.
\end{proof}

Let us denote by $V(z)$ the harmonic continuation of $\check{Q}|_{G}$ across the boundary $\d G$, i.e.,
\begin{equation}\label{Vfun}V(z)=A+B\log|z|.\end{equation}
The function $Q-V$ vanishes to first order on $\d G$ and is strictly positive on $N_{G}\setminus (\d G)$ where $N_G$ is a neighbourhood of $\overline{G}$.
(See \cite[Lemma 3.4]{AC} for an estimate valid for more general potentials.)

Now fix an integer $j$ close to $nB/2$ and consider the weighted polynomial $p_j(z)=z^je^{-\frac n 2 Q(z)},$
which satisfies
\begin{equation}\label{wassa}|p_j(z)|^2=e^{-nA}e^{(2j-nB)\log|z|}e^{-n(Q-V)(z)},\qquad (z\in\C).\end{equation}

We claim that
there are two local peak-points $r_{1,j}$ and $r_{2,j}$, close to $r_1$ and $r_2$ respectively, where $|p_j|$ achieves local maxima.
 (See Figure \ref{pic2} for an illustration.)

In what follows, we write ``$A_n\lesssim B_n$'' to denote that $A_n\le CB_n$ for all large enough $n$,
where $C$ is some constant independent of $n$. The notation $A_n\asymp B_n$ indicates that $A_n\lesssim B_n$ and $A_n\gtrsim B_n$.

\begin{lem}\label{neglo}
Suppose that $|j-nB/2|\lesssim \sqrt{n}\log n$. Then (for $k=1,2$)
\begin{itemize}
\item There are constants $\epsilon>0$ and $c>0$ such that if $r=|z|$ satisfies $|r-r_{k,j}|\le \epsilon$, then
$$|p_j(z)|^2\le |p_j(r_{k,j})|^2e^{-cn(r-r_{k,j})^2}.$$
\item If $|r-r_{k,j}|>\epsilon$, then $|p_j(z)|$ is exponentially small compared with its peak value:
$$|p_j(z)|\lesssim e^{-cn}\max\{|p(r_{1,j})|,|p(r_{2,j})|\}.$$
\end{itemize}

\end{lem}

\begin{proof}
A detailed proof of the lemma is straightforward from \eqref{wassa} using (for example) general principles in \cite[Lemma 3.4]{AC} and \cite[Lemma 3.5]{AC}; we omit details here.
\end{proof}

\begin{figure}[h]
\begin{center}
\includegraphics[width=11cm]{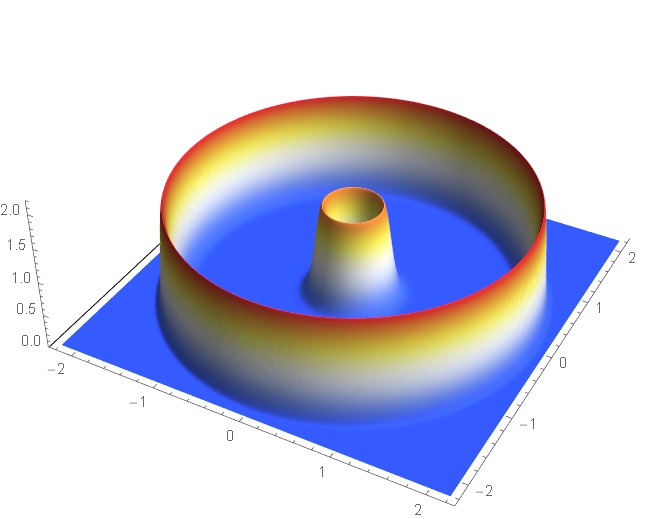}
\end{center}
\caption{Graph of the wavefunction $\frac{|p_j(z)|^2}{\|p_j\|^{2}}$ for $j=5$ and $n=30$ in potential $Q=a|z|^6-b|z|^4+c|z|^2$ with $a=0.1$, $b=0.8$ and $c=1.8$. Here $j=5 \approx \frac{Bn}{2}\approx 4.9735$.}
\label{pic2}
\end{figure}

\subsection{Fluctuations for real parts of analytic functions} Let $Q$ be a radially symmetric potential, satisfying the assumptions in \ref{kalops} above.

 Consider a smooth test function $f(z)$ which vanishes identically outside of a small neighbourhood of the closure $\overline{G}=\{r_1\le |z|\le r_2\}$ of the gap. Given a random sample $\{z_j\}_1^n$ from \eqref{bogi}, we consider the random variable (linear statistic)
\begin{align}\label{def of fluctn}
\fluct_n f=\sum_{j=1}^n f(z_j)-n\sigma(f),
\end{align}
where $\sigma(f):=\int_{\mathbb{C}}f(z)\, d\sigma(z)$.

Let us now, for simplicity, assume that $f$ is \emph{harmonic} in some neighbourhood of $\overline{G}$.  In terms of a smooth, compactly supported function $\lfun(z)$ which equals $\log|z|$ in a neighbourhood of $\overline{G}$ we can then uniquely decompose
\begin{equation}\label{decc}f(z)=\re g(z)+c\lfun(z)\end{equation}
where $g(z)$ is an analytic function in a neighbourhood of $\overline{G}$ and $c=\frac 1 {\pi i}\int_{|z|=r_1}\d f(z)\, dz$. (See for example the ``Logarithmic Conjugation Theorem'' in \cite{Ax}, cf.~also \cite{Bell,ST}.)

Given a function $h$, we consider the \emph{cumulant generating function} of $\fluct_n h$:
\begin{equation}\label{CGF1} F_{n,h}(t)=\log \mathbb{E}_n\, e^{t\fluct_n h}.
\end{equation}

In the following, the symbol $\dn$ designates differentiation in the outwards normal direction on $\d S$ (i.e., the normal direction pointing outwards from $S$).

We note the following result, a slight generalization of the main result from \cite{AM}, whose proof is in fact implicit there.

\begin{thm} \label{thereg} Suppose that $g$ is analytic in some annulus
$A_\epsilon=\{z\, ;\,r_1-\epsilon<|z|<r_2+\epsilon\}$. Modify $g$ outside the smaller annulus
$A_{\epsilon/2}$ to a $C^\infty$-smooth function with support in $A_\epsilon$. (Here $\epsilon$ small enough so that $A_\epsilon\cap(\C\setminus S)=G$.)

Then, as $n \to +\infty$, the real part $f=\re g$ satisfies, for each $t$,
\begin{equation}\label{cgf}
F_{n,f}(t)= te_f+\frac {t^2}2v_f+\bigO(n^{-\beta}),
\end{equation}
where $\beta>0$ is small enough and (writing ``$|dz|$'' for the arclength measure on $\d S$)
\begin{align}\label{expf}e_f&=\frac 1 2\int_S f\cdot \Delta\log\Delta Q\, dA+\frac
1 {8\pi}\int_{\d S} \partial_{\mathrm{n}} f\,|dz|-\frac 1 {8\pi}\int_{\d S} f\frac {\partial_{\mathrm{n}}\Delta Q}{\Delta Q}\, |dz|,\\
\label{varf}v_f&=-\int_Sf\Delta f\, dA.\end{align}
In particular, $\fluct_n f$ converges in distribution to the normal $N(e_f,v_f)$-distribution.
\end{thm}

As indicated above, we shall deduce the above theorem as a consequence of a more general statement (for general potentials, not necessarily radially symmetric), using the limit Ward identity from \cite{AM}. See Theorem \ref{gengap} below.


\subsection{Fluctuations with oscillatory behaviour}
Having handled the ``simple'' case of the real part of an analytic function,
we now turn
to the function $\lfun$ in \eqref{decc}. In fact we shall allow more general $\lfun$ which are arbitrary,
smooth enough functions of $r=|z|$, and which (to focus on the essentials) vanishes identically
outside of a small neighbourhood of $\overline{G}$.

In order to formulate our next result, we recall the definition of the discrete Gaussian distribution from \cite{K,Sz}.

\begin{defn}\label{def:discrete Gaussian} We say that an integer-valued random variable $X$ has a discrete normal distribution with parameters $\alpha \in \mathbb{R}$ and $\dNpara\in (0,1)$, and we write $X\sim dN(\alpha,\dNpara)$, if
$$
\mathbb{P}(\{X=k\}) = \frac{\dNpara^{\frac 1 2(k-\alpha)^2}}{I(\alpha,\dNpara)}, \qquad k\in \mathbb{Z},
$$
where
\begin{align}\label{def of I in def}
I(\alpha,\dNpara)= \sum\limits_{j=-\infty}^{+\infty} \dNpara^{\frac 1 2(j-\alpha)^2}.
\end{align}
\end{defn}

The distribution $dN(\alpha,\dNpara)$ can be characterized as the one maximizing the Shannon entropy among probability distributions supported on $\mathbb{Z}$ with given expectation and variance, see \cite{K}.

\smallskip

In the following, when $t$ is a real number, we write $\lfloor t\rfloor$ for the integer part and $\{t\}$ for the fractional part. (Thus $0\le \{t\}<1$ and $t=\lfloor t\rfloor+\{t\}$.)

We have the following theorem. (As before the potential $Q$ is assumed to satisfy the conditions in \ref{kalops}.)

\begin{thm}\label{flucco}
Let $\lambda(z)=\lambda(|z|)$ be a radially symmetric and $C^6$-smooth function which vanishes outside of a small neighbourhood of $\overline{G}$. Write
\begin{equation}\label{params}x=x(n):=\left\{\frac {Bn} 2\right\},\qquad \dNpara:=(\frac {r_1}{r_2})^2,\qquad \alpha=\alpha(n):=x+\frac {\log\frac {\Delta Q(r_2)}{\Delta Q(r_1)}}{4\log\frac {r_2}{r_1}}.\end{equation}

Let $X_{n}\sim dN(\alpha,\dNpara)$ and define
$Y=Y_n:=(\lambda(r_1)-\lambda(r_2))(X_{n}-\alpha)$ and
\begin{equation}\label{oscterm}\calF_n(t)=\log\mathbb{E}\,e^{tY}, \qquad t \in \mathbb{R}.\end{equation}

Then, as $n \to + \infty$,
\begin{equation}\label{cgf0}
	\log \mathbb{E}_n\,e^{t \fluct_n \lambda} = t(e_\lambda + \hat{e}_{\lambda} ) +\frac{t^2}{2}(v_{\lambda} + \hat{v}_{\lambda}) + \calF_n(t) + \bigO( \frac
{\log^5 n}
{\sqrt{n}} ),
\end{equation}
	uniformly for $|t| \leq \log n$, $t \in \mathbb{R}$.

The coefficients $e_\lambda$ and $v_\lambda$ are given by \eqref{expf} and \eqref{varf} respectively, with $h=\lambda$, while
	\begin{align*}
	& \hat{e}_{\lambda} = (\lambda(r_1)-\lambda(r_2))\frac {\log \frac {\Delta Q(r_2)}{\Delta Q(r_1)}}{4\log\frac {r_2}{r_1}}, \\
	& \hat{v}_{\lambda} = \frac {r_1 \lambda(r_{1})\lfun'(r_1)-r_2 \lambda(r_{2})\lfun'(r_2)}{2}.
	\end{align*}
\end{thm}

For two-dimensional ensembles with soft edges, Theorems \ref{thereg} and \ref{flucco} provide first results on the asymptotic distribution of fluctuations for a large class of smooth linear statistics in a regime of disconnected droplets. After this work was completed, further and complementary results were obtained in \cite{ACC,AC:Heine}. A comparison is given in Section \ref{comparison} below.

\begin{rem} It is possible to improve the error term in \eqref{cgf0}. In fact, by keeping closer track, we can prove that the error term has the bound
$\bigO(\frac {(\log n)^{\frac 52}+|t|^{5}}{\sqrt{n} })$ when $|t|\le C\log n$. We believe that the exact order of magnitude of the correction term is $\bigO(1/\sqrt{n})$.
\end{rem}

\subsubsection{Interpretation in terms of the Jacobi theta function} We can rephrase the oscillatory term $\calF_n(t)$ using
a variant of Jacobi's theta function, which plays a central role below.

\begin{defn}[``Jacobi theta function''] The Jacobi theta function is defined by
\begin{align}\label{def of Jacobi theta}
\theta(z;\tau) := \sum_{\ell=-\infty}^{+\infty} e^{2\pi i \ell z}e^{\pi i \ell^{2}\tau}, \qquad z \in \mathbb{C}, \qquad \tau \in i(0,+\infty).
\end{align}
This function satisfies $\theta(-z)=\theta(z)$ and $\theta(z+1;\tau)=\theta(z;\tau)$, i.e. it is even and periodic of period $1$, see e.g.~\cite[Chapter 20]{NIST} or \cite[Chapter 10]{S2} for other properties of $\theta$.

\end{defn}

In terms of the Jacobi theta function, the term $\calF_n(t)$ in \eqref{oscterm} takes the form
\begin{align}\label{brock}\calF_n(t)=\frac{t^2}{2} \frac{(\lambda(r_{1})-\lambda(r_{2}))^{2}}{2 \log(r_{2}/r_{1})} +\log \frac{\theta(\frac{B}{2}n+ \frac{\log \frac{\Delta Q (r_{2})}{\Delta Q(r_{1})}}{4\log (r_{2}/r_{1})} + t\frac{\lambda(r_{1})-\lambda(r_{2})}{2\log (r_{2}/r_{1})}; \frac{\pi i}{\log(r_{2}/r_{1})})}{\theta(\frac{B}{2}n+ \frac{\log \frac{\Delta Q (r_{2})}{\Delta Q(r_{1})}}{4\log (r_{2}/r_{1})}; \frac{\pi i}{\log(r_{2}/r_{1})})}.
\end{align}

The fact that the above expression equals to $\log\mathbb{E}\, e^{tY}$, where $Y$ is the random variable in Theorem \ref{flucco}, is proved at the very end of Section \ref{sec4}.

\begin{rem}
In the multi-cut regime of Coulomb gas ensembles on the real line $\R$, the oscillations are, in general, only quasi-periodic in $n$, see e.g. \cite{Widom1995, DKMVZ1999, BDE2000, Shcherbina, BG2, CFWW}. However, in the particular case where the mass of the equilibrium measure is rational on each interval of its support, these oscillations become periodic, see also e.g. \cite{Marchal}. Interestingly, a similar phenomenon happens in \cite{C, ACCL} as well as in our two-dimensional setting. Indeed, by the periodicity $\theta(z+1;\tau)=\theta(z;\tau)$ we see that $\calF_n(t)$ is periodic in $n$ if and only if $\frac{B}{2}=\sigma(\{z\,;\,|z|\le r_1\})$ is rational.  Moreover, $\calF_n$ is independent of $n$ if and only if $B=2$ or $B=0$, i.e., oscillations are absent on the outer boundary component, and, if $\C\setminus S$ contains a disc about the origin, also the ``innermost one''. Oscillations are thus absent on boundary components of simply connected spectral gaps; on boundary components of ring-shaped gaps, $\calF_n$ oscillates in $n$.
\end{rem}

\begin{rem}
For one-dimensional log-correlated point processes, discrete Gaussians and theta functions are known to emerge in the multi-cut regime, see e.g.~\cite{DIZ1997, DKMVZ1999, G2006, BG2, ClaeysGravaMcLaughlin, BCL1, FK2020, BCL2, BCL3, KM2021, CFWW}. In dimension two, the emergence of theta functions was  conjectured in \cite[Section 1.5]{LSe} and recently
proved in \cite{C} in connection with gap probabilities. The theta function also appears in \cite{ACCL} for the moment generating function of disk counting statistics near a spectral gap with hard edges. Interestingly, in the present paper as well as in \cite{C, ACCL} one obtains the theta function $\theta(\cdot;\tau)$ with the same parameter $\tau =\frac{\pi i}{\log(r_{2}/r_{1})}$. In particular, in the case where the gap is an annulus, the parameter $\tau$ depends only on the conformal type (the ratio $r_2/r_1$) of the gap; $\tau$ does not depend on the potential or on the nature of the edges (i.e. soft or hard).
\end{rem}

\subsubsection{Main strategy} In order to prove Theorem \ref{flucco} and Theorem \ref{thereg}, it is convenient
besides $Q$ to consider perturbed potentials of the form
$$\tilde{Q}=\tilde{Q}_{n,sh}=Q-\frac {sh}n,\qquad (s\in\R).$$
Let us write $\tilde{R}_n=\tilde{R}_{n,sh}$ for the $1$-point function in potential $\tilde{Q}$.

We denote by $\tilde{\mathbb{E}}_{n}$ the expectation with respect to $\tilde{Q}$ and write $\mathbb{E}_n$ for the expectation with respect to $Q$. Thus for example,
$$\tilde{\mathbb{E}}_n\fluct_n f=\int f\tilde{R}_n\, dA-n\int f\, d\sigma.$$

The following lemma is well-known, see e.g. \cite{AM} as well as Subsection \ref{fluco} below.

\begin{lem}\label{lem11} If $h(z)$ is a bounded, measurable function on $\C$, then
$F_{n,h}(t)=\log \mathbb{E}_n(e^{t\fluct_n h})$
obeys
$$F_{n,h}(t)=\int_0^t \tilde{\mathbb{E}}_{n,sh}(\fluct_n h)\, ds=\int_0^t\, ds\int_\C h(z)(\tilde{R}_{n,sh}(z)-n\Delta Q(z)\1_S(z))\, dA(z).$$
\end{lem}

In our proof of Theorem \ref{thereg} we will set $h=\re g$ where $g$ is analytic we also rely on the limit Ward identity from \cite{AM}, as explained in Section \ref{flureg}. In the case of Theorem \ref{flucco} we will set $h=\lfun$ radially symmetric and use fine asymptotics for the density $\tilde{R}_{n,s\lfun}$ near the edge of the gap $G$. We now turn to a detailed description of the latter results.

\subsection{Fine asymptotics for the edge density} \label{perl}
In the following we fix a radially symmetric and $C^6$-smooth function $\lambda(z)=\lambda(r)$, $r=|z|$. We will assume that $\lambda$ vanishes
identically outside a small neighbourhood of $\overline{G}$, and in particular near all other components of $\C\setminus S$ (except for $G$).

We also fix an arbitrary real number $s$ and put
\begin{equation}\label{pertpot}\tilde{Q}(z)=\tilde{Q}_{n,s\lfun}(z):=Q(z)-\frac {s\lfun(z)} n .\end{equation}

We now state three results which substantially improve the error-function asymptotic in Proposition \ref{crude}, for the class of ensembles under study.

\begin{thm}[``$r_1$-case'']\label{sr1} Let $G=\{z\in\C: r_1<|z|< r_2\}$ be a gap in the droplet $S$, consider $\alpha \in [0,2\pi)$  fixed and put
$$z=e^{i\alpha}\bigg(r_1+\frac t {\sqrt{2n\Delta Q(r_1)}}\bigg),\qquad (t\in\R,\,|t|\le \log n).$$
As $n \to + \infty$, the 1-point function in potential \eqref{pertpot} satisfies
\begin{multline}\label{finalr1}
\tilde{R}_{n,s\lambda}(z) = n\Delta Q(r_1) \frac{\erfc(t)}{2} +
 \frac{\sqrt{n \Delta Q(r_1)}}{\sqrt{2\pi} \, r_1} e^{-t^2}
 \Bigg[\frac s 2 \bigg(r_1\lfun'(r_1)+\frac {\lfun (r_1) -\lfun (r_2)} {\log(r_{2}/r_{1})}\bigg)\\
 + \frac{1}{6} (t^2-2) +r_1\frac{\partial_{\mathrm{n}}\Delta Q(r_1)}{\Delta Q(r_1)}\bigg(\frac{1}{2} \sqrt{\pi} t\erfc(t)e^{t^2} -\frac{1}{12}(2t^2+5)\bigg)+\frac{\log\frac{\Delta Q(r_2)}{\Delta Q(r_1)}}{4\log (r_{2}/r_{1})}\\
  + \frac{1}{2 \log(r_{2}/r_{1})} (\log \theta)' \bigg( \frac{B}{2}n + \frac{\log \frac{\Delta Q(r_{2})}{\Delta Q(r_{1})}}{4 \log(r_{2}/r_{1})} + \frac{s(\lambda(r_{1})-\lambda(r_{2}))}{2\log (r_{2}/r_{1})} ; \frac{\pi i}{\log(r_{2}/r_{1})} \bigg)
 \Bigg] + \bigO(\log^4 n).
\end{multline}

\end{thm}

(In \eqref{finalr1} and below, a prime denotes differentiation with respect to the first argument: $\theta'(z;\tau)=\frac {\d\theta}{\d z}(z;\tau)$.)

Figure \ref{fig7} depicts the error term in \eqref{finalr1} at $z=r_{1}$ with $s=0$ for a particular potential.

\begin{figure}[h!]
\includegraphics[width=11cm]{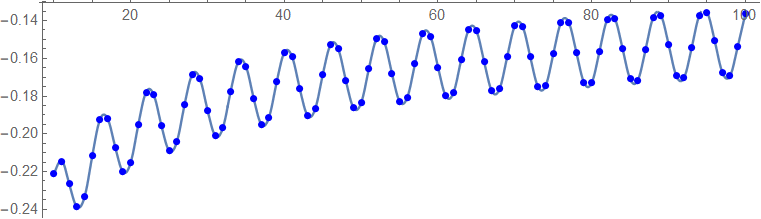}\\[0.2cm]
\includegraphics[width=11cm]{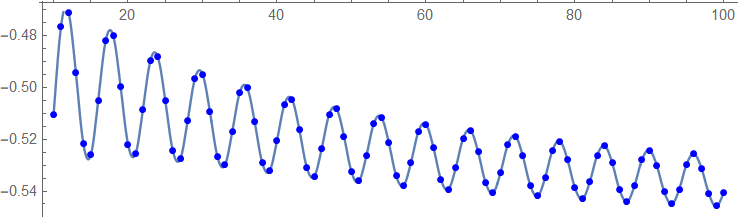}
\caption{Both pictures are made for the potential $Q=a|z|^6-b|z|^4+c|z|^2$ with $a=0.1$, $b=0.8$ and $c=1.8$. In the top picture, the dots represents $(R_{n}(r_{1})-\frac{n}{2}\Delta Q(r_{1}))/\sqrt{n \Delta Q(r_{1})}$ for integer values of $n$. In the bottom picture, the dots represents $R_{n}(r_{1})-\mathcal{E}_{n}$ for integer values of $n$, where $\mathcal{E}_{n}$ is the right-hand side of \eqref{finalr1} with $z=r_{1}$, $s=0$ and without the error term. (Between the dots, we have used spline interpolation merely to add some structure to the illustrations.) The bottom picture suggests that $R_{n}(r_{1})-\mathcal{E}_{n}=\mathcal{O}(1)$ as $n \to + \infty$, which is consistent with Theorem \ref{sr1} (the picture also suggests that the $\mathcal{O}(\log^{4} n)$-term in \eqref{finalr1} is in \\ fact $\mathcal{O}(1)$).}
\label{fig7}
\end{figure}

\begin{rem} For $t$ large negative we recover the known bulk asymptotics $\tilde{R}_{n,s\lambda}=n\Delta Q+o(n)$ from \cite{A2} (see also \eqref{bulkr} below). In fact already for $t = -C\sqrt{\log n}$ we have
\begin{align*}
\tilde{R}_{n,s\lambda}(z) = n\Delta Q(r_1) +
\frac{\sqrt{n}}{\sqrt{2}} \frac{\nlap Q(r_1)}{\sqrt{\Delta Q(r_1)}} t  + \bigO(\log^4 n), \qquad \mbox{as } n \to + \infty.
\end{align*}
Using now $\Delta Q(r_1)=\Delta Q(z)-\frac{t}{\sqrt{2n\Delta Q(r_1)}}\nlap Q(r_1)+\bigO(\log n)$, we get
\begin{align*}
\tilde{R}_{n,s\lambda}(z) = n\Delta Q(z) + \bigO(\log^4 n), \qquad \mbox{as } n \to + \infty.
\end{align*}

\end{rem}

\begin{thm}[``$r_2$-case'']\label{sr2} Let $G=\{z\in\C: r_1<|z|< r_2\}$ be a gap in the droplet $S$, consider $\alpha \in [0,2\pi)$ fixed and
$$z=e^{i\alpha}\bigg(r_2+\frac t {\sqrt{2n\Delta Q(r_2)}}\bigg),\qquad (t\in\R,\,|t|\le \log n).$$
As $n \to + \infty$, the 1-point function in potential \eqref{pertpot} satisfies
\begin{multline}\label{finalr2}
\tilde{R}_{n,s\lambda}(z) = n \Delta Q(r_2) \frac{\erfc(-t)}{2} +
\frac{\sqrt{n \Delta Q(r_2)}}{\sqrt{2\pi} \, r_{2}} e^{-t^2}\Bigg[ -\frac s 2\bigg(r_2\lambda'(r_2)+\frac {\lfun(r_1)-\lfun(r_2)}{\log(r_{2}/r_{1})}\bigg) \\
+ \frac{1}{6} (2-t^2) +  r_{2}\frac{\partial_{\mathrm{n}} \Delta Q(r_2)}{\Delta Q(r_2)}\bigg(\frac{1}{2} \sqrt{\pi} t\erfc(-t)e^{t^2} +\frac{1}{12}(2t^2+5)\bigg) -\frac{\log\frac{\Delta Q(r_2)}{\Delta Q(r_1)}}{4\log (r_{2}/r_{1})}  \\
- \frac{1}{2\log(r_{2}/r_{1})} (\log \theta)' \bigg( \frac{B}{2}n + \frac{\log \frac{\Delta Q(r_{2})}{\Delta Q(r_{1})}}{4 \log(r_{2}/r_{1})} + \frac{s(\lambda(r_{1})-\lambda(r_{2}))}{2\log (r_{2}/r_{1})} ; \frac{\pi i}{\log(r_{2}/r_{1})} \bigg) \bigg]
+ \bigO(\log^4 n).
\end{multline}

\end{thm}

Near the outer boundary $\d U=\{|z|=b_N\}$, we obtain the following simpler but yet noteworthy result. Since the parameter $s$ is of lesser interest in this case,
we state it just with respect to the 1-point function $R_n$ in (unperturbed) potential $Q$.

\begin{thm}[``Outer boundary case''] \label{ur3} Let $\{|z|=b_N\}$ be the boundary of the unbounded component $U$ of $\hat{\C}\setminus S$, and put
$$z=b_N+\frac t {\sqrt{2n\Delta Q(b_N)}},\qquad (t\in\R,\,|t|\le\log n).$$
Then as $n\to+\infty$, the $1$-point function $R_n$ in potential $Q$ satisfies
\begin{multline} \label{bolv}
R_n(z)=n\Delta Q(b_N)\frac {\erfc t} 2+\frac {\sqrt{n\Delta Q(b_N)}}{\sqrt{2\pi}\,b_N}e^{-t^2} \bigg[\frac 1 6(t^2-2)
 \\ +b_N\frac {\partial_{\mathrm{n}}\Delta Q(b_N)}{\Delta Q(b_N)}
\Big(\frac {\sqrt{\pi}} 2\,
t\, e^{t^2}\erfc t-\frac 1 {12}(2t^2+5)\Big)\bigg]+\bigO(\log^4 n).
\end{multline}
\end{thm}

Our proof of Theorem \ref{ur3} is similar to our proofs of Theorem \ref{sr1} and Theorem \ref{sr2} but is simpler, since no theta functions enter the picture. We shall therefore omit a detailed proof of Theorem \ref{ur3}.

\begin{rem}
Theorem \ref{ur3} is related to a question due to Lee and Riser \cite{LR}, who considered the elliptic Ginibre ensemble, associated with the potential $Q=c(|z|^2-\alpha\re(z^2))$, where $c=\Delta Q>0$
and
$\alpha$ is fixed with $0\le \alpha<1$. The droplet is the elliptic disc
\begin{align*}
S_\alpha=\{z=x+iy\,;\,\frac {1-\alpha}{1+\alpha}x^2+\frac {1+\alpha}{1-\alpha}y^2\le \frac 1 c\}.
\end{align*}

In \cite[Theorem 1.1]{LR} ($R_{n}$ here corresponds to $\pi c n \rho_{n}$ in \cite{LR}) it is shown that if $p$ is a point on $\d S_\alpha$, then for $|t|\le \log n$,
\begin{equation}\label{ellipse}
R_n\bigg(p+\frac t {\sqrt{2nc}}\textrm{n}_1(p)\bigg)=nc\frac {\erfc t} 2+\sqrt{\frac {nc} {2\pi}}\,\kappa(p)\,e^{-t^2}\tfrac 1 6 (t^2-2)+e^{-t^2}P(t;p)+\bigO(n^{-\beta}),
\end{equation}
where $\kappa(p)$ is the curvature of $\d S_\alpha$ at $p$, $\textrm{n}_1(p)$ is the outwards unit normal, and
where $P(t;p)$ is an explicit fifth-degree polynomial in $t$, which depends on the curvature $\kappa$ and its tangential derivative $\d_{\textrm{s}}\kappa$ at the point $p$. The number $\beta$ is arbitrary with $0<\beta<\frac 1 2$.

Since the curvature of the circle $|z|=b_N$ is just $\kappa=1/b_N$, and since $\nlap Q(b_{N})=0$ if $Q=c(|z|^2-\alpha\re(z^2))$, the $\sqrt{n}$-term in \eqref{ellipse}
matches our formula
\eqref{bolv} in the radially symmetric case $\alpha=0$.

In \cite{LR}, the authors conjecture that the $\sqrt{n}$-term in \eqref{ellipse} is universal for a ``general droplet with real analytic boundary''. At an outer boundary, this may very well be the case for the class of \textit{Hele-Shaw potentials}, for which $\Delta Q$ is constant in a neighbourhood of the droplet, but in view of \eqref{bolv} the conjecture is false at an outer boundary, for potentials such that the normal derivative $\dn\Delta Q$ is non-vanishing there.
\end{rem}

\begin{rem} In \cite{HW22,HW}, an asymptotic expansion in powers of $1/n$, for the $n$:th orthonormal polynomial near the outer boundary, is introduced. By contrast, in our derivations of fine edge asymptotics
up to $\bigO(\log^4 n)$, we do not require any knowledge of such an expansion beyond the leading term. The success of our method rather depends on a sufficiently good localization of the peak-points of the weighted polynomials. However, we expect that subleading terms from the expansion in \cite{HW}
will enter the picture in a finer analysis of the $O(1)$-term in the expansion of $R_n(z)$, such as given in
\eqref{ellipse} for the case of the elliptic Ginibre potential. In fact, a similar phenomenon has been observed for one-dimensional point-processes. For instance in \cite{GFF} the leading-order asymptotics of the Hermite polynomials is used to derive subleading correction terms for the one-point density.
\end{rem}

\subsection{Two-point correlations and the weighted Szeg\H{o} kernel}
Asymptotics for $K_n(z,w)$ with $z$ and $w$ near the boundary of the unbounded component $U$ have been studied in \cite{FJ,AC,ADM}; cf.~also \cite{ES}, for example. In \cite{AC}, some of the authors found a universal asymptotic formula in terms of the (unweighted) Szeg\H{o} kernel $S^U(z,w)$, which is the reproducing kernel for the Hardy space $H^2_0(U)$ of all holomorphic functions on $U$ which vanish at infinity, equipped with the arc-length norm $\int_{\d U}|f|^2\,|dz|$. (In the setting \eqref{union}, $S^U(z,w)$ takes the form $S^U(z,w) = \smash{\frac{1}{2\pi}}\smash{\sum_{\ell=-\infty}^{\smash{-1}}} z^{\ell} \overline{w}^{\ell}b_{N}^{-2\ell-1}$, see also \cite[eq (1.20)]{AC}.)

It is shown in \cite{AC} that
\begin{equation}\label{forex}K_n(z,w)\sim \sqrt{2\pi n}\Delta Q(z)^{\frac 1 4}\Delta Q(w)^{\frac 1 4}S^U(z,w)
\qquad \text{if}\qquad z,w\in \d U,\, z\ne w,\end{equation}
where the symbol $\sim$ indicates the relation ``asymptotically as $n\to\infty$ and up to $n$-dependent cocycles''.

We now reconsider the question of boundary asymptotics of $K_n(z,w)$, but for $z$ and $w$ in a microscopic neighbourhood of the boundary of the gap $G$. Again $K_n(z,w)$ is of order $\bigO(\sqrt{n})$, but in this case we do not obtain convergence to a limiting kernel, but rather an oscillatory behaviour in $n$. It turns out that the phenomenon can be conveniently described using a certain weighted Szeg\H{o}-kernel, such that the oscillations are incorporated in the weight.

\begin{defn} By the \emph{$n$-weighted Szeg\H{o} kernel} $\calS^G(z,w;n)$ pertaining to the gap $G=\{r_1<|z|<r_2\}$ and the parameter (``number of particles'') $n$, we mean the analytic function of $z$ and $\bar{w}$, defined for $z,w\in G$ by
\begin{equation}\label{wsg}\calS^G(z,w;n):=\frac 1 {2\pi}\sum_{\ell=-\infty}^{+\infty}\frac {(z\bar{w})^\ell}{\frac{{r_1}^{2\ell+1-2x}}{\sqrt{\Delta Q(r_1)}}+\frac {r_2^{2\ell+1-2x}}{\sqrt{\Delta Q(r_2)}}},
\qquad x=x(n):=\left\{\frac {B n} 2\right\}.
\end{equation}
\end{defn}

\begin{rem} The kernel \eqref{wsg} can be recognized as the reproducing kernel for the weighted Hardy space $H^2(G;n)$ consisting of all analytic functions on $G$ such that
	\begin{equation}\label{wnorm}\|f\|_{H^2(G;n)}^2:=\int_{\d G}|f(z)|^2|z|^{-2x(n)}(\Delta Q(z))^{-\frac 1 2}\,|dz|<\infty,\end{equation}
	where we have used the usual convention of identifying a Hardy-space function with its (a.e.) nontangential boundary function.

In this connection, we remark that the natural weighted counterpart
for the unbounded component $U$ of $\C\setminus S$ is
	$$\calS^U(z,w):=e^{\frac 1 2(\calH(z)+\overline{\calH(w)})}\,S^U(z,w),$$ where $S^U$ is the reproducing kernel for $H^2_0(U)$
and $\calH$ is the holomorphic function on $U$ which satisfies $\re\calH=\log\sqrt{\Delta Q}$ on $\d U$ and $\im\calH(\infty)=0$. It is easy to check that $\calS^U(z,w)$ is the reproducing kernel for the weighted Hardy space $H^2(U;n)$, obtained by replacing ``$G$'' by ``$U$'' in \eqref{wnorm}.

Since $B=2$ at the outer boundary we have $x(n)=0$ there, so $H^2(U;n)$ is independent of $n$. The kernel $\calS^U(z,w)$ is in many ways more natural than $S^U$; for example the asymptotic formula \eqref{forex} simplifies to
$$K_n(z,w)\sim \sqrt{2\pi n}\cdot \calS^U(z,w)$$ when $z,w\in\d U$ and $z\ne w$.
\end{rem}

Let us denote $N(G,\delta)=\{z\,:\, \dist(z,G)<\delta\}$. We have the following theorem.

\begin{thm} \label{thmr1r2} Let
$z,w\in N(G,\frac{\log n}{\sqrt{n}})$ be such that $||z|-|w||>c$ for some constant $c>0$.
	Then, as $n \to \infty$, we have
	\begin{multline}K_n(z,w)=\sqrt{2\pi n} \cdot \calS^G(z,w;n)\cdot (z\bar{w})^m \cdot (r_1r_2)^{-\tfrac{Bn}{2}}
	\\ \times
	e^{\frac n 2(Q(r_1)-Q(z))+\frac n 2(Q(r_2)-Q(w))}+
	\bigO(\log^3 n),
	\end{multline}
where $m=\lfloor Bn/2\rfloor$ and $x=\{Bn/2\}$.
\end{thm}

We are particularly interested in the case when both $z$ and $w$ are in a microscopic neighbourhood of the boundary $\d G$. The following corollary
is immediate from Theorem \ref{thmr1r2} (see Section \ref{section coro}); it explains the situation when $z$ and $w$ are near different boundary components.

\begin{cor}[``$r_1$-$r_2$ case''] \label{corr1r2}
Let $0\le \theta_1,\theta_2<2\pi$ be fixed and let $s,t \in \mathbb{R}$.
Suppose that
\begin{align*}
z=\bigg(r_1 + \frac{t}{\sqrt{2n\Delta Q(r_1)}}\bigg)e^{i\theta_1},\qquad
w=\bigg(r_2 + \frac{s}{\sqrt{2n\Delta Q(r_2)}}\bigg)e^{i\theta_2}.
\end{align*}
Then, as $n \to \infty$, we have
\begin{align}\label{balko}
K_n(z,w)=\sqrt{2\pi n} \cdot \calS^G(z,w;n)\cdot e^{im(\theta_1-\theta_2)} (r_1r_2)^{-x} e^{-\tfrac{1}{2}(t^2+s^2)}
+\bigO(\log^3 n),
\end{align}
uniformly for $s,t$ in any given compact subset of $\mathbb{R}$, where $m=\lfloor Bn/2\rfloor$ and $x=\{Bn/2\}$.
\end{cor}

Next we consider the case when both $z$ and $w$ are near the same boundary component, say the inner circle of radius $r_1$.
This situation is a bit more subtle, since the series in formula \eqref{wsg} diverges for $|z|=|w|=r_j, j=1,2$. However we can (for instance) sum it in the sense of Abel, by taking a limit as $|z|=|w|=r \in (r_{1},r_{2})$ and letting $r\to r_j$.
For $z=r_1 e^{i\theta_1}, w= r_1 e^{i\theta_2}$, $\theta_{1} \neq \theta_{2}$, this limit equals to
\begin{equation}\label{r1r1form}
\frac{1}{2\pi} \frac{\sqrt{\Delta Q(r_1)}}{r_1} r_1^{2x} \bigg(  \frac{1}{e^{i( \theta_1-\theta_2)}-1} +         \sum\limits_{\ell=0}^{\infty} \frac{e^{i(\theta_1-\theta_2)\ell }}{1 + a^{-1} \rho^{-(2\ell-2x+1) }}    -\sum\limits_{\ell=-\infty}^{-1} \frac{ e^{i(\theta_1-\theta_2)\ell}  }{1+a \rho^{2\ell -2x+1}}\bigg),
\end{equation}
where $x=\{Bn/2\}$, $\rho =r_1/r_2$ and  $a=\sqrt{\frac{\Delta Q(r_2)}{\Delta Q(r_1)}}$.
\begin{defn} \label{deffo} When $z=r_1 e^{i\theta_1}$, $w= r_1 e^{i\theta_2}$, $\theta_1\ne \theta_2$, we define the function $\calS^G(z,w;n)$ by the expression \eqref{r1r1form}. 
\end{defn}

\begin{thm}[``$r_1$-$r_1$ case''] \label{thmr1r1} Fix $0\le \theta_1,\theta_2<2\pi$ such that $\theta_1\ne \theta_2$. Also let $s,t \in \mathbb{R}$ and suppose that
	\begin{align*}
	z=\bigg(r_1 + \frac{t}{\sqrt{2n\Delta Q(r_1)}}\bigg)e^{i\theta_1},\qquad
	w=\bigg(r_1 + \frac{s}{\sqrt{2n\Delta Q(r_1)}}\bigg)e^{i\theta_2}.
	\end{align*}
	Then, as $n \to \infty$, we have,
	\begin{equation}\label{r1r1}
	K_n(z,w) = \sqrt{2\pi n} \, e^{-\tfrac{1}{2}(t^2 +s^2)}\, e^{i(\theta_1-\theta_2)m} r_1^{-2x} \calS^G(z,w;n)        +       \bigO(\log^5 n),
	\end{equation}
	uniformly for $s,t$ in any given compact subset of $\mathbb{R}$, where $\calS^G$ is given by Definition \ref{deffo}.
\end{thm}
\begin{rem}
The series
\begin{align*}
\Xi(x,\phi;\rho,a) := \sum\limits_{\ell=0}^{\infty} \frac{e^{i \phi \ell }}{1 + a^{-1} \rho^{-(2\ell-2x+1}) }    -\sum\limits_{\ell=-\infty}^{-1} \frac{ e^{i \phi \ell}  }{1+a \rho^{2\ell -2x+1}}
\end{align*}
appears in \eqref{r1r1form} with $\phi=\theta_{1}-\theta_{2}$. For $\phi=0$, one can relate this function with the $\log$-derivative of the Jacobi $\theta$ function. Indeed, using \cite[Lemma 3.28]{C} (see also Lemma \ref{lemma: Theta is expressible in terms of Jacobi} below), we obtain
\begin{align*}
\Xi(x,0;\rho,a) = \frac{(\log \theta)'\big( 2x+\frac{1}{2}+\frac{\log a}{2 \log(\rho^{-1})} ;  \frac{\pi i}{\log(\rho^{-1})} \big) + (4x-1) \log (\rho^{-1}) + \log a}{2\log(\rho^{-1})}.
\end{align*}
(We have not able to express $\Xi$ in terms of $\theta$ for general values of $\phi$.)
\end{rem}

\subsection{Comparison with follow-up work} \label{comparison} Our present effort is the first in a series of notes on disconnected droplets with soft-edge ring-shaped spectral gaps, including also two more recent works \cite{ACC,AC:Heine}. 
The common theme is fluctuations of smooth linear statistics, but this aside, each note contributes other aspects, of independent interest. A brief comparison is in order.

 Our present note derives and uses detailed information on asymptotics for correlation functions near the edge of the droplet. A corresponding analysis of correlation functions for a class of restricted Mittag-Leffler ensembles (with hard edges)
is carried out in yet another follow-up paper \cite{ACC2}.

In \cite{ACC} we prove large $n$ expansions of the free energy $\log Z_n$ with respect to perturbed potentials $\tilde{Q}$ as in \eqref{pertpot}. As a corollary, we obtain results for the cumulant generating function of fluctuations. We stress that in this approach, it is necessary to obtain a very precise expansion, up to and including the constant term, in order to rigorously study the asymptotic distribution of fluctuations. However, 
 this term is of great intrinsic interest, finding applications in field theories and so on; it 
 is conjectured to be related to the Polyakov-Alvarez formula in \cite{ZW}. In \cite{ACC} we explicitly obtain the large $n$ expansion for radially symmetric potentials and show that new
``displacement terms'' must be added to the terms expected from the Polyakov-Alvarez formula; cf.~also \cite{BKS,AC:Heine}.

The paper \cite{AC:Heine} considers some non-radially symmetric ensembles with ring-shaped spectral gaps and gives the natural counterpart of the fluctuation result for connected droplets from \cite{AM}; we turn to a brief comparison.


Consider a general smooth function $h$ supported in a small neighbourhood of $\overline{G}$. We can then decompose (cf. \cite{AC:Heine})
\begin{align}\label{ali}h(z)=\re g(z)+f_0(z)+c_1\omega(z)\end{align}
where $c_1\in \mathbb{R}$, all functions are smooth, $f_0$ is supported in a small neighbourhood of $\overline{G}$, and
\begin{itemize}
\item $g$ is analytic in $\mathbb{C}\setminus S$,
\item $f_0=0$ along the boundary $\d G$,
\item $\omega=0$ in a neighbourhood of $\{|z|\leq r_1\}$ and $\omega=1$ in a neighbourhood of $\{|z|\geq r_2\}$.
\end{itemize}
The idea in \cite{AC:Heine} is to study fluctuations of $f:=\re g+f_0$ and $c_1\omega$ separately. (Note that $g$ is constant on each of the two components of $\mathbb{C}\setminus \mathcal{N}_{G}$, where $\mathcal{N}_{G}$ is a small simply connected neighbourhood of $\overline{G}$ containing the support of $f_{0}$.)

It is shown in \cite[Theorem 1.6]{AC:Heine} that $\fluct_n f$ and $\fluct_n (c_1\omega)$ are asymptotically independent;  $\fluct_n f$ has an asymptotic Gaussian distribution while $\fluct_n(c_1\omega)$ has an oscillatory discrete normal distribution.

More precisely, it is shown in \cite[Theorem 1.5]{AC:Heine} that $\fluct_n f$ is asymptotically $N(\tilde{e}_f,\tilde{v}_f)$, where
$$\tilde{v}_f=\frac 1 4\int_\C|\nabla(f^S)|^2\, dA. $$

Here $f^S=\re g+f_0\cdot\1_S$ is the Poisson modification (well-known when $S$ is connected \cite{AM,ZW}). The fluctuations $c_1\fluct_n\omega$ have an asymptotic oscillatory discrete normal distribution by \cite[Theorem 1.3]{AC:Heine} (it corresponds to the
random variable $Y_n$ in Theorem \ref{flucco}).

\begin{ex}
Consider a radially symmetric function $\lfun$ as in Theorem \ref{flucco}. If $\lfun(r_1)=\lfun(r_2)$, we have $Y_n=0$, so by \eqref{cgf0},
 $\fluct_n\lambda$ is asymptotically normal as $n\to\infty$. Thus $\lambda^S=\lambda\cdot \1_S$ and a straightforward computation shows that $\tilde{v}_\lambda=v_\lambda+\hat{v}_\lambda$ (notation of Theorem \ref{flucco}). Observe that the ``jump term'' $\hat{v}_\lambda$ vanishes if $\lambda'(r_1)=\lambda'(r_2)=0$.
 \end{ex}

\subsection{Further comments} In the following paragraphs, we assume for simplicity that the potential $Q$ is $C^\infty$-smooth in a neighbourhood of the droplet.

The one- and two-point correlations studied above are not independent but are connected by the loop equation, see \cite{AM} or \cite[Section 6]{AC}. This relationship and our above results indicate that for $z$ near the boundary, there should be an asymptotic expansion in powers of $n^{\frac 1 2}$, of the form
$$R_n(z)=nb_1(z)+n^{\frac 1 2}b_{\frac 12}(z;n)+b_0(z;n)+n^{-\frac 1 2}b_{-\frac 12}(z;n)+\cdots,$$
where the most significant coefficients $b_1(z)$ and $b_{\frac 12}(z;n)$ (the ones that matter for our fluctuation computations) are computed in Theorems \ref{sr1}, \ref{sr2} and \ref{ur3} above. (The notation $b_{\frac 1 2}(z;n)$ indicates that $b_{\frac 1 2}$ oscillates in $n$ near a multi-connected gap. For higher correction terms $b_0,\ldots$ we do not know whether there is an oscillation.)
Somewhat related questions with respect to $\beta$-ensembles (at an outer boundary) are discussed in the references \cite{CFTW,CSA}.

 Interestingly, if $z$ is away from the boundary, the corresponding expansion does not involve ``half-integer coefficients'' $b_{\frac 12},b_{-\frac 12},\ldots$. In the bulk one has an expansion of the form $R_n=nb_1+b_0+n^{-1}b_{-1}+\cdots$ where the $b_j(z)$ obey a certain recursion, see \cite{A2} and the references therein. A similar expansion without half-integer coefficients is obtained in the exterior, unbounded component, see \cite[Theorem 6.4]{AC}. Thus the emergence of half-integer correction terms in the asymptotic expansion of the 1-point function seems to be a pure boundary phenomenon.

While much of the existing literature in dimension two has centered around connected droplets, there are exceptions. Strong asymptotics for orthogonal polynomials is worked out in \cite{BEG} for ``lemniscate potentials'' such as $Q(z)=|z|^{2d}-2\re(z^d)$ where $d\ge 2$ is an integer. This is of interest because the corresponding droplet $S=\{|z^d-1|\le d^{-\frac 1 2}\}$ consists of $d$ identical connected ``ponds'', which are organized in a non radially symmetric manner. Some results on asymptotics of the one- and two-point functions near the edge, for related ensembles, are found in the recent work \cite{BY2022}, which also contains further references on the topic.

The situation when $\Delta Q$ vanishes on some subset of $\d S$ requires a different analysis, and is likely to lead to new and interesting universality classes. (See \cite{BEG,BLY} for some works where $\Delta Q$ vanishes at a singular point on the boundary.)


The sources \cite{ACCL1,ACCL,C,Seo} study a very different kind of radially symmetric ensembles, where one starts with a connected droplet, for example a disc, and imposes a ``hard edge'' by redefining the potential to $+\infty$ in a preassigned ``forbidden region''.

We finally remark that in the setting of Laughlin's wave-function, our above setting is known as the ``plasma analogy''. Our results imply that, in the presence of spectral gaps, the plasma behaves quite differently from an electron gas in a magnetic field, as studied in e.g.~the recent work \cite{OLMSE}. (We are grateful to Jean-Marie St\'{e}phan and Benoit Estienne for this remark.)


\subsection{Plan of this paper} In Section \ref{sec2}, we prove two lemmas which are used throughout the paper:
Lemma \ref{lem11} on cumulant generating functions, and Lemma \ref{bock} on
an approximation formula for the 1-point function near $\overline{G}$.

In Section \ref{sec3} we prove Theorems \ref{sr1}--\ref{sr2} on the edge density in perturbed potential
$\tilde{Q}=Q-\frac {s\lfun}n$.

In Section \ref{sec4} we prove Theorem \ref{flucco} on the cumulant generating function of $\fluct_n\lfun$.

In Section \ref{sec5} we prove Theorems \ref{thmr1r2}-\ref{thmr1r1} on the asymptotics of $K_n(z,w)$ when the points $z,w$ are different and near the edge of the gap $G$.

In Section \ref{flureg} we formulate and prove Theorem \ref{gengap} on the asymptotic distribution
of $\fluct_n(\re g)$ where $g(z)$ is analytic in a ``gap'' $G$ in a droplet corresponding to a
fairly arbitrary potential $Q$. Our proof is modeled after the method of Ward identities in the paper \cite{AM}. We deduce Theorem \ref{thereg} as a corollary.

\section{Preparations} \label{sec2}

In this section, we supply a proof for Lemma \ref{lem11} on cumulant generating functions. We also provide suitable local approximations of the 1-point function near  $\overline{G}=\{r_1\le |z|\le r_2\}$, which are frequently used in the sequel.

\subsection{Perturbed potentials and the cumulant generating function} \label{fluco}
Recall that, given a smooth, bounded testing function $h$ we denote
$\tilde{Q}=\tilde{Q}_{n,h}=Q-\frac hn$
the perturbed potential.

We associate with $h$ two linear statistics, denoted $\trace_n h$ and $\fluct_n h$, by
$$\trace_n h=\sum_{j=1}^n h(z_j),\qquad \fluct_n h=\sum_{j=1}^n h(z_j)-n\int h\, d\sigma,$$
where $\{z_j\}_1^n$ is a random sample from the ensemble associated with the potential $\tilde{Q}$.

Below we will write $\tilde{\mathbb{E}}_{n,h}$ for the expectation with respect to $\tilde{Q}$ and $\mathbb{E}_n$ for the expectation with respect to $Q$.

We are interested in the cumulant generating function for $\fluct_n h$, which we denote
$F_{n,h}(t)=\log \mathbb{E}_n\, e^{t\fluct_n h}.$

\begin{proof}[Proof of Lemma \ref{lem11}] Recall that the Hamiltonians in external potentials $Q$ and $\tilde{Q}_{n,s h}$ are
$$H_n=\sum_{j\ne k}\log \frac 1 {|z_j-z_k|}+n\sum_{j=1}^n Q(z_j),\qquad \tilde{H}_{n,s h}=H_n-s\trace_n h.$$

The corresponding partition functions thus obey
\begin{align*}
& Z_{n} = \int e^{-H_n}\, dA_n,  & & \tilde{Z}_{n,s h}=\int e^{s\trace_n h}\, e^{-H_n}\, dA_n,
\end{align*}
and hence
$\frac {\tilde{Z}_{n,s h}}{Z_n}=\mathbb{E}_n(e^{s\trace_n h}).$ Since
\begin{align*}
d\tilde{\mathbb{P}}_{n,s h}&=\frac 1 {\tilde{Z}_n}e^{-\tilde{H}_n}dA_n=\frac {Z_n}{\tilde{Z}_n}e^{s\trace_n h}d\mathbb{P}_n,
\end{align*}
we now see that
\begin{align*}
\tilde{\mathbb{E}}_{n,s h}(\fluct_n f)=\frac {\mathbb{E}_n(e^{s\trace_n h}\fluct_n f)}{\mathbb{E}_n(e^{s\trace_n h})}=\frac {\mathbb{E}_n(e^{s\fluct_n h}\fluct_n f)}{\mathbb{E}_n(e^{s\fluct_n h})}.
\end{align*}
Taking $f=h$ we deduce that
$$\tilde{\mathbb{E}}_{n,s h}(\fluct_n h)=\frac d {ds}\log \mathbb{E}_n (e^{s\fluct_n h})= \frac{d}{ds}F_{n,h}(s).$$
To finish the proof it suffices to note that $F_n(0)=0$ and integrate in $s$ from $0$ to $t$.
\end{proof}

\subsection{Approximate 1-point function for the gap}
Consider a perturbed potential
\begin{equation}\label{pertur}\tilde{Q}=Q-\frac {s\lfun}n\end{equation} where $\lfun(z)=\lfun(|z|)$ is a smooth, radially symmetric
real-valued function on $\C$, supported in a small neighbourhood of $\overline{G}$.

We write $\tilde{\calW}_n$ for the subspace of $L^2(\mathbb{C},dA)$ consisting of all weighted polynomials
$$p(z)=P(z)e^{-n\tilde{Q}(z)/2}$$ where $P$ is a holomorphic polynomial of degree at most $n-1$; we
write
\begin{equation}\label{fullkernel}
\tilde{p}_j(z)=z^je^{-n\tilde{Q}(z)/2}\qquad \text{and}\qquad \tilde{K}_n(z,w)=\sum_{j=0}^{n-1}\frac {\tilde{p}_j(z)\overline{\tilde{p}_j(w)}}{\|\tilde{p}_j\|^2}.
\end{equation}

Evidently $\tilde{K}_n$ is the reproducing kernel of $\tilde{\calW}_n$, and the $1$-point function in potential $\tilde{Q}$ is $\tilde{R}_n(z):=\tilde{K}_n(z,z).$

\begin{defn} Given an integer $j$, $0\le j\le n-1$, we set
\begin{equation}\label{baq}\tau=\tau(j)=\frac j n,\qquad \delta_n = \frac{\log n}{\sqrt{n}},\end{equation}
and define the \emph{local 1-point approximation} for the gap $G$ by
\begin{equation}\label{Kgap}
\tilde{R}_n^G(z):=\sum_{|\tau-\frac B 2|\le C \delta_n}\frac {|\tilde{p}_j(z)|^2}{\|\tilde{p}_j\|^2},
\end{equation}
where $C=C(Q)$ is a large enough, positive constant, depending only on the potential $Q$.
\end{defn}

We also introduce the following notation for the $\delta_n$-neighbourhood of the gap:
\begin{equation}\label{dnbh}N(G,\delta_n)=\{z\,:\, \dist(z,G)<\delta_n\}.
\end{equation}

\begin{lem} \label{bock}
If the constant $C=C(Q)$ in \eqref{Kgap} is chosen large enough, we have as $n\to\infty$
\begin{equation}\label{kung}\tilde{R}_n(z)=\tilde{R}_n^G(z)\cdot (1+\bigO(e^{-c n \delta_{n}^{2}})),
\end{equation}
where $c$ is a positive constant, uniformly for $z\in N(G,\delta_n)$.
\end{lem}

\begin{rem} For $z$ in the $\delta_n$-neighbourhood $N(U,\delta_n)$ of the unbounded component of $\C\setminus S$,
the approximation should be replaced by a ``one-sided'' sum:
$\tilde{R}_n^U(z)=\sum_{1-C \delta_n\le \tau<1}\frac {|\tilde{p}_j(z)|^2}{\|\tilde{p}_j\|^2}.$
The approximation $\tilde{R}_n(z)=\tilde{R}_n^U(z)\cdot (1+\bigO(e^{-c n \delta_{n}^{2}}))$ is implicit in several works, e.g.~ \cite[Section 3]{AC}.
\end{rem}

\begin{proof}[Proof of Lemma \ref{bock}]

Fix a point $z$ with $\dist(z,G)\le \delta_n$ and write
$$\tilde{R}_n(z)=\tilde{R}_n^G(z)+\eps_n(z),\qquad
\eps_n(z)=\sum_{|\tau-\frac B2|> C\delta_n}\frac {|\tilde{p}_j(z)|^2}{\|\tilde{p}_j\|^2}.
$$

Using \eqref{wassa} (with $p_{j}$ and $Q$ replaced by $\tilde{p}_{j}$ and $\tilde{Q}$) we may infer that if $|\tau(j)-\tfrac B2|> C\delta_n$, and if the constant $C$ is large enough,
then
\begin{equation}\label{pjn}
\frac {|\tilde{p}_j(z)|}{\|\tilde{p}_j\|}\lesssim e^{-c n \delta_{n}^{2}}e^{-\frac 1 2 n(Q-\check{Q})(z)},
\end{equation}
where $\check{Q}$ is the obstacle function with respect to the (unperturbed) potential $Q$. The proof of \eqref{pjn} proceeds essentially in the same way as in \cite[Lemma 3.10]{AC} (we omit details).

Hence also
$$\eps_n(z)\lesssim e^{-c n \delta_{n}^{2}} e^{-n(Q-\check{Q})(z)}.$$

We next fix the integer $j$ closest to $\frac{Bn}{2}$ and observe using
\begin{align*}
\check{Q}-V = \begin{cases}
0, & \mbox{in } G, \\
Q-V>0, & \mbox{in } N_G\setminus \overline{G},
\end{cases}
\end{align*}
 that \eqref{wassa} implies that $|\tilde{p}_j(z)|^2 \gtrsim e^{-nA}e^{-n(Q-\check{Q})(z)}$ for all $z \in N(G,\delta_{n})$.
Moreover, Lemma \ref{neglo} together with \eqref{wassa} implies that
$\|\tilde{p}_j\|^2 \lesssim e^{-nA}.$
Hence
$\frac {|\tilde{p}_j(z)|}{\|\tilde{p}_j\|}\gtrsim e^{-\frac n2(Q-\check{Q})(z)}.$

As a consequence, $$\tilde{R}_n^G(z)\gtrsim e^{-n(Q-\check{Q})(z)}.$$

All in all, choosing $c>0$ somewhat smaller, we obtain
$\eps_n(z)\lesssim e^{-c n \delta_{n}^{2}} \tilde{R}_n^G(z)$ when
$\dist(z,G)\le \delta_n$, as desired.
\end{proof}

\begin{rem} Our choice of $\delta_n$ in \eqref{baq} is somewhat arbitrary. Modulo some minor changes, it is possible to replace $\delta_n$ by the slightly smaller quantity $M\sqrt{\frac {\log n} n}$, for large enough $M$. (This comment applies to virtually all our arguments below.)
\end{rem}

\section{Asymptotics for the edge density} \label{sec3}
In this section, we prove Theorems \ref{sr1} and \ref{sr2}.

\subsection{Preliminary computations} \label{prelcom}
To set things up, fix a radially symmetric potential $Q(z)=q(|z|)$, giving rise to a gap $G=\{r_1<|z|<r_2\}$. Also fix a $C^2$-smooth radially symmetric function $\lfun(z)$, supported in a small neighbourhood of $\overline{G}$.

We finally fix a real parameter $s$ and consider the perturbed potential \eqref{pertur}.

Our goal is to estimate the $1$-point function $\tilde{R}_n(z)=\tilde{R}_{n,s\lfun}(z)$ for $z$ in a $\delta_n$-neighbourhood of the boundary $\d G$.

For such $z$, we have (by Lemma \ref{bock}) the approximation $\tilde{R}_n(z)=\tilde{R}_n^G(z)\cdot (1+\bigO(e^{-c n \delta_{n}^{2}}))$ where
\begin{align}\label{RtildeG def}
\tilde{R}_n^G(z):=\sum_{|\tau-\frac B2|\le C \delta_n}\frac {|\tilde{p}_j(z)|^2}{\|\tilde{p}_j\|^2},
\end{align}
where we remind that
$$\tau=\tau(j)=\frac j n,\qquad \tilde{p}_j(z)=z^je^{-\frac 1 2 n\tilde{Q}(z)},\qquad \delta_n=\frac{\log n}{\sqrt{n}}.$$

We will throughout the rest of this section assume that $j$ is an integer such that
\begin{equation}\label{taub}
|\tau-\tfrac B 2|\le C \delta_n.
\end{equation}

We will deduce detailed information about $\tilde{p}_j$, including an analysis of peak-points as well as of the squared $L^2$-norm
\begin{equation}\label{sqndef}I(\tau):=\|\tilde{p}_j\|^2.\end{equation}

It is convenient to fix some notation.

\begin{defn} In the following we will write interchangeably
\begin{align}\label{def of gj}
g_j(r)=g_\tau(r):=q(r)-2\tau\log r,
\qquad (\tau=\tau(j)),
\end{align}
and we set
\begin{align}\label{def of h}
h(r)&:=r e^{s\lfun(r)}.
\end{align}
\end{defn}

It follows that
\begin{align}\label{sadel}
|\tilde{p}_j(z)|^2=e^{-n  g_j(r)}e^{s \lfun(r)}, \; (r:=|z|), \qquad \quad
I(\tau)=  2\int_{0}^{\infty} h(r)\, e^{-n g_{j} (r) }\, dr.
\end{align}

We will use a saddle point analysis to estimate $I(\tau)$.
The saddle point equation is $g_j'(r)=0$, i.e.,
\begin{align}\label{lol1}
q'(r) = \frac{2\tau}{r}.
\end{align}

For $\tau = \frac{B}{2}$, we know from \eqref{obs2} that $r=r_{1}$ and $r=r_{2}$ are two solutions to this equation.

For $\tau$ close to $\frac{B}{2}$ (more precisely: satisfying \eqref{taub}),
equation \eqref{lol1} can potentially have many positive solutions. The two relevant solutions for us are the ones that are close to $r_{1}$ and $r_{2}$ respectively. We denote these solutions by $r_{1,j}$ and $r_{2,j}$. These are the two ``peak-points'' of $|\tilde{p}_j|$ near the gap, indicated in Figure \ref{pic2}.

In view of Lemma \ref{neglo} the contribution to the integral \eqref{sadel} from other possible peaks is negligible, and we can safely focus on $r_{1,j},r_{2,j}$.

A straightforward computation using \eqref{obs2} (left for the careful reader) gives
\begin{align}
\label{droplet_tau}
 r_{k,j} & = r_{k} + \frac{2(\tau-\frac{B}{2})}{(rq')'(r_{k})} - \frac{2(rq')''(r_{k})}{(rq')'(r_{k})^{3}} (\tau-\tfrac{B}{2})^{2} + \bigO((\tau-\tfrac{B}{2})^3), \qquad \mbox{as } \tau \to \tfrac{B}{2}, \; k=1,2.
\end{align}

Evaluating the integral by Laplace's method (see e.g. \cite[Theorem 15.2.5]{S})
we now find the following asymptotic for \eqref{sadel}, as $n\to\infty$,
\begin{equation} \label{lol19}
I(\tau) = \sqrt{\frac{2\pi}{n}} \cdot ( c_1(j,n) + c_2(j,n) )\cdot (1+ \bigO(n^{-1})),
\end{equation}
where the $\bigO$-constant is uniform for the range of $j$ in \eqref{taub} and
\begin{equation}\label{vupp}
c_k(j,n) =   \frac{2h(r_{k,j}) }{ \sqrt{g_{j}''(r_{k,j})}}e^{-n g_{j}(r_{k,j})}, \qquad k=1,2.
\end{equation}

\smallskip

We next fix $k$ (either $1$ or $2$) as well as
a real parameter $t$ with $|t|\le \delta_{n}\sqrt{n}$ and put
\begin{equation}\label{upp}|z|=r =: r_k +\frac{t}{\sqrt{2n\Delta Q(r_k)}}.\end{equation}

It is convenient to introduce two functions $V_{\tau}^1$ and $V_\tau^2$ of $r$ by
\begin{equation}\label{vkdef}
V^k_{\tau}(r)=q(r_{k,j}) + 2\tau \log \bigg( \frac{r}{r_{k,j}} \bigg)= 2 \tau \log r + g_{j}(r_{k,j}),\qquad k=1,2.
\end{equation}
(We also extend the above definitions to all $z \in \mathbb{C}$ by $V_{\tau}^{k}(z):=V_{\tau}^{k}(|z|)$, $k=1,2$.)

We aim to use the identity
\begin{equation}\label{mpoint}\frac {|\tilde{p}_j(z)|^2}{c_k(j,n)}
=\frac {\sqrt{g_j''(r_{k,j})}}{2h(r_{k,j})}e^{n(V_\tau^k(r)-q(r))}e^{s\lfun(r)}.
\end{equation}

We shall require rather detailed asymptotics of $V_\tau(r)$ as $n\to\infty$, at the point \eqref{upp}. Here the parameter $\tau=\tau(j)$ is assumed to satisfy \eqref{taub}.

In the notation of \eqref{upp}, our assumption $|t| \leq \delta_{n}\sqrt{n}$ gives that $|r-r_k| = \bigO(\delta_{n})$. Also, from $\eqref{droplet_tau}$ and \eqref{taub} we get $r_{k,j}-r_k = \bigO(\delta_n)$, where the $\bigO$-constant is uniform in $j$. Hence $|r-r_{k,j}| = \bigO(\delta_{n})$ uniformly for $j$ in the range \eqref{taub}.

By construction, $V_{\tau}^k(r)$ agrees with $q(r)$ up to first order derivatives at $r=r_{k,j}$, and a Taylor expansion  gives
\begin{multline}
\label{harmonic_ext}
V_{\tau}^k(r)-q(r)= -\bigg(\frac{q'(r_{k,j})}{r_{k,j}} +q''(r_{k,j}) \bigg)\cdot \frac{ ( r-r_{k,j})^2}{2}\\
+ \bigg(  2\frac{q'(r_{k,j})}{r_{k,j}^2} -q^{(3)}(r_{k,j}) \bigg) \cdot \frac{( r-r_{k,j})^3}{6} + \bigO(\delta_n^{4}).
\end{multline}

Using the identities
\begin{align*}
 4\Delta Q(r) =\frac{q'(r)}{r} + q''(r),\qquad
 4\partial_r \Delta Q(r) &= \frac{q''(r)}{r} - \frac{q'(r)}{r^2} + q^{(3)}(r),
 \end{align*} we now rewrite \eqref{harmonic_ext} in the following way
 \begin{multline}\label{harm_ext2}
 V_{\tau}^k(r)-q(r) = -2 \Delta Q(r_{k,j})(r-r_{k,j})^2  + 2\bigg(  \frac{\Delta Q(r_{k,j})}{r_{k,j}}
 -\radlap Q(r_{k,j}) \bigg) \cdot \frac{(r-r_{k,j})^3 }{3} + \bigO(\delta_n^{4}).
\end{multline}

Next we Taylor expand the coefficients in \eqref{harm_ext2} about $r_k$ and obtain the approximation
\begin{multline}\label{rkgen}
 V_{\tau}^k(r)-q(r) = -2 \Delta Q(r_k)(r-r_{k,j})^2 +2\radlap Q(r_k) (r_{k}-r_{k,j})(r-r_{k,j})^2  \\
 + 2\bigg(  \frac{\Delta Q(r_{k})}{r_{k}} -\radlap Q(r_{k}) \bigg) \cdot \frac{(r-r_{k,j})^3 }{3} + \bigO(\delta_n^{4}).
 \end{multline}

In addition we will use the following three approximations (each, for $k=1,2$)
\begin{align}
\tfrac{1}{2}\sqrt{g''_j(r_{k,j})} &=\sqrt{\Delta Q(r_k)}\cdot \bigg(1+\frac{\radlap Q(r_k)}{2\Delta Q(r_k)} (r_{k,j}-r_k) + \bigO(\delta_n^{2}) \bigg),\\
\frac{1}{r_{k,j}} &= \frac{1}{r_k} \cdot \bigg(    1- \frac{r_{k,j}-r_k}{r_k}   +  \bigO(\delta_n^{2})\bigg),\\
 e^{s(\lfun(r)-\lfun(r_{k,j}))}
 &= 1 - s\lfun'(r_k)(r_{k,j}-r) + \bigO(\delta_n^{2}).
\end{align}

Recall the notation
\begin{equation}\label{nots}m=\lfloor Bn/2\rfloor,\qquad x=\{Bn/2\}\end{equation}
for the integer and fractional parts of $Bn/2$, where $B$ is given in \eqref{obs3}.
We now write
$$\ell:=j-m, \qquad(|\ell|\le Cn\delta_n),$$ and
\begin{equation}\label{ykl}y=y_{k,\ell}: = t- \frac{\ell/\sqrt{n}}{r_k \sqrt{2\Delta Q(r_k)}},\qquad (k=1,2).\end{equation}

As a preliminary remark, we may note that $y_{k,\ell}$ is related to $r_{k,j}-r$ via
\begin{equation}\label{Jak}
r_{k,j}-r=\frac {-y_{k,\ell}}{\sqrt{2n\Delta Q(r_k)}} + \bigO(\delta_{n}^{2}),\qquad (k=1,2).
\end{equation}

We will have frequent use of the following lemma.

\begin{lem}\label{Lemma:asymp pj over ck} Let $y=y_{k,\ell}$ be given by \eqref{ykl} and $r$ near $r_k$ by \eqref{upp}. Then
\begin{equation}\frac {|\tilde{p}_j(r)|^2}{c_k(j,n)}=
 \frac{\sqrt{\Delta Q(r_k})}{r_k} e^{-y^2}  \cdot \bigg(1 + \frac {f_k(y)}{\sqrt{n}} + \bigO(n\delta_n^{4}) \bigg)
\end{equation}
where the function $f_k(y)$ is given by the equation
\begin{multline}\label{simplef}
\sqrt{2\Delta Q(r_k)}\, f_k(y)= sy\lambda'(r_k)+ \frac{1}{r_k} (- t+y -2xy-yt^2+2ty^2- \tfrac{2}{3}y^3) \\
+ \frac{\radlap Q(r_k)}{\Delta Q(r_k)} (  ty^2-t^2y-\tfrac{1}{3}y^3 + \tfrac{1}{2}t -\tfrac{1}{2}y  ).
\end{multline}
\end{lem}

The proof is based on combining \eqref{mpoint} with the above approximations, via a somewhat involved but straightforward computation; we omit details.

\smallskip

For ease of reference, we finally recall the following fact about second order Riemann sums, which is taken from \cite[Lemma 3.4]{C}. (We refer to \cite{T} for further details about the related Euler-Maclaurin formula.)

\begin{lem}\label{LemmaRS} Let $A,a_0,B,b_0$ be bounded functions of $n$ such that
$$a_n:=An+a_0,\qquad b_n=Bn+b_0$$
are integers. Assume that $B-A$ is positive and bounded away from zero. Let $f$ be a function independent of $n$ which is, for all $n$, defined and $C^4$-smooth in the interval from $\min\{\frac {a_n}n,A\}$ to $\max\{\frac {b_n} n,B\}.$ Then, as $n\to+\infty$,
\begin{align*}\sum_{j=a_n}^{b_n}f(\frac j n)&=n\int_A^Bf(x)\, dx+\frac {(1-2a_0)f(A)+(1+2b_0)f(B)}2+\frac {(-1+6a_0-6a_0^2)f'(A)+(1+6b_0+6b_0^2)f'(B)}{12n}\\
&+\frac {(-a_0+3a_0^2-2a_0^3)f''(A)+(b_0+3b_0^2+2b_0^3)f''(B)}{12n^2}+\bigO\left(\frac {\mathfrak{m}_{A,n}(f''')+\mathfrak{m}_{B,n}(f''')}{n^3}+\sum_{j=a_n}^{b_n-1}\frac {\mathfrak{m}_{j,n}(f^{\mathrm{iv}})}{n^4}\right).
\end{align*}
Here, $\mathfrak{m}_{A,n}(g)=\max\{|g(x)|\,;\, \min\{\frac {a_n} n,A\}\le x\le \max\{\frac {a_n} n,A\}\}$, $\mathfrak{m}_{B,n}(g)=\max\{|g(x)|\,;\, \min\{\frac {b_n} n,B\}\le x\le \max\{\frac {b_n} n,B\}\}$, and for $a_n\le j<b_{n}-1$
$\mathfrak{m}_{j,n}(g)=\max\{|g(x)|\,;\, \frac {j} n\le x\le \frac {j+1} n\}$.
\end{lem}

\smallskip

We are now equipped with the tools to study the one-point function $\tilde{R}_n(z)$ near the edge of the gap $G$.

\subsection{The $r_1$-case} Let us now assume that $r=|z|$ is close to $r_1$, and specifically
\begin{equation}\label{r1case}
r=r_1+\frac t {\sqrt{2n\Delta Q(r_1)}},\qquad (t\in\R,\,|t|\le \delta_{n}\sqrt{n}).
\end{equation}
Using \eqref{RtildeG def}, \eqref{sadel} and the estimate \eqref{lol19}, we
decompose the approximate one-point function $\tilde{R}_n^G(z)$ as
\begin{align}\label{splitting of RtG}
\tilde{R}_n^G(z) = \sqrt{\frac{n}{2\pi}}(S_1 + S_2 + S_3)\cdot (1+\bigO(n^{-1})),
\end{align}
where
\begin{align}
&S_1 = \sum\limits_{j= m-Cn\delta_n}^{m-1} \frac{|z|^{2j}}{c_1(j,n)} e^{-nQ(z)+s\lfun(z)},
\\
&S_2 = \sum\limits_{j=m}^{m+Cn\delta_n} \frac{|z|^{2j}}{c_1(j,n)+c_2(j,n)} e^{-nQ(z)+s\lfun(z)},
\\
&S_3=-\sum\limits_{j=m-Cn\delta_n}^{m-1} \frac{c_2(j,n)}{c_1(j,n)}\frac{|z|^{2j}}{c_1(j,n)+c_2(j,n)}e^{-nQ(z)+s\lfun(z)} .
\end{align}

We begin by analyzing $S_1$. Write

\begin{multline*}
S_1 = \sum\limits_{j= m-Cn\delta_n}^{m-1} \frac{|z|^{2j}}{c_1(j,n)} e^{-nQ(z)+s\lfun(z)} =\sum\limits_{\ell=-Cn\delta_n}^{-1} \frac{\sqrt{\Delta Q(r_1)}}{r_1} \bigg(1 + f_1(y_{1,\ell}) \frac{1}{\sqrt{n}}+\bigO(n\delta_n^{4}) \bigg)\, e^{-y_{1,\ell}^2},
\end{multline*}
with $f_1(y)$ given in \eqref{simplef} and $y_{1,\ell}$ in \eqref{ykl}.

Using a second order Riemann sum approximation (Lemma \ref{LemmaRS}
with $n$ replaced by $\sqrt{n}$ and with $A=-C\delta_{n}\sqrt{n}$, $a_{0}=\lceil -C n \delta_{n} \rceil + C n \delta_{n}$, $B=0$, $b_{0}=-1$), we have
\begin{multline}\label{comp0}
\sum\limits_{\ell=-Cn\delta_n}^{-1} \frac{\sqrt{\Delta Q(r_1)}}{r_1}e^{-y_{1,\ell}^2}= \sqrt{n} \int\limits_{t}^{C\delta_{n}\sqrt{n}} \sqrt{2}\Delta Q(r_1) e^{-y^2} dy - \tfrac{1}{2}e^{-t^2} \frac{\sqrt{\Delta Q(r_1)}}{r_1} + \bigO(n^{-\frac 12}) \\
= \sqrt{n} \int\limits_{t}^{\infty} \sqrt{2}\Delta Q(r_1) e^{-y^2} dy  - \tfrac{1}{2}e^{-t^2} \frac{\sqrt{\Delta Q(r_1)}}{r_1} + \bigO(n^{-\frac 12}), \qquad \mbox{as } n \to \infty,
\end{multline}
and similarly
\begin{multline}\label{comp1}
\frac 1 {\sqrt{n}}\sum\limits_{\ell=-Cn\delta_n}^{-1} \frac{\sqrt{\Delta Q(r_1)}}{r_1}f_1(y_{1,\ell})e^{-y_{1,\ell}^2}=  \int\limits_{t}^{C\delta_{n}\sqrt{n}} \sqrt{2}\Delta Q (r_1) f_1(y)  e^{-y^2} dy + \bigO(n^{-\frac 12})\\
=\int\limits_{t}^{\infty} \sqrt{2}\Delta Q (r_1) f_1(y)  e^{-y^2} dy +\bigO(n^{-\frac 12}) .
\end{multline}

Inserting the expression for $f_1(y)$ in \eqref{simplef}, we compute
\begin{multline}
\int\limits_{t}^{\infty} \sqrt{2}\Delta Q (r_1) f_1(y) e^{-y^2} dy = \frac{\sqrt{\Delta Q(r_1)}}{r_1} e^{-t^2} \Big(  \frac s {2}r_1\lfun'(r_1) +  \frac{t^2+1}{6}
-x  \Big)\\
+ \frac{\nlap Q(r_1)}{\sqrt{\Delta Q(r_1)}}\Big(\frac{1}{2} \sqrt{\pi} t\erfc(t) -\frac{1}{12} e^{-t^2}(2t^2+5)\Big).
\end{multline}

In conclusion we have shown that
\begin{multline}\label{S1asy}
S_1 =  \sqrt{n} \int\limits_{t}^{\infty} \sqrt{2}\Delta Q(r_1) e^{-y^2} dy +   \frac{\sqrt{\Delta Q(r_1)}}{r_1}e^{-t^2}\Big( \frac{s}{2}r_1\lfun'(r_1)  + \tfrac{1}{6} (t^2-2)
- x\Big)\\
+ \frac{\nlap Q(r_1)}{\sqrt{\Delta Q(r_1)}}\Big(\frac{1}{2} \sqrt{\pi} t\erfc(t) -\frac{1}{12} e^{-t^2}(2t^2+5)\Big) + \bigg(\frac {(\delta_{n}\sqrt{n})^{4}}{\sqrt{n}}\bigg), \qquad \mbox{as } n \to \infty.
\end{multline}

This finishes our analysis of $S_1$, and we may turn to the sums $S_2$ and $S_3$.

\smallskip

We begin by rewriting $S_2$ as
\begin{align}\label{S2f}
& S_2 = \sum\limits_{j=m}^{m+Cn\delta_n} \frac{c_1(j,n)^{-1}|z|^{2j}}{1+\tfrac{c_2(j,n)}{c_1(j,n)}} e^{-nQ(z)+s\lfun(z)} =  \sum\limits_{j=m}^{m+Cn\delta_n}  \frac{e^{n(V_{\tau}^1(r) -q(r)) }}{1+\tfrac{c_2(j,n)}{c_1(j,n)}} \frac{\sqrt{g_j''(r_{1,j})}}{2h(r_{1,j})} e^{s\lfun(r)}.
\end{align}

Inserting the expression \eqref{vupp} for $c_k(j,n)$ we find
\begin{equation}\label{qu1}
\frac{c_2(j,n)}{c_1(j,n)} = \sqrt{\frac {g_j''(r_{1,j})}{g_j''(r_{2,j})}}\frac {h(r_{2,j})}
{h(r_{1,j})}e^{n(g_j(r_{1,j})-g_j(r_{2,j}))}.
\end{equation}

This formula suggests that we should estimate the difference $g_j(r_{2,j})-g_j(r_{1,j})$. We shall deduce a little more than we need at the present stage.

\begin{lem} \label{gjcomp} As $n\to\infty$, with $\tau=j/n$, we have
$$g_j(r_{2,j})-g_j(r_{1,j})=-(2\tau-B)\log\frac {r_2}{r_1}+\mathcal{C}
(\tau-\tfrac B 2)^2+\bigO((\tau-\tfrac B 2)^3),$$
where
$$\mathcal{C}=\frac 1 2\bigg(\frac 1 {r_1^2\Delta Q(r_1)}-\frac 1 {r_2^2\Delta Q(r_2)}\bigg).$$
\end{lem}

\begin{proof} Using \eqref{def of gj} and \eqref{obs3}, we find
$$g_j(r_2)-g_j(r_1)=-(2\tau-B)\log \frac {r_2}{r_1},$$
and therefore
$$g_j(r_{2,j})-g_j(r_{1,j})=(g_j(r_{2,j})-g_j(r_2))+(g_j(r_1)-g_j(r_{1,j}))-(2\tau-B)\log\frac {r_2}{r_1}.$$
Recalling that $g_j(r)=q(r)-2\tau\log r$ and $V^k_\tau(r)=q(r_{k,j})+2\tau\log\frac r {r_{k,j}}$, we now write
$$g_j(r_{2,j})-g_j(r_{1,j})=(V_\tau^2(r_2)-q(r_2))+(q(r_1)-V_\tau^1(r_1))
-(2\tau-B)\log\frac {r_2}{r_1}.$$
Recalling \eqref{harmonic_ext} and \eqref{droplet_tau}, we conclude that
$$q(r_k)-V^k_\tau(r_k)=\frac {(\tau-\tfrac B 2)^2}{2r_k^2\Delta Q(r_k)}+\bigO((\tau-\tfrac{B}{2})^3).$$

The proof of the lemma is complete.
\end{proof}

As before, write $\ell=j-m$ where $m=\lfloor Bn/2\rfloor$ and $x=Bn/2-m$. By \eqref{taub}, $|\ell|\le Cn\delta_{n} +1$.

The following lemma will come in handy.

\begin{lem}\label{c1c2lem} As $n \to \infty$, we have
\begin{equation}\label{c21}
\frac {c_2(j,n)}{c_1(j,n)}=e^{s(\lfun(r_{2})-\lfun(r_{1}))}\sqrt{\frac
{\Delta Q(r_1)}{\Delta Q(r_2)}}\Big(\frac {r_2}{r_1}\Big)^{2\ell-2x+1}\cdot e^{-\mathcal{C}\ell^2/n}\cdot
\bigg(1+\bigO\Big(
\frac {(\delta_{n}\sqrt{n})^{3}}{\sqrt{n}}\Big) \bigg).
\end{equation}
Moreover, the squared norm $I(\tau)$ in \eqref{sqndef} obeys
\begin{multline}\label{iitau}
I(\tau)=\sqrt{\frac {2\pi} n}(r_1r_2)^{x-\ell}\bigg(\frac {r_1^{2\ell-2x+1}}{\sqrt{\Delta Q(r_1)}}e^{s\lfun(r_1)}e^{\frac{\mathcal{C}}{2}\ell^2/n}+\frac {r_2^{2\ell-2x+1}}{\sqrt{\Delta Q(r_2)}}e^{s\lfun(r_2)}e^{-\frac{\mathcal{C}}{2}\ell^2/n}\bigg)\\
\times e^{-\frac n 2(g_j(r_{1,j})+g_j(r_{2,j}))}\cdot \bigg(1+\bigO\Big(
\frac {(\delta_{n}\sqrt{n})^{3}}{\sqrt{n}}\Big)\bigg), \qquad \mbox{as } n \to \infty.
\end{multline}
Both \eqref{c21} and \eqref{iitau} are valid uniformly for $j$ in the range \eqref{taub}.
\end{lem}

\begin{proof} The first estimate \eqref{c21} is immediate by using Lemma \ref{gjcomp}
 in \eqref{qu1}.

To prove \eqref{iitau}, we recall from \eqref{lol19}, \eqref{def of h} and \eqref{vupp} that
\begin{multline*}
I(\tau)=\sqrt{\frac {2\pi} n}\bigg(\frac {2r_{1,j}e^{s\lfun(r_{1,j})}e^{\frac n 2(g_j(r_{2,j})-g_j(r_{1,j}))}} {\sqrt{g_j''(r_{1,j})}}+
\frac {2r_{2,j}e^{s\lfun(r_{2,j})}e^{\frac n 2(g_j(r_{1,j})-g_j(r_{2,j}))}} {\sqrt{g_j''(r_{2,j})}}\bigg) \\
\times e^{-\frac n 2(g_j(r_{1,j})+g_j(r_{2,j}))} \cdot (1+ \bigO(n^{-1})), \qquad \mbox{as } n \to \infty.
\end{multline*}
The relation now follows by inserting the asymptotic in Lemma \ref{gjcomp} and simplifying.
\end{proof}

Next, in view of \eqref{rkgen} and \eqref{r1case}, we have
\begin{align}\label{lol2}
n(V_{\tau}^1(r)-q(r))= -\bigg(t-\frac{\ell/\sqrt{n}}{r_{1}\sqrt{2 \Delta Q(r_{1})}}\bigg)^2  +\bigO\bigg(\frac{(\delta_{n}\sqrt{n})^{3}}{\sqrt{n}}\bigg),\qquad (n\to\infty).
\end{align}

Inserting \eqref{c21} and \eqref{lol2} in \eqref{S2f} and denoting
\begin{equation}\label{trecal}\rho:=r_1/r_2,\qquad  a=a(s):=e^{s(\lfun(r_1)-\lfun(r_2))}\sqrt{\frac{\Delta Q(r_2)}{\Delta Q(r_1)}},\qquad  x:=\{\frac {nB}2\},\end{equation}
(note that $a(0)$ in \eqref{trecal} corresponds to $a$ in \eqref{r1r1form}) we get after some simplification that

\begin{equation}\label{s2prel}
S_2 = \frac{\sqrt{\Delta Q(r_1)}}{r_1}\sum\limits_{\ell=0}^{Cn\delta_n} e^{-(t-\frac{\ell/\sqrt{n}}{r_{1}\sqrt{2\Delta Q(r_{1})}})^2} \frac{a\rho^{2(\ell+\frac{1}{2}-x) }e^{\mathcal{C}\ell^2/n}}{1+a\rho^{2(\ell+\frac{1}{2}-x)}e^{\mathcal{C}\ell^2/n}}\cdot \bigg(1+\bigO\Big(\frac {(\delta_{n}\sqrt{n})^{3}}{\sqrt{n}}\Big)\bigg).
\end{equation}
We can analyze the above sums in three steps, in the same way as in \cite[Lemma 3.6]{C}. Since $\rho<1$, the summand becomes exponentially small as $\ell$ gets large. Hence, as a first step, we can replace $\sum_{\ell=0}^{Cn\delta_n}$ by $\sum_{\ell=0}^{C\log n}$ at the cost of an error of order $\bigO(n^{-100})$ (we can ensure that by choosing $C$ sufficiently large). Second, note that $e^{-(t-\frac{y}{r_{1}\sqrt{2\Delta Q(r_{1})}})^2}$ and $e^{\mathcal{C}y^{2}}$ are analytic functions of $y$ in a neighborhood of $0$, and therefore
\begin{align*}
e^{-(t-\frac{\ell/\sqrt{n}}{r_{1}\sqrt{2\Delta Q(r_{1})}})^2} = e^{-t^{2}}\bigg(1+\bigO\bigg( \frac{\log n}{\sqrt{n}} \bigg)\bigg), \qquad e^{\mathcal{C}\ell^2/n} = 1 + \bigO\bigg( \frac{\log n}{\sqrt{n}} \bigg),
\end{align*}
as $n \to \infty$ uniformly for $0 \leq \ell \leq C \log n$. Substituting these asymptotics yields
\begin{align*}
S_2 = e^{-t^{2}}\frac{\sqrt{\Delta Q(r_1)}}{r_1}\sum\limits_{\ell=0}^{C\log n} \frac{a\rho^{2(\ell+\frac{1}{2}-x) }}{1+a\rho^{2(\ell+\frac{1}{2}-x)}}\cdot \bigg(1+\bigO\Big(\frac {(\delta_{n}\sqrt{n})^{3}}{\sqrt{n}}\Big)\bigg).
\end{align*}
Third, since the above summand is exponentially small as $\ell \to \infty$, we can replace $\sum\limits_{\ell=0}^{C\log n}$ by $\sum\limits_{\ell=0}^{\infty}$ at the cost of an error $\bigO(n^{-100})$, and we finally obtain

$$S_2=e^{-t^2} \frac{\sqrt{\Delta Q(r_1)}}{r_1}\sum\limits_{\ell=0}^{\infty} \frac{a\rho^{2(\ell+\frac{1}{2}-x) }}{1+a\rho^{2(\ell+\frac{1}{2}-x)}}\cdot \bigg(1+\bigO\Big(\frac {(\delta_{n}\sqrt{n})^{3}}{\sqrt{n}}\Big)\bigg),\qquad (n\to\infty).$$

A similar analysis (left for the reader) of the sum $S_3$ shows that
$$
S_3 = -e^{-t^2}\frac{\sqrt{\Delta Q(r_1)}}{r_1} \sum\limits_{\ell=0}^{\infty} \frac{a^{-1} \rho^{2(\ell+\frac{1}{2}+x)}}{1+ a^{-1} \rho^{2(\ell+\frac{1}{2}+x)}}\cdot \bigg(1+\bigO\Big(\frac {(\delta_{n}\sqrt{n})^{3}}{\sqrt{n}}\Big)\bigg).
$$

\smallskip

It will be convenient to consider the following special function.

\begin{defn}[``Modified theta function'']
Given three real parameters $x,\rho,a$ with $x \in \mathbb{R}$, $0<\rho<1$, and $a>0$, we define
\begin{align}
\Theta(x;\rho,a) & = x(x-1)\log \rho + x \log a \nonumber \\
& + \sum_{j=0}^{\infty} \log ( 1+a \rho^{2(j+x)} ) + \sum_{j=0}^{\infty} \log ( 1+a^{-1} \rho^{2(j+1-x)} ). \label{def of Theta}
\end{align}
\end{defn}

As shown in \cite{C}, the function $\Theta$ is related to the Jacobi theta function $\theta$ as follows.
\begin{lem}(\cite[Lemma 3.28]{C})\label{lemma: Theta is expressible in terms of Jacobi}
We have
\begin{align*}
\Theta(x;\rho,a) = \frac{1}{2}\log \bigg( \frac{\pi a \rho^{-\frac{1}{2}}}{\log(\rho^{-1})} \bigg) + \frac{(\log a)^{2}}{4\log(\rho^{-1})} - \sum_{j=1}^{+\infty} \log(1-\rho^{2j}) + \log \theta \bigg( x + \frac{\log(a \rho)}{2\log(\rho)} ; \frac{\pi i}{\log(\rho^{-1})} \bigg),
\end{align*}
where $\theta$ is given by \eqref{def of Jacobi theta}.
\end{lem}

At this point, we note that our definition \eqref{def of Theta} of the function  $\Theta(u;\rho,a)$ gives
\begin{multline*}
\partial_{u}\Theta(u;\rho,a) = (2u-1)\log \rho + \log (a)
+ 2 \log (\rho) \bigg\{\sum_{j=0}^{+\infty} \frac{a \rho^{2(j+u)}}{1+a \rho^{2(j+u)}} - \sum_{j=0}^{+\infty} \frac{a^{-1} \rho^{2(j+1-u)}}{1+a^{-1} \rho^{2(j+1-u)}} \bigg\}.
\end{multline*}

Using Lemma \ref{lemma: Theta is expressible in terms of Jacobi}, \eqref{trecal} and $\theta(z+1;\tau)=\theta(z;\tau)$, we infer that
\begin{align*}
\partial_1 \Theta(\tfrac{1}{2}-x; \rho, a(s)) = - (\log \theta)' \bigg( \frac{B}{2}n + \frac{\log a(s)}{2\log(r_2/r_1)} ; \frac{\pi i}{\log(r_{2}/r_{1})} \bigg).
\end{align*}
Therefore,
\begin{align}\label{S2+S3}
S_2 + S_3 = -e^{-t^2} \frac{\sqrt{\Delta Q(r_1)}}{r_1} \Big(  \frac{(\log \theta)' ( \frac{B}{2}n + \frac{\log a(s)}{2\log(r_2/r_1)} ; \frac{\pi i}{\log(r_{2}/r_{1})} ) +\log(a(s))} {2\log \rho} -  x   \Big) \cdot \bigg(1+\bigO\Big(\frac {(\delta_{n}\sqrt{n})^{3}}{\sqrt{n}}\Big)\bigg).
\end{align}

Summing up, using \eqref{splitting of RtG}, \eqref{S1asy} and \eqref{S2+S3}, we have shown that
\begin{multline}\label{sumup}
\tilde{R}_n^G(z)= \frac{n}{2}\Delta Q(r_1) \erfc(t) +
\sqrt{\frac{n}{2\pi}} \frac{\sqrt{\Delta Q(r_1)}}{r_1} e^{-t^2}\Bigg[\frac{s}{2}r_1\lfun'(r_1)
+  \frac{1}{6} (t^2-2) - \frac{\log(a(s))}{2\log(\rho)} \\
+r_1\frac{\nlap Q(r_1)}{\Delta Q(r_1)}\Big(\frac{1}{2} \sqrt{\pi} t\erfc(t)e^{t^2} -\frac{1}{12}(2t^2+5)\Big)
 -\frac{(\log \theta)' ( \frac{B}{2}n + \frac{\log a(s)}{2\log(r_2/r_1)} ; \frac{\pi i}{\log(r_{2}/r_{1})} )}{2\log \rho}
\Bigg] + \bigO((\delta_{n}\sqrt{n})^{4}),
\end{multline}
as $n \to \infty$ uniformly for $|t| \leq \delta_{n}\sqrt{n}$.

\smallskip

In view of \eqref{trecal} and Lemma \ref{bock}, we obtain the statement in Theorem \ref{sr1}. $\qed$

\subsection{ The $r_2$-case } Now we consider $r=|z|$ close to the circle $r=r_2$, and more precisely
$$r=|z|=r_2+\frac t {\sqrt{2n\Delta Q(r_2)}},\qquad (t\in\R,\,|t|\le \delta_{n}\sqrt{n}).$$

This case
is similar to the $r_1$-case, but there are some subtle differences, so another careful analysis is in order. We start by writing
$$
\tilde{R}_n^G(z) = \sqrt{\frac{n}{2\pi}}(S_1 + S_2 + S_3),
$$ where
\begin{align}
&S_1 = \sum\limits_{j= m}^{m+Cn\delta_n} \frac{|z|^{2j}}{c_2(j,n)} e^{-nQ(z)+s\lfun(z)},
\\
&S_2 = \sum\limits_{j=m-Cn\delta_n}^{m-1} \frac{|z|^{2j}}{c_1(j,n)+c_2(j,n)} e^{-nQ(z)+s\lfun(z)},
\\
&S_3=-\sum\limits_{j=m}^{m+Cn\delta_n} \frac{c_1(j,n)}{c_2(j,n)}\frac{|z|^{2j}}{c_1(j,n)+c_2(j,n)}e^{-nQ(z)+s\lfun(z)} .
\end{align}

We begin by looking at $S_1$,
\begin{multline}\label{S1new}
S_1 = \sum\limits_{j= m}^{m+Cn\delta_n} \frac{|z|^{2j}}{c_2(j,n)} e^{-nQ(z)+s\lfun(z)} = \sum\limits_{\ell=0}^{Cn\delta_n} \frac{\sqrt{\Delta Q(r_2)}}{r_2}  \bigg(1 + f_2(y_{2,\ell}) \frac{1}{\sqrt{n}}+\bigO(n\delta_n^4) \bigg) e^{-y_{2,\ell}^2},
\end{multline}
where again $\ell=j-m$ while $f_2(y)$ is given by \eqref{simplef} and $y_{2,\ell}$ is given in \eqref{ykl}.

Using a second order Riemann sum approximation (Lemma \ref{LemmaRS}), we find
\begin{multline}
\sum\limits_{\ell=0}^{Cn\delta_n}  \frac{\sqrt{\Delta Q(r_2)}}{r_2} e^{-y_{2,\ell}^2} = \sqrt{n} \int\limits_{-C\log n}^{t} \sqrt{2}\Delta Q(r_{2}) e^{-y^2} dy + \tfrac{1}{2}e^{-t^2} \frac{\sqrt{\Delta Q(r_2)}}{r_2} +  \bigO(n^{-\frac 12})\\
\label{LT1}
= \sqrt{n} \int\limits_{-t}^{\infty} \sqrt{2}\Delta Q(r_2) e^{-y^2} dy  + \tfrac{1}{2}e^{-t^2} \frac{\sqrt{\Delta Q(r_2)}}{r_2} + \bigO(n^{-\frac 12}), \qquad \mbox{as } n \to \infty.
\end{multline}

(The sign of the constant term has changed with respect to \eqref{comp0}, since we now include the index $\ell=0$ in the sum.)

Next we use Riemann sum approximation on the subleading term of $S_1$ and find
\begin{multline}\label{ST1}
\frac{1}{\sqrt{n}}\sum\limits_{\ell=0}^{Cn\delta_n}  \frac{\sqrt{\Delta Q(r_2)}}{r_2} f_2(y_{2,\ell})  e^{-y_{2,\ell}^2 }=  \int\limits_{-C\log n }^{t} \sqrt{2}\Delta Q (r_2) f_2(y)  e^{-y^2} dy + \bigO(n^{-\frac 12})\\
=\int\limits_{-\infty}^{t} \sqrt{2}\Delta Q (r_2) f_2(y)  e^{-y^2} dy + \bigO(n^{-\frac 12}).
\end{multline}

Inserting the formula \eqref{simplef} for $f_2(y)$ we find
\begin{multline}\label{MT1}
\int\limits_{-\infty}^{t} \sqrt{2}\Delta Q(r_2)f_2(y) e^{-y^2} dy =  \frac{\sqrt{\Delta Q(r_2)} }{r_2} e^{-t^2} \Big(   -\frac{s}{2}r_2\lfun'(r_2) + x   - \frac{t^2+1}{6} \Big) \\
+ \frac{\radlap Q(r_2)}{\sqrt{\Delta Q(r_2)} } \Big(   \frac{1}{2} \sqrt{\pi}t \erfc(-t)   +        \frac{1}{12} e^{-t^2}(2t^2+5)   \Big)+\bigO(n^{-\frac 12}).
\end{multline}

Assembling the above information gives the approximation
\begin{multline}\label{S1NF}
S_1 =  \sqrt{n} \int\limits_{-t}^{\infty} \sqrt{2}\Delta Q(r_2) e^{-y^2} dy +   \frac{\sqrt{\Delta Q(r_2)}}{r_2}e^{-t^2}\Big( -\frac{s}{2}r_2\lfun'(r_2)  + \frac{1}{6} (2-t^2)
+ x\Big)\\
+ \frac{\radlap Q(r_2)}{\sqrt{\Delta Q(r_2)}}\Big(\frac{1}{2} \sqrt{\pi} t\erfc(-t) +\frac{1}{12} e^{-t^2}(2t^2+5)\Big)+\bigO\bigg(\frac {(\delta_{n}\sqrt{n})^{4}}{\sqrt{n}}\bigg), \qquad \mbox{as } n \to \infty.
\end{multline}

The sums $S_2$ and $S_3$ are treated similarly as in the $r_1$-case, with the result that

\begin{equation}
S_2 = \frac{\sqrt{\Delta Q(r_2)}}{r_2} e^{-t^2} \sum\limits_{\ell=0}^{\infty}  \frac{a^{-1}\rho^{2(\ell+x+\frac{1}{2})} }{1+a^{-1}\rho^{2(\ell+x+\frac{1}{2}})} +\bigO\bigg(\frac {(\delta_{n}\sqrt{n})^{3}}{\sqrt{n}}\bigg),
\end{equation}
and
\begin{equation}
S_3  = -\frac{\sqrt{\Delta Q(r_2)}}{r_2} e^{-t^2} \sum\limits_{\ell=0}^{\infty} \frac{a\rho^{2(\ell-x+\frac{1}{2})}}{1+a \rho^{2(\ell-x+\frac{1}{2})}} +\bigO\bigg(\frac {(\delta_{n}\sqrt{n})^{3}}{\sqrt{n}}\bigg).
\end{equation}

Hence
$$
S_2 + S_3 = e^{-t^2} \frac{\sqrt{\Delta Q(r_2)}}{r_2} \Big(   \frac{\partial_1 \Theta(x+\tfrac{1}{2}; \rho, 1/a)-\log(1/a)}{2\log {\rho}} -  x   \Big) + \bigO\bigg(\frac {(\delta_{n}\sqrt{n})^{3}}{\sqrt{n}}\bigg).
$$
Using Lemma \ref{lemma: Theta is expressible in terms of Jacobi}, \eqref{trecal} and $\theta(z+1;\tau)=\theta(z;\tau)$, we obtain
\begin{align*}
\partial_1 \Theta(x+\tfrac{1}{2}; \rho, a(s)^{-1}) = (\log \theta)' \bigg( \frac{B}{2}n + \frac{\log a(s)}{2 \log(r_{2}/r_{1})} ; \frac{\pi i}{\log(r_{2}/r_{1})} \bigg),
\end{align*}
and all in all, we obtain the asymptotics, as $n\to\infty$,
\begin{multline}
\tilde{R}_n^G(z)= \frac{n}{2}\Delta Q(r_2) \erfc(-t) \\
+\sqrt{\frac{n}{2\pi}} \frac{\sqrt{\Delta Q(r_2)}}{r_2} e^{-t^2}\Big(-\frac{s}{2}r_2\lfun'(r_2) +  \frac{1}{6} (2-t^2) +\frac{(\log \theta)' ( \frac{B}{2}n + \frac{\log a(s)}{2\log(r_2/r_1)} ; \frac{\pi i}{\log(r_{2}/r_{1})} )+\log(a(s))}{2\log \rho} \Big) \\  +  \sqrt{\frac{n}{2\pi}} e^{-t^2}\frac{\nlap Q(r_2)}{\sqrt{\Delta Q(r_2)}}\Big(\frac{1}{2} \sqrt{\pi} t\erfc(-t)e^{t^2} +\frac{1}{12}(2t^2+5)\Big)
+ \bigO((\delta_{n}\sqrt{n})^{4}).
\end{multline}

\smallskip
 Using Lemma \ref{bock}, we obtain the statement in Theorem \ref{sr2}. $\qed$

\section{Edge fluctuations}\label{sec4}

In this section, we prove Theorem \ref{flucco} on the distribution of fluctuations $\fluct_n\lfun$ (recall the definition \eqref{def of fluctn}), where $\lfun$ is a fixed, radially symmetric, sufficiently smooth function supported in some small neighbourhood of the closure $\overline{G}$, and $d\sigma=n\Delta Q\cdot \1_S\, dA$ is the equilibrium measure.

By now the strategy is clear: Lemma \ref{lem11} tells us how to compute the cumulant generating function using the perturbed $1$-point functions $\tilde{R}_{n,s\lfun}$, and these we know to a good precision thanks to Theorems \ref{sr1}, \ref{sr2} coupled with well-known bulk-asymptotics which is found in the literature.

\subsection{Expectation of fluctuations}
Fix a $C^6$-smooth test function $f$ supported in a small neighbourhood of the closure of the gap $G=\{r_1<|z|<r_2\}$. We do not assume that $f$ is radially symmetric.

Given a random sample $\{z_j\}_1^n$ associated with the perturbed
potential
$$\tilde{Q}=Q-\frac {s\lfun} n,$$
we consider the linear statistic $\fluct_n f$.

The goal of the present section is to use asymptotics for the $1$-point function $\tilde{R}_{n,s\lfun}$ to obtain a good approximation of the expectation
\begin{equation}\label{2be}
\tilde{\mathbb{E}}_{n,s\lfun}(\fluct_n f)=\int_{\mathbb{C}} f\cdot (\tilde{R}_{n,s\lfun}-n\Delta Q\1_S)\, dA.
\end{equation}
To this end, we will use boundary asymptotics from Theorems \ref{sr1} and \ref{sr2}, as well as the following well-known lemmas:

\begin{lem}[``Bulk asymptotics'']\label{wlk}
Assume $Q, \lambda$ are $C^{6}$-smooth in a neighbourhood of $S$. For $z\in \Int S$ such that
$$\dist(z,\d S)\ge \delta_n,$$
we have
\begin{align}\label{bulkr}\tilde{R}_{n,s\lambda}(z)&=n\Delta\tilde{Q}(z)+\frac 1 2 \Delta\log\Delta \tilde{Q}(z)+\bigO(n^{-1})\\
&=n\Delta Q(z)+\frac 1 2 \Delta\log\Delta Q(z)-s\Delta\lfun(z)+\bigO(n^{-1}).\nonumber
\end{align}
\end{lem}

A proof of Lemma \ref{wlk} can for example be found in \cite[Corollary 2.3]{A2}. See alternatively \cite{BBS2008} for related, very general results in a setting of several complex variables.

\begin{lem}[``Localization'']\label{loca} Let $\delta(z)=\dist(z,S)$. There are then positive constants $C$ and $c$ such that, whenever $|s|\le 1$,
$$\tilde{R}_{n,s\lfun}(z)\le Ce^{-cn\min\{\delta(z)^2,1\}},\qquad z\in\C.$$
\end{lem}

 Lemma \ref{loca} is much simpler to prove than Lemma \ref{wlk} and can be found in many sources. See \cite[Lemma 3.1]{AM}, for example.
 Also recall that when studying fluctuations using the integral formula in Lemma \ref{lem11}, we may restrict to values of the parameter $s$ for which $0\le s\le 1$. This is done tacitly in the rest of this section.

 \smallskip

In view of Lemma \ref{loca}, we see that in \eqref{2be}, restricting the integration to the $\delta_n$-neighbourhood of the droplet $S$, causes merely a negligible error of order $O(e^{-cn\delta_n^2})$.

With this in mind, we introduce two annular regions
\begin{align*}
\calA(r_k,\delta_n)=\{z\in\C\,:\,||z|-r_k|< \delta_n\},\qquad k=1,2.
\end{align*}

Combining \eqref{2be} with Lemma \ref{wlk} and Lemma \ref{bulkr} it is straightforward to check that
\begin{align}\label{badem}
\tilde{\mathbb{E}}_{n,s\lfun}(\fluct_n f)&=I_{b}+I(r_1)+I(r_2)+\bigO(\delta_n)
\end{align}
where
\begin{equation}\label{ird}
I_{b}=\int_S f\cdot \bigg(\frac 1 2\Delta\log \Delta Q-s\Delta\lfun\bigg)\, dA, \quad I(r_k)=\int_{\calA(r_k,\delta_n)} f\cdot (\tilde{R}_{n,s\lfun}-n\Delta Q\1_S)\, dA,\quad k=1,2.
\end{equation}

The integral $I_{b}$ is already in a suitable form; we turn to estimates for $I(r_1)$ and $I(r_2)$.

\subsection{The integral over $\ROne$.} In the following we consider a fixed point $z\in\ROne$, which we represent uniquely as
$$z=e^{i\theta}\bigg(r_1 + \frac{t}{\sqrt{2n\Delta Q(r_1)}}\bigg),$$
where $0\le \theta<2\pi$ while $t$ is real and $|t|\le \log n$.

We shall use the formula \eqref{finalr1}.

In order to analyze the integral $I(r_1)$, we must estimate the difference
$$\tilde{R}_{n,s\lfun}(z)-n\Delta Q(z)\,\1_S(z).$$
We want to replace ``$\Delta Q(z)$'' by ``$\Delta Q(r_1)$''. By Taylor's formula, using $r_1=|z|-\frac{t}{\sqrt{2n\Delta Q(r_1)}},$ we have

$$\Delta Q(z)=\Delta Q(r_{1})+\frac{t}{\sqrt{2n\Delta Q(r_1)}}\nlap Q(r_1)+\bigO(\delta_n^2).$$

Combining this with \eqref{finalr1}, we get for $|t|\le \log n$ that,

\begin{align}
\tilde{R}_{n,s\lfun}(z)&-n\Delta Q(z)\1_S(z)=n\Delta Q(r_1)\bigg(\frac {\erfc t} 2-\1_{t\le 0}\bigg)  + \sqrt{\frac{n}{2\pi}} \frac{\sqrt{\Delta Q(r_1)}}{r_1} e^{-t^2} \frac{t^2-2}{6} \label{mform2} \\
&+\sqrt{n}\frac{\nlap Q(r_1)}{\sqrt{2\Delta Q(r_1)}}\left(\frac{t\erfc(t)}{2}-t\1_{t<0} -e^{-t^2}\frac{1}{12\sqrt{\pi}}(2t^2+5)\right)\nonumber \\
&+\sqrt{\frac{n}{2\pi}}\frac{\sqrt{\Delta Q(r_1)}}{r_1} e^{-t^2} \Bigg[\frac{\log a(s)}{2\log(r_{2}/r_{1})} + \frac s 2 r_1\lfun'(r_1) \nonumber \\
&+ \frac{1}{2 \log(r_{2}/r_{1})} (\log \theta)' \bigg( \frac{B}{2}n + \frac{\log a(s)}{2\log(r_2/r_1)} ; \frac{\pi i}{\log(r_{2}/r_{1})} \bigg)
\Bigg] + \bigO((\delta_{n}\sqrt{n})^{4}). \nonumber
\end{align}

By \eqref{ird} and \eqref{mform2} we have the approximation
$$I(r_1)=B_{n,1}+C_{n,1}+D_{n,1}+E_{n,1}+\bigO((\delta_{n}\sqrt{n})^{4})\delta_{n}),$$
where the terms are, in turn:
\begin{align*}
B_{n,1}&=n\int_{\ROne}f(z)\Delta Q(r_1)\bigg(\frac {\erfc t}{2}-\1_{t\le 0} \bigg) \, dA(z),\\
C_{n,1}&=\sqrt{\frac n {2\pi}}\frac{\sqrt{\Delta Q(r_1)}}{r_1}  \int_{\ROne} f(z)\,e^{-t^2} \frac{t^2-2}{6}  dA(z),\\
D_{n,1}&=\sqrt{n}\frac {\dn\Delta Q(r_1)}{\sqrt{2\Delta Q(r_1)}}\int_{\ROne} f(z)\bigg( \frac{t\erfc(t)}{2}-t\1_{t<0}-\frac 1 {12\sqrt{\pi}} (2t^2+5)e^{-t^2}\bigg)\, dA(z),\\
E_{n,1}&=\sqrt{\frac n {2\pi}}\frac {\sqrt{\Delta Q(r_1)}} {r_1}\Bigg[\frac{\log a(s)}{2\log(r_{2}/r_{1})} + \frac s 2 r_1\lfun'(r_1) + \frac{1}{2 \log(r_{2}/r_{1})}  \\
&\times (\log \theta)' \bigg( \frac{B}{2}n +\frac{\log a(s)}{2\log (r_{2}/r_{1})} ; \frac{\pi i}{\log(r_{2}/r_{1})} \bigg)
 \Bigg] \int_{\ROne}f(z)e^{-t^2}
\, dA(z).
\end{align*}
Note that the parameter $s$ only enters the last term $E_{n,1}$.

\smallskip

In the thin annulus $\ROne$ we now expand the smooth function $f$ in a uniformly convergent Fourier series
$$f(re^{i\theta})=\sum_{j=-\infty}^{+\infty}f_j(r)e^{ij\theta},\qquad |r-r_1|\le\delta_n.$$
Writing the above integrals in polar coordinates, we see that only the term $f_0(r)$ will give a non-zero contribution. We can therefore assume that $f(z)=f_0(r)$ is radially symmetric.

A computation of the Jacobian gives
\begin{equation}\label{prex}dA=\frac 1 \pi r\, dr d\theta=\frac 1 \pi \bigg(r_1+\frac{t}{\sqrt{2n\Delta Q(r_1)}}\bigg)\frac 1 {\sqrt{2n\Delta Q(r_1)}}\, dtd\theta.\end{equation}

While we shall occasionally need the exact expression \eqref{prex}, it will for
the most part suffice to use the approximation
\begin{align}\label{approx Jac}
dA=\frac 1 \pi \frac {r_1} {\sqrt{2n\Delta Q(r_1)}}\, dtd\theta\,\cdot (1+\bigO(\delta_n)).
\end{align}

Let us write $\tilde{\calA}_n$ for the domain $\ROne$ in the $(t,\theta)$ coordinates:
$$\tilde{\calA}_n=\{(t,\theta)\, ;\, |t|\le \sqrt{2\Delta Q(r_{1})} \delta_{n} \sqrt{n}\,,\,0\le \theta<2\pi\}.$$

\subsubsection*{Computation of $B_{n,1}$}
To handle the $B_{n,1}$-term, we write
$$f_0(z)-f_0(r_1)=\dn f_0(r_1)\cdot \frac t{\sqrt{2n\Delta Q(r_1)}}+\bigO(\delta_{n}^{2}).$$
 Using that
$$\int_\R\bigg(\frac {\erfc t} 2-\1_{t\le 0}\bigg)\, dt=0,$$
and the precise Jacobian \eqref{prex},
we deduce that
\begin{align*}
B_{n,1}&=\bigg(f_0(r_1)+r_1 \dn f_0(r_1) \bigg) \frac{1}{2\pi}\int_{\tilde{\calA}_n} t\bigg(\frac {\erfc t} 2-\1_{t<0}\bigg)\, dtd\theta +\bigO(\delta_n^{2}\sqrt{n}).
\end{align*}

To evaluate the above integrals we note that
$$\int_{r_{1}-\sqrt{2\Delta Q(r_{1})}\delta_{n}\sqrt{n}}^{r_{1}+\sqrt{2\Delta Q(r_{1})}\delta_{n}\sqrt{n}} t \bigg(\frac {\erfc t} 2-\1_{t<0} \bigg)\, dt=\int_0^{\infty} t\erfc t\, dt + \bigO(e^{-cn\delta_n^2}), \qquad \mbox{as } n \to \infty,$$
and use Fubini's theorem,
\begin{align*}\int_0^\infty t\erfc t\, dt&=\frac 2 {\sqrt{\pi}}\int_0^\infty t\, dt\int_t^\infty e^{-s^2}\, ds=\frac 1 {\sqrt{\pi}}\int_0^\infty s^2 e^{-s^2}\, ds=\frac 1 4,
\end{align*}
so we see that
$$B_{n,1}=\frac{1}{4}\bigg(f_0(r_1)+r_1 \dn f_0(r_1) \bigg)+\bigO(\delta_n^{2}\sqrt{n}).$$

This gives
\begin{align*}
\lim_{n\to\infty}B_{n,1}&=\frac 1 {4}\frac 1 {2\pi}\bigg(\int_0^{2\pi}\dn f_0(r_1)\,r_1d\theta +\frac 1 {r_1}\int_0^{2\pi} f_0(r_1)\, r_1d\theta \bigg)\\
&=\frac 1 {8\pi} \int_{|z|=r_1}\dn f(z)\,|dz|+\frac 1 {8\pi r_1}\int_{|z|=r_1}f(z)\, |dz|.
\end{align*}

\subsubsection*{Computation of $C_{n,1}$}
Using that $\int_\R e^{-t^2}\, dt=\sqrt{\pi}$ and $\int_\R t^2e^{-t^2}\, dt=\frac {\sqrt{\pi}}2$, and using \eqref{approx Jac}, we compute
\begin{align*}
C_{n,1}&=\frac 1 \pi \sqrt{\frac n {2\pi}} \frac{\sqrt{\Delta Q(r_1)}}{r_1} \frac {r_1} {\sqrt{2n\Delta Q(r_1)}}\int_{\tilde{\calA}_n} f(z) \frac{t^2-2}{6} e^{-t^2}\, dtd\theta +\bigO(\delta_n)\\
&= f_0(r_1)\frac 1 {2\pi}\cdot \frac 1 {\sqrt{\pi}} \cdot 2\pi\cdot \int_\R \frac{t^2-2}{6} e^{-t^2}\, dt + \bigO(\delta_n) =f_0(r_1)\frac 1 {\sqrt{\pi}} \bigg( \frac{\sqrt{\pi}}{12} -\frac{\sqrt{\pi}}{3} \bigg) + \bigO(\delta_n) \\
& =-\frac{1}{4} f_0(r_1) + \bigO(\delta_n) =-\frac{1}{8\pi r_{1}}\int_{|z|=r_1} f(z)\,|dz| + \bigO(\delta_n).
\end{align*}

\subsubsection*{Computation of $D_{n,1}$} Using that $\int_0^\infty t\erfc t\, dt=\frac 1 4$, $\int_\R e^{-t^2}\, dt=\sqrt{\pi}$ and $\int_\R t^2e^{-t^2}\, dt=\frac {\sqrt{\pi}}2$,  we find
\begin{align*}
D_{n,1}&=\frac {r_1}{2\pi}\frac {\dn\Delta Q(r_1)}{\Delta Q(r_1)}\int_{\tilde{\calA}_n} f(z)\bigg( \frac{t\erfc t}{2}-\1_{t<0}-\frac 1 {12\sqrt{\pi}} (2t^2+5)e^{-t^2}\bigg)\, dtd\theta + \bigO(\delta_n)\\
&= -\frac{r_1}{4} f_0(r_1)(\dn\log\Delta Q) (r_1)+\bigO(\delta_n) =-\frac {1} {8\pi}\int_{|z|=r_1}f(z)(\dn\log\Delta Q) (z)\,|dz| + \bigO(\delta_n).
\end{align*}

\subsubsection*{Computation of $E_{n,1}$} A straightforward computation using that $\int_\R e^{-t^2}\, dt=\sqrt{\pi}$ shows that
\begin{align*}
E_{n,1}&=f_0(r_1)\Bigg[\frac{\log a(s)}{2\log(r_{2}/r_{1})} + \frac s 2 r_1\lfun'(r_1) + \frac{(\log \theta)' \big( \frac{B}{2}n + \frac{\log a(s)}{\log(r_2/r_1)} ; \frac{\pi i}{\log(r_{2}/r_{1})} \big)}{2 \log(r_{2}/r_{1})} \Bigg] +\bigO(\delta_n).
\end{align*}

\subsection{The integral over $\RTwo$} For $z\in\RTwo$ we write
$$z=e^{i\theta}\bigg(r_2+\frac t {\sqrt{2n\Delta Q(r_2)}}\bigg).$$

Using \eqref{finalr2}, the term $I(r_2)$ in \eqref{ird} decomposes as
$$I(r_2)=B_{n,2}+C_{n,2}+D_{n,2}+E_{n,2}$$
where
\begin{align*}
B_{n,2}&=n\int_{\RTwo}f(z)\Delta Q(r_2)\bigg(\frac {\erfc(-t)} 2-\1_{t\ge 0}\bigg)\, dA(z),\\
C_{n,2}&=-\sqrt{\frac n {2\pi}}\frac{\sqrt{\Delta Q(r_{2})}}{r_{2}}\int_{\RTwo} f(z) \frac{t^2-2}{6} e^{-t^2}\, dA(z),\\
D_{n,2}&=\sqrt{n}\frac {\dn\Delta Q(r_{2})}{\sqrt{2\Delta Q(r_{2})}}\int_{\RTwo} f(z)\bigg( \frac{t\erfc(-t)}{2} - t \1_{t>0}+\frac 1 {12\sqrt{\pi}} (2t^2+5)e^{-t^2}\bigg)\, dA(z),\\
E_{n,2}&=-\sqrt{\frac n {2\pi}}\frac {\sqrt{\Delta Q(r_{2})}} {r_{2}}\Bigg[ \frac{\log a(s)}{2\log(r_{2}/r_{1})}   + \frac s 2 r_2\lfun'(r_2) + \frac{1}{2\log(r_{2}/r_{1})}  \\
& \times (\log \theta)' \bigg( \frac{B}{2}n + \frac{\log a(s)}{2 \log(r_{2}/r_{1}) } ; \frac{\pi i}{\log(r_{2}/r_{1})} \bigg) \Bigg]
\int_{\RTwo}f(z)e^{-t^2}
\, dA(z).
\end{align*}

In a similar way as in the $r_1$-case, we now deduce the asymptotic relations as $n\to\infty$ (note that $\dn=-\partial_{r}$ on $|z|=r_{2}$)
\begin{align*}
B_{n,2}&=\frac 1 {8\pi}\int_{|z|=r_2}\dn f(z)\,|dz| - \frac 1 {8\pi r_{2}}\int_{|z|=r_{2}}f(z)\,|dz| +\bigO(\delta_n^{2}\sqrt{n}),\\
C_{n,2}&=\frac{1}{8\pi r_{2}} \int_{|z|=r_{2}} f(z)\,|dz|+ \bigO(\delta_n), \quad \\
D_{n,2}&=-\frac 1 {8\pi}\int_{|z|=r_2}f(z)(\dn\log\Delta Q)(z)\, |dz|+\bigO(\delta_n),\\
E_{n,2}&=-\frac 1 {2\pi r_2}\Bigg[ \frac{\log a(s)}{2\log(r_{2}/r_{1})}   + \frac s 2 r_2\lfun'(r_2) + \frac{1}{2\log(r_{2}/r_{1})}  \\
&   \times (\log \theta)' \bigg( \frac{B}{2}n + \frac{\log a(s)}{2\log(r_2/r_1)} ; \frac{\pi i}{\log(r_{2}/r_{1})} \bigg) \Bigg]\int_{|z|=r_2}f(z)\,|dz|+\bigO(\delta_n).
\end{align*}

\subsection{Conclusion of the proof} Let us denote the average values of a function $f$ with respect to the
boundary component $\{z:|z|=r_{k}\}$ of the gap $G$ by
\begin{equation}\label{perioddef}
M_k(f)=\frac 1 {2\pi r_k}\int_{|z|=r_k}f(z)\, |dz|,\qquad k=1,2.
\end{equation}

Using the decomposition \eqref{badem} we obtain
\begin{equation}\label{badem2}
\tilde{\mathbb{E}}_{n,s\lfun}(\fluct_n f)=I(f)+II(f)+III(f)+IV(f)+V(f)+\bigO((\delta_{n}\sqrt{n})^{4}\delta_{n}), \qquad \mbox{as } n \to \infty
\end{equation}
uniformly for $s \in [0,t]$, where
\begin{align*}
I(f)&=\int_S f\bigg(\frac 1 2 \Delta\log\Delta Q-s\Delta \lfun \bigg)\, dA,\\
II(f)&=\frac 1 {8\pi}\int_{\d S}\dn f(z)\,|dz|,\\
III(f)&=-\frac 1 {8\pi}\int_{\d S}f(z)\dn \log \Delta Q(z)\, |dz|, \\
IV(f)&= \frac{\log a(s)}{2\log(r_2/r_1)}(M_{1}(f)-M_{2}(f)) + s \frac{r_{1}\lambda'(r_{1})M_{1}(f)-r_{2}\lambda'(r_{2})M_{2}(f)}{2}, \\
V(f) &  = \frac{M_{1}(f)-M_{2}(f)}{2\log(r_{2}/r_{1})}(\log \theta)' \bigg( \frac{B}{2}n + \frac{\log a(s)}{2\log(r_2/r_1)} ; \frac{\pi i}{\log(r_{2}/r_{1})} \bigg).
\end{align*}

Denote the cumulant generating function of $\lfun$ by $F_{n,\lambda}(t)=\log \mathbb{E}_n (e^{t\fluct_n\lambda}).$

Making use of Lemma \ref{lem11}, we insert $f=\lambda$ in the formula \eqref{badem2} and integrate in $s$ from $0$ to $t$. This gives

\begin{align} \label{cumulant_jacobi}
& F_{n,\lambda} (t) = t\bigg(e_{\lambda} + (\lambda(r_{1})-\lambda(r_{2})) \frac{\log \frac{\Delta Q (r_{2})}{\Delta Q(r_{1})}}{4 \log(r_{2}/r_{1})} \bigg)  \\
&+\frac{t^2}{2}\bigg(v_\lambda + \frac {r_1 \lambda(r_{1})\lfun'(r_1)-r_2 \lambda(r_{2})\lfun'(r_2)}{2} \bigg)   +\calF_n(t)+\bigO\bigg(\frac {\log^5 n}{\sqrt{n}}\bigg) \nonumber,
\end{align}
where $\calF_n(t)$ is given by \eqref{brock}.

Since $\theta(z+1;\tau)=\theta(z;\tau)$, we can replace $\frac{Bn}{2}$ by $x$ in \eqref{brock}. Using then the well-known relation $\theta(\tau^{-1}z;-\tau^{-1}) = e^{\pi i z^{2} \tau^{-1}}\sqrt{-i\tau}\theta(z;\tau)$ (see \cite[eq 21.5.8]{NIST}), we obtain
\begin{align*}
e^{\mathcal{F}_{n}} = \frac{\theta\big(\frac{-i}{\pi}\log(\frac{r_{2}}{r_{1}}) ( \alpha + \frac{t(\lambda(r_{1})-\lambda(r_{2}))}{2\log(\frac{r_{2}}{r_{1}})} ); \frac{i}{\pi}\log(\frac{r_{2}}{r_{1}})\big)}{\theta\big(\frac{-i}{\pi}\log(\frac{r_{2}}{r_{1}}) \alpha ; \frac{i}{\pi}\log(\frac{r_{2}}{r_{1}})\big)}e^{-t(\lambda(r_{1})-\lambda(r_{2}))\alpha},
\end{align*}
where $\alpha = \alpha(n)=x+\frac {\log \frac {\Delta Q(r_2)}{\Delta Q(r_1)}}{4\log\frac {r_2}{r_1}}$. Using \eqref{def of Jacobi theta}, we can rewrite the above as
\begin{align*}
e^{\mathcal{F}_{n}} = \frac{\sum_{\ell=-\infty}^{+\infty} (\frac{r_{1}}{r_{2}})^{(\ell-\alpha)^{2}}e^{(\ell-\alpha)t(\lambda(r_{1})-\lambda(r_{2}))}}{\sum_{\ell=-\infty}^{+\infty} (\frac{r_{1}}{r_{2}})^{(\ell-\alpha)^{2}}} = \mathbb{E}[e^{t(\lambda(r_{1})-\lambda(r_{2}))(X_{n}-\alpha)}],
\end{align*}
where $X_{n}\sim dN(\alpha,\dNpara)$, $\dNpara=(\frac {r_1}{r_2})^2$. This proves \eqref{oscterm}. $\qed$

\begin{rem} In the above formulas, we observe that further simplifications occur if $f$ is radially symmetric, namely that $M_k(f)=f(r_{k})$ (by \eqref{perioddef}).
If $f$ is the real part of an analytic function in the gap $G$, then $M_1(f)=M_2(f)$ and the formulas also simplify.
\end{rem}

\section{Two-point correlations} \label{sec5} In this section, we prove Theorems \ref{thmr1r2} and \ref{thmr1r1} on off-diagonal asymptotics for the kernel $K_n(z,w)$.

\subsection{Preliminary computations} We now prove Theorem \ref{thmr1r2}. To this end consider two points $z$ and $w$ of the form
$$z=e^{i\theta_1}\bigg(r_1+\frac t {\sqrt{2n\Delta Q(r_1)}}\bigg),\qquad w=e^{i\theta_2}\bigg(r_2+\frac s {\sqrt{2n\Delta Q(r_2)}}\bigg).$$

As before, we shall begin by estimating the truncated kernel
$$K_n^G(z,w)=\sum_{|\tau-\frac B 2|\leq C \delta_n}\frac {p_j(z)\overline{p_j(w)}}{I(\tau)},$$
where $p_j(z)=z^je^{-\frac n 2 Q(z)}$ and $I(\tau)=\|p_j\|^2$.

Using the relation \eqref{iitau} (with $s=0$) to approximate $I(\tau)$,
as well as the Taylor approximation $g_j(r_{k,j})=g_j(r_k)+\mathcal{D}(\tau-\tfrac B 2)^2+\bigO((\tau-\tfrac B 2)^3)$ (the exact value of $\mathcal{D}$ will not be important to us),
we find that
\begin{multline*}
K_n^G(z,w)=\sqrt{\frac n {2\pi}}\sum_{|\ell|\leq C n\delta_{n}}
\frac {(z\bar{w})^j(r_1r_2)^{\ell-x}e^{-\frac n 2 (Q(z)+Q(w))}e^{{\frac n 2}(g_j(r_{1,j})+g_j(r_{2,j}))}}
{\frac {r_1^{2\ell-2x+1}}{\sqrt{\Delta Q(r_1)}} e^{\frac{\mathcal{C}}{2}\ell^{2}/n}+\frac {r_2^{2\ell-2x+1}}{\sqrt{\Delta Q(r_2)}} e^{-\frac{\mathcal{C}}{2}\ell^{2}/n}} +\bigO\big((\delta_{n}\sqrt{n})^{3}\big)\\
=\sqrt{\frac n {2\pi}}e^{\frac n 2(Q(r_{1})-Q(z))+\frac n 2(Q(r_{2})-Q(w))}
\sum_{|\ell|\leq C n\delta_{n}}\left(\frac {z\bar{w}}{r_{1}r_{2}}\right)^{j}
\frac{(r_1r_2)^{\ell-x}}{\frac {r_1^{2\ell-2x+1}}{\sqrt{\Delta Q(r_1)}}+\frac {r_2^{2\ell-2x+1}}{\sqrt{\Delta Q(r_2)}}} +\bigO\big((\delta_{n}\sqrt{n})^{3}\big),
\end{multline*}
where in the last step we have used a similar analysis as the one done below \eqref{s2prel}, relying on the fact that the summand is exponentially small as $|l|\to \infty$:
$$ \frac{(r_1r_2)^{\ell-x}}{\frac {r_1^{2\ell-2x+1}}{\sqrt{\Delta Q(r_1)}}+\frac {r_2^{2\ell-2x+1}}{\sqrt{\Delta Q(r_2)}}}\le C_1e^{-c_1|\ell|}, \qquad \mbox{for some } C_1,c_1>0.$$
The above asymptotics can be rewritten as
\begin{multline*}
K_n^G(z,w)=\sqrt{\frac n {2\pi}} e^{\frac n 2(Q(r_{1})-Q(z))+\frac n 2(Q(r_{2})-Q(w))}(z\bar{w})^m (r_1r_2)^{-\tfrac{Bn}{2}}\\
\times\sum_{|\ell|\leq C n\delta_{n}}\frac{(z\bar{w})^\ell}{\frac {r_1^{2\ell-2x+1}}{\sqrt{\Delta Q(r_1)}}+\frac {r_2^{2\ell-2x+1}}{\sqrt{\Delta Q(r_2)}}} +\bigO\big((\delta_{n}\sqrt{n})^{3}\big)
\end{multline*}

Recalling the definition \eqref{wsg} of $\calS^G(z,w;n)$, we now recognize that
\begin{equation}\label{vosm}
K_n^G(z,w)=\sqrt{2\pi n}(z\bar{w})^m (r_1r_2)^{-\tfrac{Bn}{2}}
\calS^G(z,w;n)e^{\frac n 2(Q(r_1)-Q(z)+Q(r_2)-Q(w))}+\bigO\big((\delta_{n}\sqrt{n})^{3}\big).
\end{equation}

By \eqref{pjn}, if $j$ is such that $|\tau-\frac B 2|>C \delta_n$, we have the estimate $\frac {|p_j(z)|}{\|p_j\|}\lesssim e^{-c n\delta_{n}^{2}}$, so we conclude from \eqref{vosm} that $K_n(z,w)=K_n^G(z,w)+\bigO(n^{-100})$.
This concludes our proof of Theorem \ref{thmr1r2}.

\subsection{The $r_1-r_2$ case}\label{section coro} We now prove Corollary \ref{corr1r2}. Using \eqref{vkdef} and \eqref{harm_ext2}, as $n \to \infty$ we get
\begin{align*}
& z^m r_1^{-\frac{Bn}{2}} e^{\frac{n}{2}(Q(r_1)-Q(z))} = e^{i\theta_1 m } |z|^{-x} e^{\frac{n}{2}(V_{\frac{B}{2}}^1(z)- Q(z) )} = e^{i\theta_1 m } |z|^{-x} e^{-\frac{t^2}{2}} + \bigO(n^{-1/2}), \\
& \bar{w}^m r_2^{-\frac{Bn}{2}} e^{\frac{n}{2}(Q(r_2)-Q(w))} = e^{-i\theta_2 m } |w|^{-x} e^{\frac{n}{2}( V_{\frac{B}{2}}^2(w)- Q(w) )} = e^{-i\theta_2 m } |w|^{-x} e^{-\frac{s^2}{2}} + \bigO(n^{-1/2}).
\end{align*}
Combining the above with Theorem \ref{thmr1r2}, the proof of Corollary \ref{corr1r2} follows. $\qed$

\subsection{The $r_1$-$r_1$ case} We now prove Theorem \ref{thmr1r1}. Thus we assume that
$\theta_1\ne \theta_2$, fix two real numbers $s,t$ and put
\begin{align*}
z=\bigg(r_1 + \frac{t}{\sqrt{2n\Delta Q(r_1)}} \bigg)e^{i\theta_1}, \qquad
w=\bigg(r_1 + \frac{s}{\sqrt{2n\Delta Q(r_1)}} \bigg)e^{i\theta_2}
\end{align*}

We will do estimates first for the truncated kernel
\begin{equation}\label{kunk}K_n^G(z,w)=\sum_{|\tau-\frac B 2|\le C\delta_n}\frac
{p_j(z)\overline{p_j(w)}}{\|p_j\|^2}.
\end{equation}
By the end of the proof, we shall be able to transfer the estimates to the full kernel $K_n(z,w)$.

In the following, we keep the notation $\rho =r_1/r_2$ ,  $a=a(0)=\sqrt{\frac{\Delta Q(r_2)}{\Delta Q(r_1)}}$ ,  $m=\lfloor Bn/2\rfloor$ and $x=\{Bn/2\}$.

We adapt our proof of edge asymptotics for $R_n(z)=K_n(z,z)$ in Theorem \ref{sr1} using the partial summation technique developed in \cite{AC} (cf.~also \cite[Chapter 3]{R} for an elementary source on partial summation). To this end, many detailed computations can be recycled modulo some minor notational changes. To avoid tedious repetitions, we shall be brief about such details.

Our starting point is the asymptotic formula
\begin{equation} \label{startp}
K_n^G(z,w) = \sqrt{\frac{n}{2\pi}} (S_1 + S_2+S_3)  (1+\bigO(n^{-1})), \qquad \mbox{as } n \to \infty,
\end{equation}
with
\begin{align}
&S_1 = \sum\limits_{j= m-Cn\delta_n}^{m-1} \frac{(z\bar{w})^j}{c_1(j,n)} e^{-\tfrac{n}{2}(Q(z)+Q(w))} ,
\\
&S_2 = \sum\limits_{j=m}^{m+Cn\delta_n} \frac{(z\bar{w})^{j}}{c_1(j,n)+c_2(j,n)} e^{-\tfrac{n}{2}(Q(z)+Q(w))} ,
\\
&S_3=-\sum\limits_{j=m-Cn\delta_n}^{m-1} \frac{c_2(j,n)}{c_1(j,n)}\frac{(z\bar{w})^{j}}{c_1(j,n)+c_2(j,n)} e^{-\tfrac{n}{2}(Q(z)+Q(w))},
\end{align}
where we used the notation \eqref{vupp}.

The sums $S_2$ and $S_3$ can be treated by adapting our arguments from the diagonal case $z=w$, but the sum $S_1$ requires a new analysis.

We start by rewriting $S_2$ and $S_3$ in the following way
\begin{align}
S_2 = \sum\limits_{j=m}^{m+Cn\delta_n} \frac{e^{i(\theta_1-\theta_2)j }}{1 + \frac{c_2(j,n)}{c_1(j,n)}} \frac{|z\bar{w}|^{j}}{c_1(j,n)}  e^{-\tfrac{n}{2}(Q(z)+Q(w))},
\end{align}
and
\begin{align}
S_3 = -\sum\limits_{j=m-Cn\delta_n}^{m-1}     \frac{e^{i(\theta_1-\theta_2)j}}{1 + \tfrac{c_1(j,n)}{c_2(j,n)}} \frac{|z\bar{w}|^{j}}{c_1(j,n)}  e^{-\tfrac{n}{2}(Q(z)+Q(w))}.
\end{align}

To simplify these expressions, we recall from Lemma \ref{Lemma:asymp pj over ck} that (with $r=|z|$)

\begin{align*}
\frac {|z|^j e^{-\frac{n}{2} Q(z)}} {\sqrt{c_1(j,n)}} = \frac {(\Delta Q(r_1))^{\frac 14}}{\sqrt{r_1}}
e^{-\frac{1}{2} (t-\frac{\ell/\sqrt{n}}{r_{1} \sqrt{2\Delta Q(r_1)}})^2} \cdot (1+\bigO(\delta_n^{2}\sqrt{n})) , \qquad \ell=j-m,
\end{align*}
as $n \to \infty$ uniformly for $|\ell| \leq Cn\delta_{n}$.

Making use of a similar relation with $z$ replaced by $w$, we obtain
$$
S_2 = \frac{\sqrt{\Delta Q(r_1)} }{r_1}        \sum\limits_{j=m}^{m+Cn\delta_n} \frac{e^{i(\theta_1-\theta_2)j }}{1 + \frac{c_2(j,n)}{c_1(j,n)}}e^{-\frac 1 2((t-\frac{\ell/\sqrt{n}}{r_{1}\sqrt{2\Delta Q(r_{1})}})^2+(s-\frac{\ell/\sqrt{n}}{r_{1}\sqrt{2\Delta Q(r_{1})}})^2)} + \bigO(\delta_n^{2}\sqrt{n}),
$$
and
$$
S_3 = -\frac{\sqrt{\Delta Q(r_1)} }{r_1}       \sum\limits_{j=m-Cn\delta_n}^{m-1} \frac{e^{i(\theta_1-\theta_2)j }}{1 + \frac{c_1(j,n)}{c_2(j,n)}}e^{-\frac 1 2((t-\frac{\ell/\sqrt{n}}{r_{1}\sqrt{2\Delta Q(r_{1})}})^2+(s-\frac{\ell/\sqrt{n}}{r_{1}\sqrt{2\Delta Q(r_{1})}})^2)}+ \bigO(\delta_n^{2}\sqrt{n}).
$$
By Lemma \ref{c1c2lem},
\begin{equation}\label{c2c1 Section5}
\frac{c_2(j,n)}{c_1(j,n)} = a^{-1} \Big(\frac{r_2}{r_1} \Big)^{2\ell-2x+1}\cdot e^{-\mathcal{C}\ell^2/n}\cdot
\bigg(1+\bigO\Big(
\frac {(\delta_{n}\sqrt{n})^{3}}{\sqrt{n}}\Big) \bigg),
\end{equation}
where $x=Bn/2-m$ with $m=\lfloor Bn/2 \rfloor$.
In particular, $\frac{c_{2}(j,n)}{c_{1}(j,n)}$ is exponentially large as $\ell \to \infty$ and exponentially small as $\ell \to -\infty$. Hence, following the same three steps done after equation \eqref{s2prel} (see also \cite[Lemma 3.6]{C}), we infer that we can replace $e^{-\frac 1 2((t-\frac{\ell/\sqrt{n}}{r_{1}\sqrt{2\Delta Q(r_{1})}})^2+(s-\frac{\ell/\sqrt{n}}{r_{1}\sqrt{2\Delta Q(r_{1})}})^2)}$ by $e^{-\frac{1}{2}(t^{2}+s^{2})}$ and $e^{-\mathcal{C}\ell^2/n}$ by $1$ in the above sums at the cost of an error of order $\frac{\log n}{\sqrt{n}}$. We thus arrive at
\begin{align*}
& S_2 = \frac{\sqrt{\Delta Q(r_1)} }{r_1}    e^{-\frac 1 2 (t^2+s^2)}  e^{i(\theta_1-\theta_2)m}  \sum\limits_{\ell=0}^{Cn\delta_n} \frac{a \rho^{2\ell-2x+1}e^{i(\theta_1-\theta_2)\ell }}{1 + a \rho^{2\ell-2x+1} } +\bigO\Big(
\frac {(\delta_{n}\sqrt{n})^{3}}{\sqrt{n}}\Big) , \\
& S_3 = - \frac{\sqrt{\Delta Q(r_1)} }{r_1}    e^{-\frac 12(t^2+s^2)} e^{i(\theta_1-\theta_2)m}  \sum\limits_{\ell=-Cn\delta_n}^{-1} \frac{a^{-1} \rho^{-(2\ell-2x+1)} e^{i(\theta_1-\theta_2)\ell}  }{1+a^{-1} \rho^{-(2\ell -2x+1)}} + \bigO\Big(
\frac {(\delta_{n}\sqrt{n})^{3}}{\sqrt{n}}\Big).
\end{align*}
Furthermore, in the above sums $Cn\delta_n$ can be replaced by $\infty$ at the cost of an exponentially small error term.

We next turn to $S_{1}$. The idea behind the following analysis is that since $z\bar{w}$ is close to the unimodular constant $e^{i(\theta_1-\theta_2)}\ne 1$, there should be large cancelations in the sum. This suggests using partial summation, similar to the approach in \cite{AC}.

To this end, using \eqref{def of h} with $s=0$ and \eqref{vupp}, we start by writing
\begin{align*}\frac {z^j}{\sqrt{c_1(j,n)}}e^{-\frac n 2 Q(z)}
=e^{ij\theta_1}\frac {(g_j''(r_{1,j}))^{\frac 1 4}}{(2r_{1,j})^{\frac 1 2}} e^{\frac n 2(V_\tau^1(z)-Q(z))}=e^{ij\theta_1}\frac {(\Delta Q(r_1))^{\frac 1 4}}
{r_1^{\frac 1 2}}e^{\frac n 2(V_\tau^1(z)-Q(z))}\cdot (1+\bigO(\delta_n)).
\end{align*}
Using a similar relation with $z$ replaced by $\bar{w}$, we conclude that
$$
S_1 = \frac{\sqrt{\Delta Q(r_1)}}{r_1}\sum\limits_{j=m-Cn\delta_n}^{m-1} e^{i(\theta_1-\theta_2)j} e^{\tfrac{n}{2}(V_{\tau}^1(z)-Q(z))} e^{\tfrac{n}{2}(V_{\tau}^1(w)-Q(w))} + \bigO(\delta_n).
$$

Again set
$$r=|z|=r_1+t/\sqrt{2n\Delta Q(r_1)}.$$

By the estimates \eqref{rkgen}, \eqref{droplet_tau} and \eqref{Jak}
we have
\begin{align*}n(V_{\tau}^1(r)-q(r)) &= n\Big(-2 \Delta Q(r_1)(r-r_{1,j})^2 +A_1(r-r_{1,j})^2(r_1-r_{1,j})+B_1(r-r_{1,j})^3+\bigO(\delta_n^4) \Big)\\
&=- \Big(t-\frac{\ell/\sqrt{n}}{r_1\sqrt{2\Delta Q(r_1)}}\Big) ^2 +\frac {A_{1}'(\ell/\sqrt{n})^{3}+B_{1}'t(\ell/\sqrt{n})^{2}+C_{1}'t^{2}\ell/\sqrt{n}+D_{1}'t^{3}}{\sqrt{n}}+\bigO(n\delta_n^4),
\end{align*}
where $A_1,B_1,A_1',B_1',C_{1}',D_{1}'$ are some constants whose exact values are irrelevant for what follows.

In this notation, we have
\begin{align}\label{S1 lol}
S_1= \frac{\sqrt{\Delta Q(r_1)}}{r_1} e^{im(\theta_{1}-\theta_{2})} \sum\limits_{\ell = 1}^{Cn\delta_n} c_{\ell} d_\ell + \bigO(\delta_n),
\end{align}
where
\begin{align}\label{billy}
c_{\ell}& = \exp\left( - \tfrac{1}{2} \Big(t+\tfrac{\ell/\sqrt{n}}{r_1\sqrt{2\Delta Q(r_1)}} \Big)^2 -\tfrac{1}{2} \Big(s +\tfrac{\ell/\sqrt{n}}{r_1\sqrt{2\Delta Q(r_1)}} \Big)^2\right) \exp(\bigO(n\delta_n^4)) \\
& \times \exp \bigg(\tfrac{A_{1}'(\ell/\sqrt{n})^{3}+B_{1}'t(\ell/\sqrt{n})^{2}+C_{1}'t^{2}\ell/\sqrt{n}+D_{1}'t^{3}}{2\sqrt{n}} + \tfrac{A_{1}'(\ell/\sqrt{n})^{3}+B_{1}'s(\ell/\sqrt{n})^{2}+C_{1}'s^{2}\ell/\sqrt{n}+D_{1}'s^{3}}{2\sqrt{n}} \bigg), \nonumber
\\
d_{\ell} &= e^{-i(\theta_1-\theta_2)\ell}.
\end{align}

As a simple but useful consequence of \eqref{billy}, we note the estimate
\begin{equation}\label{butt}c_\ell\lesssim e^{-k\ell^2/n},\qquad (1\le \ell\le Cn\delta_n),
\end{equation}
where $k>0$ and where the implied constant is uniform for $s,t$ in any given compact subset of $\mathbb{R}$.

The partial summation formula  gives
$$
\sum\limits_{\ell=1}^{Cn\delta_n} c_\ell d_\ell = c_{Cn\delta_n} D_{Cn\delta_n} -c_{1}D_{0} - \sum\limits_{\ell=1}^{Cn\delta_n-1}(c_{\ell+1}-c_{\ell}) D_\ell,
$$
where
$$D_\ell = \sum\limits_{j=0}^{\ell} d_j = \frac{1-\eps^{\ell+1}}{1-\eps},\qquad  \eps: = e^{-i(\theta_1-\theta_2)}.$$
Hence we can write
\begin{equation}\label{basic}
\sum\limits_{\ell=1}^{Cn\delta_n} c_\ell d_\ell = c_{Cn\delta_n} D_{Cn\delta_n} - c_1D_0 - (c_{Cn\delta_n}-c_1) \frac{1}{1-\eps} + \sum\limits_{\ell=1}^{Cn\delta_n-1} (c_{\ell+1}-c_{\ell}) \frac{\eps^{\ell+1}}{1-\eps}.
\end{equation}

Noting that $c_{Cn\delta_n} = \bigO(e^{-kn\delta_n^2})$ and $D_{\ell}=\bigO(1)$,
we find that
$$
\sum\limits_{\ell=1}^{Cn\delta_n} c_\ell d_\ell = c_1 \frac{\eps}{1-\eps} + \sum\limits_{\ell=1}^{Cn\delta_n-1} (c_{\ell+1}-c_{\ell}) \frac{\eps^{\ell+1}}{1-\eps} + \bigO(e^{-kn\delta_n^2}).
$$

Next we use the cancelations (which follows from \eqref{billy})
\begin{equation}\label{bucklo}c_{\ell+1}-c_\ell =
c_\ell \bigg(-\frac{\sqrt{2}\ell/\sqrt{n}+r_{1}(s+t)\sqrt{\Delta Q(r_{1})}}{r_{1}^{2}\Delta Q(r_{1})\sqrt{2n}}+\bigO(n\delta_n^4)\bigg),\end{equation}
to prove that
\begin{equation}
\label{differens}
\bigg|\sum\limits_{\ell=1}^{Cn\delta_n-1} (c_{\ell+1}-c_{\ell}) \frac{\eps^{\ell +1}}{1-\eps}\bigg| = \frac{\bigO(n^{-\frac 1 2})} {|1-\eps|^2}+\frac{\bigO(n^2\delta_n^5)}
{|1-\eps|}=\bigO\bigg(\frac{(\delta_{n}\sqrt{n})^{5}}{\sqrt{n}}\bigg).
\end{equation}

A detailed proof of \eqref{differens} runs as follows: inserting \eqref{bucklo} we see that
\begin{equation}\label{detail0}\sum\limits_{\ell=1}^{Cn\delta_n-1} (c_{\ell+1}-c_{\ell}) \frac{\eps^{\ell +1}}{1-\eps}=\frac A n \frac 1 {1-\eps}\sum_{\ell=1}^{Cn\delta_n-1}\ell c_\ell \eps^\ell+\frac B {\sqrt{n}}\frac 1 {1-\eps}\sum_{\ell=1}^{Cn\delta_n-1}c_\ell \eps^\ell+\bigO(n^2\delta_n^5),\end{equation}
where $A$ and $B$ are independent of $n$.

Using \eqref{basic} with $c_{\ell}$ replaced by $\ell c_{\ell}$, and using $|\epsilon| \le 1$, $(\ell+1)c_{\ell+1}-\ell c_\ell=\ell(c_{\ell+1}-c_\ell)+c_{\ell+1}$, \eqref{bucklo}, \eqref{butt} and a Riemann sum approximation we obtain
\begin{align*}
\bigg|\sum_{\ell=1}^{Cn\delta_n-1}\ell c_\ell \eps^\ell\bigg|&\lesssim \frac{1}{|1-\eps|}\frac 1 {\sqrt{n}}\sum_{\ell=1}^{Cn\delta_n-2}\ell c_\ell \bigg(1+\frac{\ell}{\sqrt{n}}\bigg) +\frac 1 {|1-\eps|}
\sum_{\ell=1}^{Cn\delta_n-2}c_{\ell+1}\\
&\lesssim \frac{1}{|1-\eps|}\sum_{\ell=1}^{Cn\delta_n-2}\bigg(\frac{\ell}{\sqrt{n}}+\frac{\ell^{2}}{n}\bigg) e^{-k\ell^2/n}+\frac 1 {|1-\eps|}\sum_{\ell=1}^{Cn\delta_n-2} e^{-k\ell^2/n}\\
&\lesssim \frac {\sqrt{n}}{|1-\eps|}\int_0^{C\sqrt{n}\delta_n}(t^{2}+t+1)e^{-kt^2}\, dt\lesssim \frac {\sqrt{n}}{|1-\eps|}.
\end{align*}

A similar estimate shows that
\begin{align*}
\bigg|\sum_{\ell=1}^{Cn\delta_n-1}c_\ell \eps^\ell\bigg|\lesssim \frac{1}{|1-\epsilon|}\int_0^{C\sqrt{n}\delta_n}(t+1)e^{-kt^2}\, dt\lesssim \frac{1}{|1-\epsilon|}.
\end{align*}

Inserting these estimates in \eqref{detail0}, we have verified the estimate \eqref{differens}.

\smallskip

Now note that the first term $c_1$ in \eqref{billy} satisfies
$$c_1 = e^{ -\tfrac{1}{2}(t^2+s^2)} + \bigO(n^{-\frac 12}(\delta_{n}\sqrt{n})^{3}),$$
whence, by \eqref{S1 lol}, \eqref{basic} and the above estimates \eqref{butt}, \eqref{differens},
$$
S_1 =  e^{i(\theta_1-\theta_2)m} \frac{\sqrt{\Delta Q(r_1)}}{r_1} e^{-\tfrac{1}{2}(t^2 +s^2)} \frac{1}{e^{i( \theta_1-\theta_2)}-1} + \bigO\bigg(\frac{(\delta_{n}\sqrt{n})^{5}}{\sqrt{n}}\bigg).
$$

Summing up, we have shown that
\begin{multline}\label{bjudo}K_n^G(z,w)=\frac{\sqrt{n\Delta Q(r_1)}}{\sqrt{2\pi}\,r_1} e^{-\tfrac{1}{2}(t^2 +s^2)} e^{i(\theta_1-\theta_2)m} \\
\times \bigg(  \frac{1}{e^{i( \theta_1-\theta_2)}-1} +         \sum\limits_{\ell=0}^{\infty} \frac{e^{i(\theta_1-\theta_2)\ell }}{1 + a^{-1} \rho^{-(2\ell-2x + 1}) }    -\sum\limits_{\ell=-\infty}^{-1} \frac{ e^{i(\theta_1-\theta_2)\ell}  }{1+a \rho^{2\ell -2x+1}} \bigg)              +       \bigO((\delta_{n}\sqrt{n})^{5}).
\end{multline}

There remains to replace the approximate kernel $K_n^G$ by the full kernel $K_n$.

But if $|\tau-\frac B 2|\ge C\delta_n$ then by \eqref{pjn} we have an easy estimate $\frac {|p_j(z)|}{\|p_j\|}\le e^{-kn\delta_{n}^{2}}$ with a suitable $k>0$.
Using also a similar estimate for $w$, we find
$$|K_n(z,w)-K_n^G(z,w)|=\bigO(e^{-kn\delta_n^2}).$$

Hence \eqref{bjudo} is true also when $K_n^G(z,w)$ is replaced with the full kernel $K_n(z,w)$.
Our proof of Theorem \ref{thmr1r1} is complete. $\qed$

\section{Fluctuations for real parts of analytic functions}\label{flureg}
\label{sec6}
In this section, we state and prove a general theorem on the behaviour of fluctuations $\fluct_n f$
where $f$ is a smooth function of the form $f=\re g$ in a neighbourhood of a general (smooth) gap $G$, i.e., a connected component of $\C\setminus S$ which disconnects the droplet $S$.
When we specialize to the radially symmetric case, we will obtain Theorem \ref{thereg} as an immediate consequence.

We require some preparations.

\subsection{Setup} We now return to the setting of a general potential function $Q$ as in the first paragraph (\ref{backgr}) in the introduction. We do not make any hypotheses about radial symmetry, but we assume that the droplet $S$ has a \textit{gap} $G$, i.e., a component of the complement $\C\setminus S$ which disconnects $S$. It is convenient to assume that the set $G$ is bounded,
and that the boundary $\d G$ consists of two disjoint, real-analytic Jordan curves.

Under these conditions, the obstacle function $\check{Q}$, which is harmonic in $G$, continues harmonically across the boundary to some neighbourhood of the closure $\overline{G}$. Let us fix such a small neighbourhood $N_1$ and denote by $V$ the harmonic continuation of $\check{Q}|_G$ to $N_1$ (to be specific, we assume that $N_1$ is small enough that the boundary $\d N_1\subset \Int S$ and $N_1\cap (\C\setminus S)=G$).

It is important to assume that the Laplacian $\Delta Q$ is \textit{strictly} positive in a small neighbourhood $N_1'$ of the boundary $\d G$ of the gap. Restricting $N_{1}'$ if necessary, we assume without loss of generality that $N_1'\subset N_{1}$. We will also assume that $Q$ is $C^4$-smooth
in $N_{1}'$. We then consider the $C^2$-smooth function
$$L=\log \Delta Q|_{N_1'}.$$
We continue $L$ to a $C^2$-smooth function on $\C$ in a way such that
$$\supp L\subset N_1.$$

In addition to $L$, we will use its Poisson-modification $L^G$. By definition, $L^G=L$ in $\C\setminus G$, while $L^G$ is harmonic in $G$ and continuous up to the boundary. I.e., in $G$, $L^G$ is the solution to Dirichlet's problem with boundary values $L$. Choosing the neighbourhood $N_1$ somewhat smaller we can assume that $L^G|_G$ continues harmonically to a harmonic function $\tilde{L}^G$ on $N_1$.

At the boundary $\d G$, the function $L^G$ has one-sided normal derivatives, but these do not agree in general. The jump in the normal derivative is encoded in the function
\begin{equation}\label{wolg}\calN(L^G)=-\dn L + \dn (\tilde{L}^G)\qquad (\text{on}\, \d G),
\end{equation}
which we call the \emph{Neumann jump} of $L^G$. Here we follow the convention that the normal derivative $\dn$ is taken in the direction of the normal pointing out of $S$ (i.e., into $G$).

\subsection{Statement of result}\label{genarp} We shall consider a class of test-functions on a small neighbourhood of the closure of the gap $G$. It is convenient to fix some notation.

 Given the neighbourhood $N_1$ of $\overline{G}$ above, we fix another neighbourhood $N_2$ such that $\overline{G}\subset N_2\subset\overline{N}_2\subset N_1$. We will consider test-function $f$
 such that
 \begin{enumerate}[label=(\roman*)]
 \item There is an analytic function $g$ on $N_1$ such that
 $$f=\re g\qquad \text{on}\qquad N_2.$$
  \item $f$ is globally $C^2$-smooth and $\supp f\subset N_1$.
 \end{enumerate}

Let $\{z_j\}_1^n$ be a corresponding random sample with respect to the Boltzmann-Gibbs measure \eqref{bogi}, and recall the definition $\fluct_n f=\sum_{j=1}^n f(z_j)-n\sigma(f)$, where $d\sigma=\Delta Q\cdot\1_S dA$ is the equilibrium measure and $\sigma(f) = \int_{\mathbb{C}}fd\sigma$. We will consider the cumulant generating function
$$F_{n,f}(t)=\log \mathbb{E}_n e^{t\fluct_n f}.$$

We will prove the following result.

\begin{thm} \label{gengap} If $f$ is as above, there exists a small number $\beta>0$ such that
\begin{align}\label{asymp of Fnf}
F_{n,f}(t)=te_f+\frac {t^2} 2 v_f+\bigO(n^{-\beta}), \qquad \mbox{as } n \to \infty,
\end{align}
uniformly for $t$ in compact subsets of $\mathbb{R}$, where
\begin{equation}\label{expeco}e_f=\frac 1 2\int_S f\cdot \Delta \log\Delta Q\, dA+\frac 1 {8\pi}
\int_{\d  G}\d_n f\, |dz|+\frac 1 {8\pi}\int_{\d G}f\cdot\calN(L^G)\, |dz|,
\end{equation}
and $v_f=-\int_S f\Delta f\, dA$. In particular $\fluct_n f$ converges in distribution to the normal $N(e_f,v_f)$ as $n\to\infty$.
\end{thm}

\begin{rem} With a slight extra twist, Theorem \ref{gengap} can be generalized by allowing the addition of a smooth function $f_0$ which vanishes at the boundary $\d G$. This is
done in \cite{AC:Heine},  (cf. also \cite{AM}).
\end{rem}

Before proving Theorem \ref{gengap}, we pause to explain how Theorem \ref{thereg} follows from it.

\begin{proof}[Proof of Theorem \ref{thereg} from Theorem \ref{gengap}] We now specialize to an annular gap of the form $G=\{r_1<|z|<r_2\}$. When comparing the results, the only issue is that the expressions for $e_f$ in \eqref{expeco} and \eqref{expf} do not appear to be quite the same. However, they will be the same if
\begin{equation*}\int_{\d G}f\cdot \calN(L^G)\,|dz|=-\int_{\d G}f\cdot \dn L\,|dz|,
\end{equation*}
i.e., if
\begin{equation}\label{IF}\int_{\d G}f\cdot \dn(\tilde{L}^G)\, |dz|=0.\end{equation}

However, in the present case, we have explicitly
$$\tilde{L}^G(z)=C_1+C_2 \log|z|,\qquad (C_2=\frac {\log(\frac {\Delta Q(r_2)}{\Delta Q(r_1)})}{\log\frac {r_2} {r_1}})$$
where $C_{1}$ is an irrelevant constant. It follows that
$$\int_{\d G}f\cdot \dn(\tilde{L}^G)\, |dz|=2\pi C_2\cdot (M_{1}(f)-M_{2}(f)).$$
where the average value $M_k(f)$ is defined in \eqref{perioddef}.
But $f$ is harmonic in a neighbourhood of $\{r_1\le |z|\le r_2\}$, so the two average values are the same. This proves that \eqref{IF} holds, and the rest of Theorem \ref{thereg} is immediate from Theorem \ref{gengap}.
\end{proof}

\subsection{The limit Ward identity} We now prepare for the proof of Theorem \ref{gengap}.
We will keep it brief, and we refer to \cite{AM} for more details.

We will work with an auxiliary function $v$, which is fixed by the end of the computation.
At the outset, $v$ may be any complex-valued, say $C^1$-smooth, compactly supported function.

Following \cite{AM} we now associate with $v$ and a random sample $\{z_j\}_1^n$ from \eqref{bogi} a random variable denoted $W_n^+[v]$, whose definition is:
\begin{align}\label{Wnplus}
W_n^+[v]:=B_n[v]-\trace_n[v\d Q]+\frac 1 n \trace_n[\d v],
\end{align}
where
$$B_n[v]=\frac 1 {2n}\sum_{1\leq j\ne k\leq n}\frac {v(z_j)-v(z_k)}{z_j-z_k}.$$

(The assignment $v\mapsto W_n^+[v]$ is known as a ``Ward's tensor''; it has a meaning in conformal field theory \cite{KM}. This, however, is not used in what follows.)

The basic \textit{Ward identity} asserts that
$$\mathbb{E}_n[W_n^+[v]]=0$$ for each function $v$ as above. The proof is not hard; in
\cite[Proposition 2.1]{AM} a proof is given depending on a change of variables. There is also a short proof based on integrating by parts, see e.g.~\cite{AKS2,BBNY2}.

In the following, we will consider smooth perturbations of the potential of the following kind. Fix a $C^2$-smooth, bounded real-valued function $h$ and set
$\tilde{Q}=Q-\frac h n.$
Note that $\tilde{\sigma}=\sigma$. Also write
$$d\sigma_n=\frac 1 n R_n\, dA,\qquad d\tilde{\sigma}_n=\frac 1 n \tilde{R}_n\, dA,\qquad d\sigma=\Delta Q\,\1_S\, dA,$$
where $R_n$, $\tilde{R}_n$ are 1-point functions in potentials $Q$, $\tilde{Q}$, respectively.

In potential $\tilde{Q}$ the Ward's tensor becomes
\begin{align}\label{Wtnplus}
\tilde{W}_n^+[v]=B_n[v]-\trace_n[v\d Q]+\frac 1 n \trace_n[\d v+v\d h].
\end{align}
Note that here everything is acting on a random sample $\{\tilde{z}_j\}$ from $\tilde{\mathbb{P}}_n$.

We emphasize that Ward's identity holds also for this kind of perturbed potentials, i.e., we have
\begin{equation}\label{wdid}
\tilde{\mathbb{E}}_n[\tilde{W}_n^+[v]]=0.
\end{equation}

Now denote
$$\nu_n(f):=\mathbb{E}_n(\fluct_n f)=n(\sigma_{n}(f)-\sigma(f)),\qquad \tilde{\nu}_n(f):=\tilde{\mathbb{E}}_n(\fluct_n f)=n(\tilde{\sigma}_{n}(f)-\sigma(f)),$$
and introduce the functions
$$D_n(z)=\nu_n(k_z),\qquad \tilde{D}_n(z)=\tilde{\nu}_n(k_z),$$
where
$$k_z(w)=\frac 1 {z-w}.$$
It is easy to show that these functions are small near infinity: $|\tilde{D}_n(z)|\le Cn/|z|^2$ as $z\to\infty$; see \cite[Lemma 2.4]{AM} for details.

Finally recall that the obstacle function associated to $Q$ is
$$\check{Q}(z)=-2U^\sigma(z)+\gamma=2\int \log|z-w|\, d\sigma(w)+\gamma$$
where $\gamma$ is the Robin's constant. Also recall our standing assumption that $S$ is the contact set:
$$S=\{z\,;\, Q(z)=\check{Q}(z)\}.$$
$\check{Q}$ is harmonic on $\C\setminus S$ and $Q>\check{Q}$ there.

Clearly
$\d \check{Q}(z)=\sigma(k_z)$, so
$$\tilde{D}_n(z)=n(\tilde{\sigma}_{n}(k_{z})-\sigma(k_{z}))=n(\tilde{\sigma}_n(k_z)-\d \check{Q}(z))$$
and
so (in the sense of distributions)
$$\dbar\tilde{D}_n(z)=\tilde{R}_n-n\Delta Q\,\1_S.$$

As a consequence,
$$\tilde{\nu}_n(f)=n(\tilde{\sigma}_{n}(f)-\sigma(f))=\int f\dbar[\tilde{D}_n]\, dA=-\int \tilde{D}_n\cdot\dbar f.$$

We now rewrite Ward's identity, using that (with $\tilde{R}_{n,2}$ the 2-point function in potential $\tilde{Q}$)
\begin{align*}
\tilde{\mathbb{E}}_n[B_n[v]]&=\frac 1 {2n}\int_{\C^2}\frac {v(z)-v(w)}{z-w}\, \tilde{R}_{n,2}(z,w)dA(z)dA(w)\\
&=\frac 1 n\int_{\C^2}  v(z)k_z(w)(\tilde{R}_n(z)\tilde{R}_n(w)-|\tilde{K}_n(z,w)|^2) dA(z)dA(w)\\
&=\int_{\C} v(z)\tilde{\sigma}_n(k_z)\tilde{R}_n(z)\, dA(z)-\frac 1 {2n}
\int_{\C^2}\frac {v(z)-v(w)}{z-w}|\tilde{K}_n(z,w)|^2dA(z)dA(w)
\end{align*}
and
$$\int_{\C} v(z)\tilde{\sigma}_n(k_z)\tilde{R}_n(z)\, dA(z)=\int_{\C} v\cdot \d\check{Q}
\cdot \tilde{R}_n\, dA+\frac 1 n\int_{\C} v\tilde{D}_n\tilde{R}_n\, dA.$$

Combining the above computations, we obtain
\begin{equation}\label{ebn}
\tilde{\mathbb{E}}_n[B_n[v]]=\int_{\C} v\cdot \d\check{Q}\cdot \tilde{R}_n+\frac 1 n\int_{\C} v\tilde{D}_n\tilde{R}_n-\frac 1 {2n}
\int_{\C^2}\frac {v(z)-v(w)}{z-w}|\tilde{K}_n(z,w)|^2.
\end{equation}

Inserting this in Ward's identity \eqref{wdid} we find
\begin{align}\label{wef}\int_{\C} v\cdot \d\check{Q}\cdot \tilde{R}_n&+\frac 1 n\int_{\C} v\tilde{D}_n\tilde{R}_n-\frac 1 {2n}
\int_{\C^2}\frac {v(z)-v(w)}{z-w}|\tilde{K}_n(z,w)|^2-\int_{\C} v\cdot \d Q\cdot \tilde{R}_n\\
&+\frac 1 n\int_{\C}(\d v+v\cdot \d h)\tilde{R}_n=0.\nonumber
\end{align}

This can be written as
\begin{multline}\label{chik}
\int_{\C} v\cdot   \d (\check{Q}-Q)\cdot \tilde{R}_n+\int_S v\cdot \tilde{D}_n\cdot \Delta Q+(\tilde{\sigma}_n-\sigma)(v\cdot\tilde{D}_n)
\\
=\frac 1 {2n}\int_{\C^{2}}\frac {v(z)-v(w)}{z-w}|\tilde{K}_n(z,w)|^2 -
\sigma(\d v+v \cdot\d h)-(\tilde{\sigma}_n-\sigma)(\d v+v\cdot\d h).
\end{multline}

Using estimates proved in \cite{AM}, we now identify negligible terms in the exact relationship \eqref{chik}.

First, by standard estimates on the $1$-point function $\tilde{R}_n$
it is easy to see that the last term is
\begin{equation}\label{3o}
(\tilde{\sigma}_n-\sigma)(\d v+v\cdot\d h)=\bigO \bigg(\frac {\log n}{\sqrt{n}}\bigg).
\end{equation}
(To prove this, it suffices combine the standard bulk and exterior asymptotic estimates $\frac 1 n|\tilde{R}_n(z)-n\Delta Q(z)\1_S(z)|=\bigO(\frac 1 n)$ when $\dist(z,\d S)\ge \delta_n$ with the global estimates
$R_n\lesssim n$, $\tilde{R}_n\lesssim n$. See for instance \cite{AM,A2,BBS2008}, or references given there.)

Moreover, from \cite[Proposition 4.5]{AM}
we have the estimate
\begin{equation}\label{2o}\frac 1 n \int_{\C^{2}}\frac {v(z)-v(w)}{z-w}|\tilde{K}_n(z,w)|^2=\sigma(\d v)+\bigO\bigg(\frac {\log^4 n}{\sqrt{n}}\bigg).
\end{equation}
Finally, by \cite[Proposition 4.6]{AM} there is a small number $\beta>0$ such that
\begin{equation}\label{1o}(\tilde{\sigma}_n-\sigma)(v\cdot\tilde{D}_n)=\bigO(n^{-\beta}).\end{equation}

Inserting \eqref{3o}-\eqref{1o} in \eqref{chik}
and using that $Q=\check{Q}$ on $S$ we now find that
\begin{equation}\label{wdl}\int_{\C\setminus S} v\cdot \d(\check{Q}-Q)\cdot \tilde{R}_n+\sigma(v\cdot\tilde{D}_n)=-\frac 1 2
\sigma(\d v)-\sigma(v\cdot\d h)+\bigO(n^{-\beta}).\end{equation}

It is convenient to rewrite the above formula, using that
$$\dbar \tilde{D}_n=\tilde{R}_n\qquad \text{on}\qquad \C\setminus S.$$
This gives
$$\int_S v\cdot \Delta Q\cdot \tilde{D}_n+\int_{\C\setminus S}v\cdot \d (\check{Q}-Q)\cdot\dbar
\tilde{D}_n=-\frac 1 2
\sigma(\d v)-\sigma(v\cdot\d h)+\bigO(n^{-\beta}).$$
But using that $\check{Q}$ is harmonic in $\C\setminus S$ we have (by a suitable application of Green's formula)
$$\int_{\C\setminus S}v\cdot \d (\check{Q}-Q)\cdot\dbar
\tilde{D}_n=\int_{\C\setminus S} \dbar v\cdot \d (Q-\check{Q})\cdot \tilde{D}_n+
\int_{\C\setminus S} v\cdot\Delta Q\cdot \tilde{D}_n,$$
where we have used that $\d (\check{Q}-Q)=0$ on $\partial (\mathbb{C}\setminus S)$.

Combining the last two identities, we obtain the following result, which corresponds to the ``limit Ward identity'' in \cite{AM}.
\begin{prop}\label{dentity} Let $v$ be any
complex-valued, $C^1$-smooth,
bounded function on $\C$. Then
$$\int_\C [v\Delta Q+\dbar v\cdot \d(Q-\check{Q})]\tilde{D}_n
=-\frac 1 2
\sigma(\d v)-\sigma(v\cdot\d h)+\bigO(n^{-\beta}).$$
\end{prop}

\begin{rem} The above proof actually works for less regular functions $v$ that are Lipschitz continuous and uniformly smooth in $\C\setminus \d S$ (and meeting some mild condition on the growth near infinity). This fact is used in \cite{AM} to study fluctuations of general smooth linear statistics when the droplet is connected and equals to the contact set. In \cite{AC:Heine} we exploit this fact concerning Proposition \ref{dentity} in order to study fluctuations in more general regimes; 
see \cite[Sections 8 and 9]{AC:Heine}.
\end{rem}

Armed with the limit Ward identity, we now return to the issue of finding the asymptotic distribution of $\fluct_n f=\sum_1^n f(z_j)-n\sigma(f)$ where $f$ is the real part of an analytic function in each gap.

Following \cite{AM} our strategy is to first find good approximations for the expectations
$$\nu_n(f)=\mathbb{E}_n(\fluct_n f),\qquad \tilde{\nu}_n(f)=\tilde{\mathbb{E}}_n(\fluct_n f),$$
and then use Lemma \ref{lem11} to compute the cumulant generating function with sufficient precision.

For the following computation, we remind that $L(z)$ denotes a fixed smooth extension of $\log\Delta Q|_{N_1'}$ to $\C$, with $\supp L\subset N_1$.

\subsection{Computation of $\nu_n(f)$} Assume as before that $f=\re g$ on $N_{2}$ where $g$ is holomorphic in $N_1$, with $\overline{G}\subset N_{2}\subset \overline{N}_{2}\subset N_{1}$. Thus $f=f_+ + f_-$ where $f_+=g/2$ is analytic and $f_-=\bar{g}/2$ is conjugate-analytic in $N_2$. Also $f_+$,$f_-$ are globally $C^\infty$-smooth and vanish identically outside $N_1$.

Now set
\begin{equation}\label{vanin}v_+=\frac {\dbar f_+}{\Delta Q},\qquad v_-=\frac {\d f_-}{\Delta Q}=\bar{v}_+.\end{equation}
Then $v_+,v_-$ are globally $C^1$-smooth and supported in the interior of the droplet $S$, close to the boundary $\d G$, where $\Delta Q>0$ by hypothesis. (It is understood that $v_+=v_-=0$ identically in $N_{2}$, and also in $\C\setminus N_1$.)

Since $Q=\check{Q}$ on $S$ and since the functions in \eqref{vanin} vanish on $N_{2}\cup (\mathbb{C}\setminus N_{1})$, we infer that $\dbar v_+\cdot \d (Q-\check{Q}) \equiv 0 \equiv \d v_-\cdot\dbar(Q-\check{Q}) $ on $\mathbb{C}$, and thus
\begin{equation}\label{band}
v_+\Delta Q+\dbar v_+\cdot \d (Q-\check{Q})=\dbar f_+,\qquad \text{and}\qquad  v_-\Delta Q+\d v_-\cdot\dbar(Q-\check{Q})=\d f_-.
\end{equation}

Multiplying the first identity in \eqref{band} by $D_n$ and recalling that $\nu_n(f_+)=-\int
D_n\cdot\dbar f_+$, we find that
$$-\nu_n(f_+)=\int [v_+\Delta Q+\dbar v_+\cdot \d(Q-\check{Q})]D_n.$$
By Proposition \ref{dentity} (with $h=0$) we now see that
\begin{equation}\label{brak}
\nu_n(f_+)=\frac 1 2 \sigma(\d v_+)+\bigO(n^{-\beta}),\qquad (n\to\infty).
\end{equation}

Taking complex conjugates, we infer that
\begin{equation}\label{lol5}
\nu_n(f_-)=\frac 1 2 \sigma(\dbar v_-)+\bigO(n^{-\beta}),\qquad (n\to\infty).
\end{equation}

Now observe that (with $L$ the above smooth extension of $\log\Delta Q$)
\begin{align*}
\sigma(\d v^+)&=\int_S\d\left(\frac {\dbar f_+} {\Delta Q}\right)\,\Delta Q\, dA=\int_S \frac {\Delta f_+\Delta Q-\dbar f_+ \cdot\d\Delta Q}{(\Delta Q)^2}\, \Delta Q\, dA\\
&=\int_S\Delta f_+\, dA-\int_S\dbar f_+\cdot \d L\, dA.
\end{align*}

By use of Green's formula we can rewrite this as
\begin{align}\label{lol3}
\sigma(\d v^+)=\int_S\Delta f_+-\int_\C\dbar f_+\cdot \d L
=\int_S\Delta f_++\int_\C f_+\Delta L.
\end{align}
and taking the conjugate yields
\begin{align}\label{lol4}
\sigma(\dbar v_-)&=\int_S\Delta f_-+\int_\C f_-\Delta L.
\end{align}
So, summing up \eqref{brak} and \eqref{lol5}, and using \eqref{lol3} and \eqref{lol4}, we find
\begin{equation}\label{diffa}\nu_n(f)=\frac 1 2 \int_S\Delta f+\frac 1 2 \int_\C f\Delta L+\bigO(n^{-\beta}).\end{equation}

\subsection{Computation of $\tilde{\nu}_n(f)-\nu_n(f)$} As before, we decompose $f=f_++f_-$ where
$f_+=g/2$ and $f_-=\bar{g}/2$, and we define $v_+$ and $v_-$ by \eqref{vanin}.

Multiplying the first identity in \eqref{band} by $\tilde{D}_n$ and recalling that
$\tilde{\nu}_n(f_+)=-\int\tilde{D}_n\cdot\dbar f_+$, we see that
\begin{equation}
-\tilde{\nu}_n(f_+)=\int_\C[v_+\Delta Q+\dbar v_+\cdot \d(Q-\check{Q})]\cdot \tilde{D}_n.
\end{equation}

Applying Proposition \ref{dentity} (with $v$ replaced by $v_{+}$) and subtracting \eqref{brak}, we now obtain
\begin{equation*}\tilde{\nu}_n(f_+)-\nu_n(f_+)=\sigma(v_+\cdot \d h)+\bigO(n^{-\beta})=
\int_S\dbar f_+\cdot \partial h+\bigO(n^{-\beta})=- \int_S \Delta f_+ \cdot h+\bigO(n^{-\beta}).
\end{equation*}
Adding the complex conjugate relation for $f_-$, we obtain
\begin{equation}\label{diffo}\tilde{\nu}_n(f)-\nu_n(f)=-\int_S\Delta f\cdot h\, dA+\bigO(n^{-\beta}).
\end{equation}

\begin{proof}[Proof of Theorem \ref{gengap}] Consider $f=\re g$ as above.
By Lemma \ref{lem11},
\begin{equation}\label{lemp}
F_{n,f}(t):=\log \mathbb{E}_n [e^{t\fluct_n f}]=\int_0^t \tilde{\mathbb{E}}_{n,sf}(\fluct_n f)\, ds,
\end{equation}
where $\tilde{\mathbb{E}}_{n,sf}$ is expectation with respect to the potential
$\tilde{Q}_s=Q-\frac {sf}n.$ Let us now put
$$e_f=\frac 1 2 \int_S\Delta f+\frac 1 2 \int_\C f\Delta L,\qquad
v_f=-\int_\C f\Delta f\, dA.$$

By \eqref{diffa} and \eqref{diffo} (with $h$ replaced by $sf$), there is a small number $\beta>0$ such that
\begin{equation}\label{inta}
\tilde{\mathbb{E}}_{n,sf}(f)-e_f=sv_f+\bigO(n^{-\beta}).
\end{equation}
(It is easy to see that the $\bigO$-constant is uniform for $s$ in bounded subsets of $\R$.)

Integrating \eqref{inta} in $s$ from $0$ to $t$ and using \eqref{lemp}, we find \eqref{asymp of Fnf}. If we ignore the error-term, we recognize the cumulant generating function of a normal $N(e_f,v_f)$-distributed random variable. Since \eqref{asymp of Fnf} holds uniformly for $t$ in compact subsets of $\mathbb{R}$, we have shown that $\fluct_n f$ converges in distribution to $N(e_f,v_f)$.

It remains to prove that the expectation term $e_f$ agrees with the formula
\ref{expeco} in Theorem \ref{gengap}. For this, we use Green's formula and the definition of Neumann's jump $\calN(L^G)$ in \eqref{wolg}.

In detail, we have
$$\frac 1 2\int_S\Delta f\, dA=\frac 1{8\pi}\int_{\d S}\dn f\, |dz|.$$
Applying twice Green's identity, since $f$ and $L^G$ are harmonic in the gap,
\begin{align*}
0 = \frac 1 2\int_G f\Delta L^{G}\, dA = -\frac 1 {8\pi}\int_{\d G} f\cdot \dn L^{G}\,|dz|+\frac 1 {8\pi}\int_{\d G} \dn f\,\cdot L^G\,|dz|
\end{align*}
and thus, using again Green's identity and also that $L=L^G=\tilde{L}^G$ on $\d G$,
\begin{align*}
\frac 1 2 \int_G f\Delta L\, dA&=-\frac 1 {8\pi}\int_{\d G} f\cdot \dn L\,|dz|+\frac 1 {8\pi}\int_{\d G} \dn f\,\cdot L^G\,|dz|, \\
&=-\frac 1 {8\pi}\int_{\d G} f\cdot \dn L\,|dz|+\frac 1 {8\pi}\int_{\d G} f\,\cdot \dn(\tilde{L}^G)\,|dz|, \\
&=\frac 1 {8\pi}\int_{\d G}f\cdot\calN(L^G)\,|dz|.
\end{align*}
(The normal derivatives are on the outwards direction from $S$ and $\tilde{L}^G$ is harmonic continuation of $L^G|_G$).

We have shown that $e_f$ indeed agrees with the formula \eqref{expeco}. The proof is complete.
\end{proof}

\subsection*{Acknowledgement.} CC acknowledges support from the Swedish Research Council, Grant No. 2021-04626.




\begin{thebibliography}{999}
\bibitem{ADM} Akemann, G., Duits, M., Molag, L., \textit{The Elliptic Ginibre Ensemble: A Unifying Approach to Local and Global Statistics
for Higher Dimensions}, J. Math. Phys. \textbf{64}, 023503 (2023). 
\bibitem{A2} Ameur, Y., \textit{Near-boundary asymptotics of correlation kernels},
J. Geom. Anal. \textbf{23} (2013), 73--95.
\bibitem{ACC} Ameur, Y., Charlier, C., Cronvall, J., \textit{Free energy and fluctuations in the random normal matrix model with spectral gaps}, to appear in Constr. Approx. (cf. arxiv: 2312.13904).
\bibitem{ACC2} Ameur, Y., Charlier, C., Cronvall, J., \textit{Random normal matrices: eigenvalue
correlations near a hard wall}, J. Stat. Phys. \textbf{98}, article no. 98 (2024).
\bibitem{ACCL1} Ameur, Y., Charlier, C., Cronvall, J., Lenells, J.,
\textit{Exponential moments for disk counting statistics at the hard edge of random normal matrices}, J. Spectr. Theory \textbf{13} (2023), 841-902.
\bibitem{ACCL} Ameur, Y., Charlier, C., Cronvall, J., Lenells, J., \textit{Disk counting statistics near hard edges of random normal matrices: the multi-component regime}, Adv. Math. \textbf{441} (2024), Paper No. 109549.
\bibitem{AC:Heine} Ameur, Y., Cronvall J., \textit{On fluctuations of Coulomb systems and universality of the Heine distribution}, arxiv: 2411.10288.
\bibitem{AC} Ameur, Y., Cronvall, J., \textit{Szeg\H{o} type asymptotics for the reproducing kernel in spaces of full-plane weighted polynomials}, Commun. Math. Phys. \textbf{398} (2023), 1291-1348. 
\bibitem{AM} Ameur, Y., Hedenmalm, H., Makarov, N., \textit{Random normal matrices and Ward identities}, Ann. Probab. \textbf{43} (2015), 1157--1201.
\bibitem{AKM} Ameur, Y., Kang, N.-G., Makarov, N., \textit{Rescaling Ward identities in the random normal matrix model}, 
Constr. Approx. \textbf{50} (2019), 63--127.
\bibitem{AKS2} Ameur, Y., Kang, N.-G., Seo, S.-M., \textit{The random normal matrix model: insertion of a point charge}, Potential Anal. \textbf{58} (2023), 331-372.
\bibitem{Ax} Axler, S., \textit{Harmonic functions from a complex analysis viewpoint}, Am. Math. Monthly \textbf{93} (1986), 246--258.
\bibitem{BC} Ball, J.A., Clancey, K.F., \textit{Reproducing kernels for Hardy spaces on multiply connected domains}, Integral Equations Operator Theory \textbf{25} (1996), 35--57.
\bibitem{BMe} Balogh, F., Merzi, D., \textit{Equilibrium Measures for a Class of Potentials
with Discrete Rotational Symmetries}, Constr. Approx. \textbf{42} (2015), 399--424.
\bibitem{BBNY2} Bauerschmidt, R., Bourgade, P., Nikula, M., Yau, H.-T., \textit{The two-dimensional Coulomb plasma: quasi-free approximation and central limit theorem}, Adv. Theor. Math. Phys. \textbf{23}, 841--1002, (2019).
\bibitem{Bell} Bell, S.R., \textit{The Cauchy Transform, Potential Theory and Conformal Mapping},
Chapman \& Hall 2016.
\bibitem{BBS2008} Berman, R., Berndtsson, B., Sj\"{o}strand, J.: \textit{Asymptotics of Bergman kernels}, Ark. Mat. \textbf{46} (2008).

\bibitem{BEG} Bertola, M., Elias Rebelo, J.G., Grava, T., \textit{Painlev\'{e} IV critical asymptotics for orthogonal polynomials in the complex plane}, SIGMA Symmetry Integrability Geom. Methods Appl. \textbf{14} (2014), Paper No. 091.

\bibitem{BCL1} Blackstone, E., Charlier, C., Lenells, J., \textit{Oscillatory asymptotics for Airy kernel determinants on two intervals}, Int. Math. Res. Not. IMRN \textbf{2022} (2022), no. 4, 2636--2687.

\bibitem{BCL2} Blackstone, E., Charlier, C., Lenells, J., \textit{Gap probabilities in the bulk of the Airy process}, Random Matrices Theory Appl. \textbf{11} (2022), no. 2, Paper No. 2250022, 30 pp.

\bibitem{BCL3} Blackstone, E., Charlier, C., Lenells, J., \textit{The Bessel kernel determinant on large intervals and Birkhoff's ergodic theorem}, Comm. Pure Appl. Math. \textbf{76} (2023), no. 11, 3300--3345.


\bibitem{BDE2000} Bonnet, G., David, F., Eynard, B., \textit{Breakdown of universality in multi-cut matrix models}, J. Phys. A \textbf{33} (2000), no. 38, 6739--6768.

\bibitem{BG2} Borot, G., Guionnet, A., \textit{Asymptotic expansion of $\beta$ matrix models in the multi-cut regime}, Forum Math., Sigma \textbf{12} (2024).

\bibitem{BE} Byun, S.-S., Ebke, M. \textit{Universal scaling limits of the symplectic elliptic Ginibre ensemble}, Random Matrices Theory Appl. \textbf{12} (2023), no. 1, Paper No. 2250047, 33 pp.


\bibitem{BF} Byun, S.-S., Forrester, P.J., \textit{Progress on the study of Ginibre ensembles,} KIAS Springer Series in Mathematics \textbf{3} (2025).


\bibitem{BF2022} Byun, S.-S., Forrester, P.J., \textit{Spherical induced ensembles with symplectic symmetry}, SIGMA Symmetry Integrability Geom. Methods Appl. \textbf{19} (2023), paper 033, 28 pages.

\bibitem{BKS} Byun, S.-S., Kang, N.-G., Seo, S.-M., \textit{Partition functions of determinantal and Pfaffian Coulomb gases with radially symmetric potentials}, Commun. Math. Phys. \textbf{401} (2023), 1627--1663.

\bibitem{BLY} Byun, S.-S., Lee, S.-Y., Yang, M., \textit{Lemniscate ensembles with spectral singularity}, arXiv:2107.07221,
2021.

\bibitem{BS2021} Byun, S.-S., Seo, S.-M., \textit{Random normal matrices in the almost-circular regime}, Bernoulli \textbf{29} (2023), 1615--1637.

\bibitem{BY2022} Byun, S.-S., Yang, M., \textit{Determinantal Coulomb gas ensembles with a class of discrete rotational symmetric potentials}, SIAM J. Math. Anal. \textbf{55} (2023), 6867--6897.

\bibitem{CFTW} Can, T., Forrester, P.J., T\'{e}llez, G., Wiegmann, P., \textit{Singular behavior at the edge of Laughlin states.} Phys.
Rev. B \textbf{89}, 235137 (2014).
\bibitem{CSA} Cardoso, G., St\'{e}phan, J.-M., Abanov, A., \textit{The boundary density profile of a Coulomb droplet. Freezing at the edge}, J. Phys. A.: Math. Theor. \textbf{54}(1), (2021), 015002.

\bibitem{C} Charlier, C., \textit{Large gap asymptotics on annuli in the random normal matrix model}, Math. Ann. \textbf{388} (2024), 3529--3587. 

\bibitem{CFWW} Charlier, C., Fahs, B., Webb, C., Wong, M.-D., \textit{Asymptotics of Hankel determinants with a multi-cut regular potential and Fischer-Hartwig singularities}, to appear in Memoirs of the
American Mathematical Society. (arXiv:2111.08395).

\bibitem{ClaeysGravaMcLaughlin} Claeys, T., Grava, T., McLaughlin, K.T.-R., \textit{Asymptotics for the partition function in two-cut random matrix models}, Commun. Math. Phys. \textbf{339} (2015), no. 2, 513--587.

\bibitem{Deift} Deift, P., \textit{Orthogonal polynomials and random matrices: a Riemann-Hilbert approach}, Courant Lecture Notes in Mathematics, 3.

\bibitem{DIZ1997} Deift, P., Its, A., Zhou, X., \textit{A Riemann-Hilbert approach to asymptotic problems arising in the theory of random matrix models, and also in the theory of integrable statistical mechanics}, Ann. of Math. (2) \textbf{146} (1997), no. 1, 149--235.

\bibitem{DKMVZ1999} Deift, P., Kriecherbauer, T., McLaughlin, K.T.-R., Venakides, S., Zhou, X., \textit{Uniform asymptotics for polynomials orthogonal with respect to varying exponential weights and applications to universality questions in random matrix theory}, Comm. Pure Appl. Math. \textbf{52} (1999), no. 11, 1335--1425.

\bibitem{ES} Estienne, B., St\'{e}phan, J.-M., \textit{Entanglement spectroscopy of chiral edge modes in the Quantum Hall effect}, Phys. Rev. B \textbf{101}, 115136 (2020).

\bibitem{FK2020} Fahs, B., Krasovsky, I., \textit{Sine-kernel determinant on two large intervals}, Comm. Pure Appl. Math. 77 (2024), no. 3, 1958--2029.

\bibitem{F} Forrester, P.J., \textit{A review of exact results for fluctuation formulas in random matrix theory.} Probab. Surveys \textbf{20} (2023),  170--225.


\bibitem{Fo} Forrester P.J., \textit{Log-gases and Random Matrices} (LMS-34), Princeton University Press, Princeton 2010.

\bibitem{FH} Forrester, P.J., Honner, G., \textit{Exact statistical properties of the zeros of complex random polynomials}, J. Phys. A. \textbf{41}, 375003 (1999).

\bibitem{FJ} Forrester, P.J., Jancovici, B., \textit{Two-dimensional one-component plasma
in a quadrupolar field}, International Journal of Modern Physics A \textbf{11}, no. 5 (1996).

\bibitem{FM1} Forrester. P.J., Mays, A., \textit{Finite-size corrections in random matrix theory and Odlyzko's dataset for the Riemann zeros}, Proc. A. \textbf{471} (2015), no. 2182, 20150436, 21 pp.

\bibitem{FM2} Forrester. P.J., Mays, A., \textit{Finite size corrections relating to distributions of the length of longest increasing subsequences}, Adv. in Appl. Math. \textbf{145} (2023), paper no. 102482, 33 pp.

\bibitem{GFF} Garoni, T.M., Forrester, P.J., Frankel, N.E., \textit{Asymptotic corrections to the eigenvalue density of the GUE and LUE}, J. Math. Phys. \textbf{46} (2005), no. 10, 103301, 17 pp.

\bibitem{G2006} Grava, T., \textit{Partition function for multi-cut matrix models}, J. Phys. A \textbf{39} (2006), no. 28, 8905--8919.

\bibitem{Ke} Kemp, A., \textit{Characterizations of a discrete normal distribution}, Journal of Statistical Planning and Inference, \textbf{63} (1997), 223--229.

\bibitem{HM} Hedenmalm, H., Makarov, N., \textit{Coulomb gas ensembles and Laplacian growth}, Proc. London. Math. Soc. \textbf{106} (2013), 859--907.

\bibitem{HW22} Hedenmalm, H., Wennman, A., \textit{Berezin density and planar orthogonal polynomials}, Trans. Amer. Math. Soc. \textbf{377} (2024), no. 7, 4825--4863.

\bibitem{HW} Hedenmalm, H., Wennman, A., \textit{Planar orthogonal polynomials and boundary universality in the random normal matrix model}, Acta Math. \textbf{227} (2021), 309-406.

\bibitem{J} Johansson, K., \textit{On fluctuations of eigenvalues of
random Hermitian matrices}, Duke Math. J. \textbf{91} (1998),
151--204.

\bibitem{KM} Kang, N.-G., Makarov, N. G.,
\textit{Gaussian free field and conformal field theory},
Ast\'{e}risque \textbf{353} (2013), viii+136 pp.

\bibitem{KS} Katori, M., Shirai, T. \textit{Zeros of the i.i.d. Gaussian Laurent Series on an Annulus: Weighted Szeg\H{o} Kernels and Permanental-Determinantal Point Processes}, Commun. Math. Phys. 392, 1099--1151 (2022).

\bibitem{K} Kemp, A.W., \textit{Characterizations of a discrete normal distribution}, J. Statist. Plann. Inference \textbf{63}, 223--229.

\bibitem{KM2021} Krasovsky, I., Maroudas, T.-H., \textit{Airy-kernel determinant on two large intervals}, arXiv:2108.04495.

\bibitem{LR} Lee, S.-Y., Riser, R., \textit{Fine asymptotic behaviour of random normal matrices: ellipse case}, J. Math. Phys. \textbf{57} (2016), 023302.

\bibitem{LSe} Lebl\'{e}, T., Serfaty, S., \textit{Fluctuations of two-dimensional Coulomb gases}, Geom. Funct. Anal. \textbf{28} (2018), 443--508.

\bibitem{MR} Marceca, F., Romero, J.-L., \textit{Improved discrepancy for the planar Coulomb gas at
low temperatures}, To appear in J. Anal. Math., arxiv:2212.14821.


\bibitem{Marchal} Marchal, O., \textit{Asymptotic expansion of Toeplitz determinants of an indicator function with discrete rotational symmetry and powers of random
unitary matrices}, Lett. Math. Phys. \textbf{113} (2023), no. 4, Paper No. 78, 16 pp.

\bibitem{ML} McCullough, S., Shen, L.-C., \textit{On the Szeg\H{o} kernel of an annulus}, Proc. Amer. Math. Soc. \textbf{121} (1994), 1111-1121.

\bibitem{M} Mehta, M. L., \textit{Random matrices}, Third Edition, Academic Press 2004.

\bibitem{OLMSE} Oblak, B., Lapierre, B., Moosavi, P., St\'{e}phan, J.-M., Estienne, B., \textit{Anisotropic quantum Hall droplets}, Physical Review X \textbf{14} (2024), paper no 011030.


\bibitem{NIST} Olver, F.W.J., Olde Daalhuis, A.B., Lozier, D.W, Schneider, B.I., Boisvert, R.F., Clark, C.W., Miller, B.R., Saunders, B.V., NIST Digital Library of Mathematical Functions. http://dlmf.nist.gov/, Release 1.1.6 of 2022-06-30.

\bibitem{RV} Rider, B., Vir\'{a}g, B., \textit{The noise in the circular law and the Gaussian free field}, Int. Math. Res. Not. IMRN 2007, no. 2, Art. ID rnm006, 33 pp.

\bibitem{R} Rudin, W., \textit{Principles of mathematical analysis}, McGraw Hill 1976.

\bibitem{Sa} Sakai, M., \textit{Regularity of a boundary having a Schwarz function}, Acta Math. \textbf{166} (1991), 263--297.

\bibitem{ST} Saff, E. B., Totik, V., \textit{Logarithmic potentials with
external fields}, Springer 1997.

\bibitem{Shcherbina} Shcherbina, M., \textit{Fluctuations of Linear Eigenvalue Statistics of $\beta$ Matrix Models in the Multi-cut Regime}, J. Stat. Phys. \textbf{151} (2013), no. 6, 1004--1034.

\bibitem{Seo}
Seo, S.-M., \textit{Edge Behavior of Two-Dimensional Coulomb Gases Near a Hard Wall}, Ann. Henri Poincar\'{e} \textbf{23} (2022), 2247--2275.

\bibitem{S} Simon, B., \textit{Advanced Complex Analysis: A comprehensive course in analysis 2B}, American Mathematical Society 2015.

\bibitem{S2} Simon, B., \textit{Basic Complex Analysis: A comprehensive course in analysis 2A}, American Mathematical Society 2015.

\bibitem{Sz} Szablowski, P.J., \textit{Discrete Normal distribution and its relationship with Jacobi Theta functions} Statistics and Probability Letters, \textbf{52} (2001), 289--299.

\bibitem{T} Temme, N.M., \textit{Special Functions: An Introduction to the Classical Functions of Mathematical Physics}, Wiley 2011. 


\bibitem{Widom1995} Widom, H., \textit{Asymptotics for the Fredholm determinant of the sine kernel on a union of intervals}, Comm. Math. Phys. \textbf{171} (1995), 159--180.

\bibitem{ZW} Zabrodin, A., Wiegmann, P., \textit{Large $N$ expansion for the 2D Dyson gas}, J. Phys. A \textbf{39} (2006), no.28, 8933--8964.
\end{thebibliography}
\end{document}